\documentclass[10pt]{article}

\advance\hoffset by -1in
\advance\textwidth by 2in
\advance\voffset by -0.7in
\advance\textheight by 1.4in

\usepackage{amsmath,amssymb}

\usepackage{changepage}

\usepackage[utf8x]{inputenc}

\usepackage{textcomp,marvosym}

\usepackage{cite}

\usepackage{nameref,hyperref}

\usepackage{microtype}
\DisableLigatures[f]{encoding = *, family = * }

\usepackage[table]{xcolor}

\usepackage{array}

\newcolumntype{+}{!{\vrule width 2pt}}

\newlength\savedwidth

\raggedright
\usepackage[aboveskip=1pt,labelfont=bf,labelsep=period,justification=raggedright,singlelinecheck=off]{caption}

\makeatletter
\renewcommand{\@biblabel}[1]{\quad#1.}
\makeatother

\usepackage{lastpage,fancyhdr,graphicx}
\usepackage{epstopdf}
\fancyheadoffset[L]{2.25in}
\fancyfootoffset[L]{2.25in}
\usepackage{amssymb,amsfonts,amsmath,enumerate,latexsym,yhmath,amsthm}
\usepackage{url}

\usepackage{tikz}
\usepackage{tkz-graph}
\usetikzlibrary{decorations,shapes}
\tikzstyle{tre}=[circle,draw,minimum size=2.5mm, inner sep=3pt]
\tikzstyle{mintre}=[circle,draw,minimum size=1mm, inner sep=0pt]
\newcommand{\etq}[1]{%
\draw (#1) node {\tiny $#1$};
}
 \theoremstyle{plain}
 \newtheorem{theorem}{Theorem}
 \newtheorem{lemma}[theorem]{Lemma}
 \newtheorem{corollary}[theorem]{Corollary}
 \newtheorem{proposition}[theorem]{Proposition}

 \newtheorem{example}[theorem]{Example}
 \newtheorem{definition}[theorem]{Definition}

\renewcommand{\leq}{\leqslant}
\renewcommand{\geq}{\geqslant}

\newcommand{\TT}{\mathcal{T}}

\newcommand{\ZZ}{\mathbb{Z}}
\newcommand{\RR}{\mathbb{R}}
\newcommand{\NN}{\mathbb{N}}
\newcommand{\var}{\mathrm{var}}
\newcommand{\sd}{\mathit{sd}}

\newcommand{\MDM}{\mathrm{MDM}}
\newcommand{\FS}{\mathit{FS}}

\begin{document}
\vspace*{0.2in}

\begin{flushleft}
{\Large
\textbf\newline{Sound Colless-like balance indices for multifurcating  trees} 
}
\newline
\\
Arnau Mir\textsuperscript{13},
Francesc Rossell\'o\textsuperscript{13},
Luc\'\i a Rotger\textsuperscript{23}
\\
\bigskip
\textbf{1} Dept. of Mathematics and Computer Science, University of the Balearic Islands, E-07122 Palma, Spain
\\
\textbf{2} Dept. of Mathematics and Computing, University of La Rioja, E-26004 Logro\~no, Spain
\\
\textbf{3} Balearic Islands Health Research Institute (IdISBa), E-07010 Palma, Spain
\end{flushleft}

\section*{Abstract}
The Colless index is one of the most popular and natural balance indices for bifurcating phylogenetic trees, but it makes no sense for multifurcating trees. In this paper we propose a family of \emph{Colless-like} balance indices $\mathfrak{C}_{D,f}$, which depend on a dissimilarity $D$ and a function $f:\NN\to \RR_{\geq 0}$, 
 that generalize  the Colless index  to multifurcating phylogenetic trees.  We provide two   functions $f$ such that the most balanced phylogenetic trees according to the corresponding indices $\mathfrak{C}_{D,f}$  are exactly the fully symmetric ones. Next, for each one of these two functions $f$ and for three popular dissimilarities  $D$ (the variance, the standard deviation, and the mean deviation from the median),  we determine the range of values of  $\mathfrak{C}_{D,f}$ on the sets of phylogenetic trees with a given number $n$ of leaves. We end the paper by assessing the performance of one of these indices on TreeBASE and using it to show that the trees in this database do not seem to follow either the uniform model for multifurcating trees or the $\alpha$-$\gamma$-model, for any values of $\alpha$ and $\gamma$.


\section*{Introduction}

Since the early 1970s, the shapes of phylogenetic trees have been used to test hypothesis about the evolutive forces underlying their assembly \cite{MH97}. The topological feature of phylogenetic trees most used  in this connection is their symmetry, which captures the symmetry of the evolutionary histories described by the phylogenetic trees. The symmetry of a tree is usually measured by means of its \emph{balance}  \cite[pp. 559--560]{fel:04}, the tendency of the children of any given node to have the same number of descendant leaves. Several \emph{balance  indices} have been proposed so far to quantify the balance of a phylogenetic tree. The two most popular ones are the \emph{Colless index} \cite{Colless:82}, which only works for bifurcating trees, and the \emph{Sackin index} \cite{Sackin:72,Shao:90}, which can be used on multifurcating trees, but there are many others: see for instance \cite{MMRV,cherries,MRR:12,Shao:90} and  \cite[pp. 562--563]{fel:04}. 

The \emph{Colless index} $C(T)$ of a  bifurcating phylogenetic tree $T$ is defined as follows:  if we call the  \emph{balance value} of every internal node $v$ in $T$ the absolute value of the difference between the number of descendant leaves of its pair of children, then $C(T)$ is the sum of the balance values of its internal nodes. In this way, 
the Colless index of a bifurcating tree measures the average balance value of its internal nodes, and therefore it quantifies  in a very intuitive way its balance. In particular, $C(T)=0$ if, and only if, $T$ is a fully symmetric  bifurcating tree with $2^m$ leaves, for some  $m$. 

Unfortunately, the Colless index can only be used as it stands on bifurcating trees. A natural generalization to multifurcating trees would be to define the balance value of a node as some measure of the spread of the numbers of descendant leaves of its children, like the standard deviation, the mean deviation from the median, or any other dissimilarity applied to these numbers, and then to add up all these balance values. But this definition has a drawback: this sum can be 0 on a non-symmetric multifurcating tree, and hence the resulting index does not capture the symmetry of a tree in a sound way. For an example of this misbehavior,  consider the tree depicted in Fig. \ref{fig:nocoll1}: all children of each one of its nodes have the same number of descendant leaves and therefore the balance value of each node in it would be 0, but the tree is not symmetric. Replacing the number of descendant leaves by the number of descendant nodes, which in a bifurcating tree is simply twice the number of descendant leaves minus 1, does not save the day: again, all children of each node in the tree depicted in Fig. \ref{fig:nocoll1} have the same number of descendant nodes.


In this paper we overcome this drawback by taking a suitable  function $f:\NN\to \RR_{\geq 0}$ and then replacing in this schema the number of descendant leaves or the number of descendant nodes of a node by the \emph{$f$-size} of the subtree rooted at the node, defined as the sum of the images under $f$ of the degrees of the nodes in the subtree. Then, we  define the \emph{balance value} (relative to such a function $f$ and a dissimilarity $D$) of an internal node in a phylogenetic tree as the value of $D$  applied to the $f$-sizes of the subtrees rooted at the children of the node. Finally, we define the \emph{Colless-like index} $\mathfrak{C}_{D,f}$ of a phylogenetic tree as the sum of the balance values relative to $f$ and $D$ of its internal nodes. 

The advantage of such a general definition is that there exist functions $f$ such that, for every dissimilarity $D$, the resulting index $\mathfrak{C}_{D,f}$ satisfies that $\mathfrak{C}_{D,f}(T)=0$ if, and only if, $T$ is \emph{fully symmetric}, in the sense that, for every internal node $v$, the subtrees rooted at the children of  $v$ have all the same shape. Two such functions turn out to be $f(n)=\ln(n+e)$ and $f(n)=e^n$.

The different growth pace of these two functions make them to quantify balance in different ways. We show it by finding the trees with largest  $\mathfrak{C}_{D,f}$ value when $f$ is one of these two functions and $D$ is  the variance, the standard deviation, or the mean deviation from the median. We show that the choice of the dissimilarity $D$ does not mean any major difference in the maximally balanced trees relative to $\mathfrak{C}_{D,f}$ for a fixed such $f$, but that changing the function $f$ implies completely different maximally unbalanced trees. 

Finally, we perform some experiments on the TreeBASE phylogenetic database \cite{Treebase}. On the one hand, we compare the behavior of one of our Colless-like indices, the one obtained by taking $f(n)=\ln(n+e)$ and as dissimilarity the mean deviation from the median, $\MDM$, with that of two other balance indices for multifurcating trees: the Sackin index and the total cophenetic index \cite{MRR:12}. On the other hand, we use this Colless-like index to contrast the goodness of fit of the trees in TreeBASE to the uniform distribution and to the $\alpha$-$\gamma$-model for multifurcating trees \cite{Ford2}. 

\section*{Materials}

\subsection*{Notations and conventions}

\noindent Throughout this paper, by a \emph{tree} we always mean a rooted, finite tree  without out-degree 1 nodes. As usual, we understand such a tree as a directed graph, with its arcs pointing away from the root. Given a tree $T$, we shall denote its sets of nodes, of \emph{internal} (that is, non-leaf) nodes, and of arcs by $V(T)$, $V_{int}(T)$, and $E(T)$, respectively, and the out-degree of a node $v\in V(T)$ by $\deg(v)$.
A tree $T$ is \emph{bifurcating} when $\deg(v)=2$ for every $v\in V_{int}(T)$. Whenever we want to emphasize the fact that a tree need not be bifurcating, we shall call it \emph{multifurcating}. The \emph{depth} of a node in a tree $T$ is the length (i.e., number of arcs) of the directed path from the root to it, and the \emph{depth} of $T$ is the largest depth of a leaf in it. 
We shall always make the abuse of language of saying that two isomorphic trees are equal, and hence we shall always identify any  tree with its isomorphism class. We shall denote by $\TT_n^*$ the set of (isomorphism classes  of) trees with $n$ leaves, and by $\TT^*$ the union $\bigcup_{n\geq 1}\TT_n^*$.

A \emph{phylogenetic tree} on a (non-empty, finite) set $X$ of  \emph{labels}  is  a tree with its leaves bijectively labelled in the set $X$.   We shall always identify every leaf in a phylogenetic tree $T$ on $X$ with its label, and in particular we shall denote its set of leaves by $X$.  
Two phylogenetic trees $T_1,T_2$ on $X$ are \emph{isomorphic}  when there exists an isomorphism of directed graphs between them that preserves the labelling of the leaves.  We shall also  make always the abuse of language of considering two isomorphic phylogenetic trees as equal.  Given a set of labels $X$, we shall denote by $\TT_X$  the set of (isomorphism classes of) phylogenetic trees on $X$, and we shall denote by
$\TT_n$, for every $n\geq 1$, the set $\TT_{\{1,2,\ldots,n\}}$. Notice that if $|X|=n$, then any bijection $X\leftrightarrow \{1,2,\ldots,n\}$ induces a bijection  $\TT_X\leftrightarrow\TT_n$. 
Moreover, if $|X|=n$, there is a forgetful mapping $\pi_X:\TT_X\to \TT_n^*$ that sends 
every phylogenetic tree to the corresponding unlabeled tree, which we shall call its \emph{shape}.

No closed formula is known for the numbers $\big|\TT^*_n\big|$ or $\big|\TT_n\big|$. Felsenstein \cite[Ch. 3]{fel:04} gives an easy recurrence to compute $\big|\TT_n\big|$ and describes how to obtain such a recurrence for $\big|\TT_n^*\big|$; an explicit algorithm to compute the latter is provided in \cite{XZL}.  These numbers $\big(\big|\TT_n\big|\big)_n$ and $\big(\big|\TT_n^*\big|\big)_n$ form sequences A000311
and A000669, respectively, in Sloane's \textsl{On-Line Encyclopedia of Integer Sequences}  \cite{Sloane}, where more information about them can be found.

A \emph{comb} is a bifurcating phylogenetic tree with all its internal nodes having a leaf child: see Fig.~\ref{fig:exs1}. We shall generically denote  every comb in $\TT_n$, as well as their shape in $\TT_n^*$, by $K_n$.
A \emph{star} is a phylogenetic tree of depth 1: see Fig.~\ref{fig:exs2}. For consistency with later notations, we shall denote the star in $\TT_n$, and its shape in $\TT_n^*$, by $\FS_n$.

Let $T_1,\ldots,T_k$ be phylogenetic trees on pairwise disjoint sets of labels $X_1,\ldots,X_k$, respectively. The phylogenetic tree
$T_1\star\cdots\star T_k$ on  $X_1\cup \cdots\cup X_k$ is obtained by adding to the disjoint union of $T_1,\ldots,T_k$ a new node $r$ and new arcs from $r$ to the root of each $T_i$. In this way, the trees $T_1,\ldots,T_k$ become the subtrees 
of $T_1\star\cdots\star T_k$ rooted at the children of its root $r$; cf. Fig.  \ref{fig:star}. A similar construction produces a tree $T_1\star\cdots\star T_k$ from a set of (unlabeled) trees $T_1,\ldots,T_k$.


Given a node $v$ in a tree $T$, we shall denote by $T_v$ the subtree of $T$ rooted at $v$ and by $\kappa_v$ its number of descendant leaves, that is, the number of leaves of $T_v$.
An internal node $v$ of a tree $T$ is \emph{symmetric} when, if $v_1,\ldots,v_k$ are its children,  the trees $T_{v_1},\ldots,T_{v_k}$ are isomorphic. A  tree $T$ is \emph{fully symmetric} when all its internal nodes are symmetric, and a phylogenetic tree is \emph{fully symmetric} when its shape is so.

Given a  number $n$ of leaves, there may exist several    fully symmetric  trees with $n$ leaves. For instance, there are 
three fully symmetric  trees with 6 leaves, depicted in Fig. \ref{fig:1}.
As a matter of fact, every fully symmetric  tree with $n$ leaves is characterized by an ordered factorization $n_1\cdots n_k$ of $n$, with  $n_1,\ldots,n_k\geq 2$. More specifically,
for every $k\geq 1$ and $(n_1,\ldots,n_k)\in \NN^k$ with $n_1,\ldots,n_k\geq 2$, let
$\FS_{n_1,\ldots,n_k}$ be the  tree defined, up to isomorphism, recursively as follows:
\begin{itemize}
\item $\FS_{n_1}$ is the star  with $n_1$ leaves.

\item If $k\geq 2$, $\FS_{n_1,\ldots,n_k}$ is a tree whose root has $n_1$ children, and the subtrees  at each one of these children are (isomorphic to) $\FS_{n_2,\ldots,n_k}$.
\end{itemize}
Every $\FS_{n_1,\ldots,n_k}$ is fully symmetric, and every fully symmetric  tree is isomorphic to some $\FS_{n_1,\ldots,n_k}$. 
Therefore, for every $n$, the number of fully symmetric  trees with $n$ leaves is equal to the number $H(n)$ of ordered factorizations of $n$
(sequence A074206 in Sloane's \textsl{On-Line Encyclopedia of Integer Sequences}  \cite{Sloane}).

\subsection*{The Colless index}\label{sec:colless}

The \emph{Colless index} $C(T)$ of a  bifurcating tree $T$ with $n$ leaves  is defined as follows \cite{Colless:82}: if, 
for every $v\in V_{int}(T)$, we denote by $v_1$ and $v_2$  its two children and by
$\kappa_{v_1}$ and $\kappa_{v_2}$ their respective numbers of descendant leaves, then 
$$
C(T)=\sum_{v\in V_{int}(T)} |\kappa_{v_1}-\kappa_{v_2}|.
$$
The Colless index of a phylogenetic tree is simply defined as the Colless index of its shape

It is well-known that the maximum Colless index on the set of bifurcating trees with $n$ leaves is reached at the comb $K_n$, and it is 
$$
C(K_n)=\binom{n-1}{2}
$$ 
(see, for instance, \cite{Rogers93}). As a matter of fact, for every $n$ this maximum is only reached at the comb. Since we have not been able to find 
an explicit reference for this last result in the literature and we shall make use of it later, we provide a proof here.

\begin{lemma}\label{lem:cat-C}
For every bifurcating  tree $T$ with $n$ leaves, if $T\neq K_n$,  then $C(T)<C(K_n)$.
\end{lemma}

\begin{proof}
Let $T$ a bifurcating tree with $n$ leaves different from the comb $K_n$. Let $x$ be an internal node of smallest depth in it without any leaf child, and let $T_1\star T_2$ and $T_3\star T_4$ be the subtrees rooted at its children (see Fig.~\ref{fig:cat1}); for every $i=1,2,3,4$, let $t_i$ be the number of leaves of $T_i$. Assume, without any loss of generality, that $t_1\leq t_2$ and $t_1+t_2\leq t_3+t_4$. Let then $T'$ be the tree obtained by pruning $T_2$ from $T$ and regrafting it to the other arc starting in $x$
(see again Fig.~\ref{fig:cat1}).

It turns out that $C(T')>C(T)$. Indeed, the only nodes whose children change their numbers of descendant leaves from $T$ to $T'$ are (cf.  Fig.~\ref{fig:cat1}): the node $x$; the parent $y$ of the roots of $T_1$ and $T_2$ in $T$, which is removed in $T'$; and the parent $z$ of the root of $T_2$ in $T'$, which does not exist in $T$. Therefore,
$$
\begin{array}{l}
C(T')-C(T)\\
\qquad\quad  =|t_3+t_4-t_2|+[t_3+t_4+t_2-t_1|-|t_2-t_1|-|t_3+t_4-t_2-t_1|\\
\qquad\quad = t_3+t_4-t_2+t_3+t_4+t_2-t_1-t_2+t_1-t_3-t_4+t_2+t_1\\
\qquad\quad =t_1+ t_3+t_4>0.
\end{array}
$$

So, this procedure takes a bifurcating tree  with $n$ leaves $T\neq K_n$ and  produces a new bifurcating tree 
$T'$ with the same number $n$ of leaves and strictly larger Colless index. Since 
the number of bifurcating trees with $n$ leaves is finite, the Colless index cannot increase indefinitely, which means that 
if we iterate this procedure, we must eventually stop at a comb $K_n$. And since the Colless index strictly increases at each iteration, we conclude that if $T\neq K_n$, then $C(T)<C(K_n)$.
\end{proof}

\section*{Methods}

\subsection*{Colless-like indices}

\noindent  Let  $f:\NN\to\RR_{\geq 0}$ be a function that sends each natural number to a positive real number. The \emph{$f$-size} of a tree $T\in \TT^*$ is defined as
$$
\delta_f(T)=\sum_{v\in V(T)} f(\deg(v)).
$$
If $T\in \TT_X$, for some set of labels $X$, then $\delta_f(T)$ is defined as $\delta_f(\pi_X(T))$.

So, $\delta_f(T)$ is the sum of the degrees of all nodes in $T$, with these degrees weighted by means of the function $f$.
Examples of $f$-sizes include:
\begin{itemize}
\item The \emph{number of leaves}, $\kappa$, which is obtained by taking $f(0)=1$ and $f(n)=0$ if $n>0$.
\item The \emph{order} (the number of nodes), $\tau$, which corresponds to $f(n)=1$ for every $n\in \NN$.
\item The usual \emph{size} (the number of arcs), $\theta$, which corresponds to $f(n)=n$ for every $n\in \NN$.
\end{itemize}
Notice that $\delta_f$ satisfies the following recursion:
$$
\delta_f(T_1\star\cdots \star T_k)=\delta_f(T_1)+\cdots+\delta_f(T_k)+f(k).
$$
Table 1 in the Supporting File S2 gives the abstract values of $\delta_f(T)$ for every $T\in \TT^*_n$ with $n=2,3,4,5$.

\begin{example}\label{ex:deltabin}
If $T$ is a bifurcating tree with $n$ leaves, and hence with $n-1$ internal nodes, all of them of out-degree 2, then
$$
\delta_f(T)=(f(0)+f(2))n-f(2).
$$
\end{example}

\begin{example}\label{ex:fullysym}
For every fully symmetric tree $\FS_{n_1,\ldots,n_k}$,
$$
\delta_f(\FS_{n_1,\ldots,n_k}) = n_1\cdots n_k\cdot f(0)+n_1\cdots n_{k-1}\cdot f(n_k)+\cdots+n_1\cdot f(n_2)+f(n_1).
$$
\end{example}

Let now 
$$
\RR^+=\bigcup\limits_{k\geq 1}  \RR^k=\big\{(x_1,\ldots,x_k)\mid k\geq 1, x_1,\ldots,x_k\in \RR\big\}
$$
be the set of all non-empty finite-length sequences of real numbers.
A \emph{dissimilarity} on $\RR^+$ is any mapping $D:\RR^+\to \RR_{\geq 0}$ satisfying the following conditions: for every  $(x_1,\ldots,x_k)\in \RR^+$,
\begin{itemize}
\item $D(x_1,\ldots,x_k)=D(x_{\sigma(1)},\ldots,x_{\sigma(k)})$, for every permutation $\sigma\in \mathcal{S}_k$;
\item $D(x_1,\ldots,x_k)= 0$ if, and only if, $x_1=\cdots=x_k$. 
\end{itemize}
The dissimilarities that we shall explicitly use in this paper are
the \emph{mean deviation from the median},
$$
\MDM(x_1,\ldots,x_k)=\frac{1}{k}\sum_{i=1}^k\big| x_i-\mbox{\it Median}(x_1,\ldots,x_k)|,
$$
 the (\emph{sample}) \emph{variance},
$$
\var(x_1,\ldots,x_k)=\frac{1}{k-1}\sum_{i=1}^k\big(x_i-\mbox{\it Mean}(x_1,\ldots,x_k)\big)^2,
$$
and the  (\emph{sample}) \emph{standard deviation},
$$
\sd(x_1,\ldots,x_k)=+\sqrt{\var(x_1,\ldots,x_k)}.
$$

Let  $D$ be a dissimilarity on $\RR^+$, $f:\NN\to \RR_{\geq 0}$ a function,
and $\delta_f$ the corresponding $f$-size, and let $T\in \TT^*$. For every internal node $v$ in $T$, with children $v_1,\ldots,v_k$,  the \emph{$(D,f)$-balance value} of $v$ is
$$
{bal}_{D,f}(v)=D(\delta_f(T_{v_1}),\ldots,\delta_f(T_{v_k})).
$$
So, ${bal}_{D,f}(v)$ measures, through $D$, the spread of the $f$-sizes of the subtrees rooted at the children of $v$.
In particular, ${bal}_{D,f}(v)=0$ if, and only if, $\delta_f(T_{v_1})=\cdots=\delta_f(T_{v_k})$.

\begin{definition}
Let  $D$ be a dissimilarity on $\RR^+$ and $f:\NN\to \RR_{\geq 0}$ a function. For every $T\in \TT^*$, its \emph{Colless-like  index relative to $D$ and $f$}, $\mathfrak{C}_{D,f}(T)$, is the sum of the $(D,f)$-balance values of the internal nodes of $T$:
$$
\mathfrak{C}_{D,f}(T)=\displaystyle\sum_{v\in V_{int}(T)} {bal}_{D,f}(v).
$$
If $T\in \TT_X$, for some set of labels $X$, then $\mathfrak{C}_{D,f}(T)$ is defined as $\mathfrak{C}_{D,f}(\pi_X(T))$.
\end{definition}

\begin{example}
If we take $D=\MDM$ and $f$ the constant mapping 1, so that $\delta_f=\tau$,  the usual  order  of a tree, then
$$
\begin{array}{rl}
\mathfrak{C}_{\MDM,\tau}(T) & =\displaystyle\sum_{v\in V_{int}(T)} \MDM(\tau_{v_1},\ldots,\tau_{v_{\deg(v)}})\\
& =\displaystyle\sum_{v\in V_{int}(T)} \dfrac{1}{\deg(v)}\sum_{i=1}^{\deg(v)} |\tau_{v_1}-\mbox{\it Median}(\tau_{v_1},\ldots,\tau_{v_{\deg(v)}})|,
\end{array}
$$
where, for every $v\in V_{int}(T)$, $v_1,\ldots,v_{\deg(v)}$ denote its children and  $\tau_{v_1},\ldots,\tau_{v_{\deg(v)}}$ their numbers of descendant nodes.
\end{example}

Notice that $\mathfrak{C}_{D,f}$ gets larger as the $f$-sizes of the subtrees rooted at siblings get more different, and therefore it behaves as a balance index for trees, in the same way as, for instance, the Colless index for bifurcating trees: the smaller the value of $\mathfrak{C}_{D,f}(T)$, the more balanced is $T$ relative to the $f$-size $\delta_f$.

It is clear that $\mathfrak{C}_{D,f}$ satisfies the following recursion:
$$
\mathfrak{C}_{D,f}(T_1\star\cdots \star T_k)=\mathfrak{C}_{D,f}(T_1)+\cdots+\mathfrak{C}_{D,f}(T_k)+D(\delta_f(T_1),\ldots,\delta_f(T_k)).
$$
Therefore these Colless-like  indices are \emph{recursive tree shape statistics} in the sense of \cite{Matsen}, relative to the $f$-size $\delta_f$. Table 1 in the  Supporting File S2 also gives the abstract values of $\mathfrak{C}_{D,f}(T)$, for $D=\MDM$, $\var$, and $\sd$, and for every $T\in \TT^*_n$ with $n=2,3,4,5$.

Next result shows that, if we take $D=\MDM$ or $D=\sd$, then the restriction of any index $\mathfrak{C}_{D,f}$ to bifurcating trees defines, up to a constant factor, the usual Colless index.


\begin{proposition}\label{lem:Phibin}
Let $T$ be a bifurcating tree with $n$ leaves and $f:\NN\to \RR_{\geq 0}$ any function.  
Then, 
$$
\mathfrak{C}_{\MDM,f}(T)=\frac{f(0)+f(2)}{2}\cdot C(T),\qquad
\mathfrak{C}_{\sd,f}(T)=\frac{f(0)+f(2)}{\sqrt{2}}\cdot C(T).
$$
\end{proposition}

\begin{proof}
Notice that, for every $x,y\in \RR$, $\MDM(x,y) =\frac{1}{2}|x-y|$ and $\sd(x,y) =\frac{1}{\sqrt{2}}|x-y|$.
We shall prove the statement for $MDM$; the proof for $\sd$ is identical, replacing the 2 in the denominator by $\sqrt{2}$. 
For every internal node $v$ in a bifurcating tree $T$, if $v_1$ and $v_2$ denote its children,
$$
\begin{array}{l}
 {bal}_{MDM,f}(v)  =\dfrac{1}{2}|{\delta_f}(T_{v_1})-{\delta_f}(T_{v_2})|\\[2ex] 
 \qquad\quad   =
\dfrac{1}{2}|((f(0)+f(2))\kappa_{v_1}-f(2))-((f(0)+f(2))\kappa_{v_2}-f(2))|\\[1ex] 
 \qquad\qquad   \mbox{(by Example \ref{ex:deltabin})}\\[1ex]
\qquad\quad    =
\dfrac{f(0)+f(2)}{2}\cdot |\kappa_{v_1}-\kappa_{v_2}|
\end{array}
$$
and therefore
$$
\begin{array}{rl}
\mathfrak{C}_{MDM,f}(T)  & \displaystyle = \sum\limits_{v\in V_{int}(T)}  {bal}_{MDM,f}(v)
=\frac{f(0)+f(2)}{2} \cdot\hspace*{-2ex} \sum\limits_{v\in V_{int}(T)}  |\kappa_{v_1}-\kappa_{v_2}|\\[3ex] & \displaystyle =\frac{f(0)+f(2)}{2}\cdot C(T),
\end{array}
$$
as we claimed. 
\end{proof}

If we define the \emph{quadratic Colless index} of a bifurcating tree $T$ as
$$
C^{(2)}(T)=\sum_{v\in V_{int}(T)} (\kappa_{v_1}-\kappa_{v_2})^2,
$$
where, for every $v\in V_{int}(T)$, $v_1,v_2$ denote its children,
then, using that $\var(x,y)=\frac{1}{2}(x-y)^2$, a similar argument proves the following result.

\begin{proposition}\label{prop:C2}
Let $T$ be a bifurcating tree with $n$ leaves and $f:\NN\to \RR_{\geq 0}$ any function.  
Then, 
$$
\mathfrak{C}_{\var,f}(T)=\frac{(f(0)+f(2))^2}{2}\cdot C^{(2)}(T).
$$
\ \hspace*{\fill}\qed
\end{proposition}

As far as the cost of computing  Colless-like indices goes, we have the following result.

\begin{proposition}\label{cost}
If the cost of computing $D(x_1,\ldots,x_k)$ is  in $O(k)$ and the cost of computing 
each $f(k)$ is  at most in $O(k)$, then,
for every $T\in \TT_n^*$, the cost of computing $\mathfrak{C}_{D,f}(T)$ is in $O(n)$.
\end{proposition}

\begin{proof}
Assume that  every $f(k)$ is computed in time at most $O(k)$. For every $k\geq 2$, let $m_k$ the number of internal nodes in $T$ of out-degree $k$.
Since the sizes $\delta_f(v)$  are additive, in the sense that if $v$ has children $v_1,\ldots,v_k$, then $\delta_f(v)=\sum_{i=1}^k \delta_f(v_i)+f(k)$, we can compute the whole vector $\big(\delta_f(v)\big)_{v\in V(T)}$  in  time $O(n+\sum_{k\geq 2} m_k\cdot k)=O(n)$ by traversing the tree in post-order.
 
 Assume now that $D(x_1,\ldots,x_k)$ can be computed in time $O(k)$.
 Then,  for every internal node $v$ of out-degree $k$, ${bal}_{D,f}(v)=D(\delta_f(T_{v_1}),\ldots,\delta_f(T_{v_k}))$ can be computed in time $O(k)$, by simply reading the $k$ sizes of its children (which are already computed) and applying $D$ to them. This shows that the whole vector $\big({bal}_{D,f}(v)\big)_{v\in V(T)}$ can be computed again in time
 $O(\sum_{k\geq 2} m_k\cdot k)=O(n)$. Finally, 
 we compute $\mathfrak{C}_{D,f}(T)$ by adding the entries of $\big({bal}_{D,f}(v)\big)_{v\in V(T)}$, which still can be done in time  
 $O(n)$.
 \end{proof}

The dissimilarities mentioned previously in this subsection can be computed in a number of sums and multiplications that is linear in the length of the input vector, and the specific functions $f$ that we shall consider in the next subsection, basically exponentials and logarithms, can be approximated to any desired precision in constant time by using addition and look-up tables \cite{Yan}.

\subsection*{Sound Colless-like  indices}

\noindent It is clear that, for every dissimilarity $D$, for every function $f:\NN\to \RR_{\geq 0}$ and for every fully symmetric tree
$\FS_{n_1,\ldots,n_k}$, $\mathfrak{C}_{D,f}(\FS_{n_1,\ldots,n_k})=0$, because ${bal}_{D,f}(v)=0$ for every $v\in V_{int}(\FS_{n_1,\ldots,n_k})$. We shall say that a Colless-like  index $\mathfrak{C}_{D,f}$ is \emph{sound} when the converse implication is true.

\begin{definition}
A Colless-like  index $\mathfrak{C}_{D,f}$ is  \emph{sound} when, for every $T\in \TT^*$, $\mathfrak{C}_{D,f}(T)=0$ if, and only if, $T$ is fully symmetric. 
\end{definition}

In other words, $\mathfrak{C}_{D,f}$ is sound when, according to it, the most balanced trees  are exactly the fully symmetric trees.

The  Colless index $C$ and its quadratic version $C^{(2)}$ are  sound  for \emph{bifurcating}  trees. Unfortunately, neither their direct generalizations  $\mathfrak{C}_{\MDM,\kappa}$, $\mathfrak{C}_{\sd,\kappa}$ and $\mathfrak{C}_{\var,\kappa}$ ---where $\kappa$ denotes the number of leaves--- nor 
 $\mathfrak{C}_{\MDM,\tau}$, $\mathfrak{C}_{\sd,\tau}$ and $\mathfrak{C}_{\var,\tau}$ ---where $\tau$ denotes the number of nodes--- or even replacing $\tau$ by $\theta$, the usual size, which is simply $\tau-1$,
are  sound for multifurcating trees. For example,
the tree $T$ in Fig. \ref{fig:nocoll1} is not fully symmetric, but 
$\mathfrak{C}_{\MDM,\kappa}(T)=\mathfrak{C}_{\var,\kappa}(T)=\mathfrak{C}_{\sd,\kappa}(T) =\allowbreak\mathfrak{C}_{\MDM,\tau}(T)=\mathfrak{C}_{\var,\tau}(T)=\mathfrak{C}_{\sd,\tau}(T)=0$. 

%

As a matter of fact, the soundness of $\mathfrak{C}_{D,f}(T)=0$ does not depend on $D$, but only on $f$, as the following lemma shows.

\begin{lemma}\label{lem:sound1}
$\mathfrak{C}_{D,f}$ is sound if, and only if, $\delta_f(T_1)\neq \delta_f(T_2)$ for every pair of  different fully symmetric trees $T_1,T_2$.
\end{lemma}

\begin{proof}
As far as the ``only if'' implication goes, if there exist two different (i.e., non isomorphic)  fully symmetric trees $T_1,T_2$ such that
$\delta_f(T_1)= \delta_f(T_2)$, then the tree $T=T_1\star T_2$ is not fully symmetric, but
$$
\mathfrak{C}_{D,f}(T)=\mathfrak{C}_{D,f}(T_1)+\mathfrak{C}_{D,f}(T_2)+D(\delta_f(T_1), \delta_f(T_2))=0.
$$

Conversely, assume that,  for every pair of fully symmetric trees $T_1,T_2$, if 
$\delta_f(T_1)= \delta_f(T_2)$ then $T_1= T_2$. We shall prove by complete induction on $n$  that
if $T$ is a tree with $n$ leaves such that
$\mathfrak{C}_{D,f}(T)=0$, then $T$ is fully symmetric.  If $T$ has only one leaf, it is clearly fully symmetric. Now, assume that $n>1$ and hence that $T$ has depth at least 1. Let $T_1,\ldots,T_k$, $k\geq 2$, be its subtrees rooted at the children of its root,
so that $T=T_1\star\cdots\star T_k$.
 Then,
$$
0=\mathfrak{C}_{D,f}(T)=\sum_{i=1}^k\mathfrak{C}_{D,f}(T_i)+D(\delta_f(T_1),\ldots, \delta_f(T_k))
$$
implies, on the one hand, that $\mathfrak{C}_{D,f}(T_1)=\cdots=\mathfrak{C}_{D,f}(T_k)=0$, and hence, by induction, that $T_1,\ldots,T_k$ are fully symmetric, and, on the other hand, that $D(\delta_f(T_1),\ldots, \delta_f(T_k))=0$, and hence that $\delta_f(T_1)=\cdots= \delta_f(T_k)$, which, by assumption, implies that $T_1=\cdots=T_k$: in summary,
 $T$ is fully symmetric. 
\end{proof}

The following problem now arises:
\medskip

\noindent\textbf{Problem.} \emph{To find functions $f:\NN\to \RR_{\geq 0}$ such that  $\mathfrak{C}_{D,f}$ is  sound.}
\medskip

Unfortunately, many natural functions $f$ do not define sound Colless-like  indices, as the following examples show.

\begin{example}\label{ex:quad}
If $f(n)=an^2+bn+c$, for any $a,b,c$, then $\mathfrak{C}_{D,f}$ is not  sound, because, for example, $\delta_f(\FS_{2,2,2,7})=\delta_f(\FS_{14,4})=420a+70b+71c$.
\end{example}

\begin{example}\label{ex:monom}
 If $f(n)=n^d$, for any $d\geq 0$, then $\mathfrak{C}_{D,f}$ is not  sound.
Indeed, for every $d\geq 3$ (the case when $d\leq 2$ is a particular case of the last example), take
\begin{itemize}
\item $k=2^d+1$ and $l=2$;
\item $n_i=2^{(d-1)^id^{k-i-1}}$ for $i=1,\ldots,k-1$;
\item $n_k=2$;
\item $m_1=2^{(d-1)d^{k-2}+1}$;
\item $m_2=2^{((d-1)^2(d^{k-2}-(d-1)^{k-2})+d-1)/d}$; notice that this exponent is an integer number, because $k$ is odd and therefore $d$ divides $(d-1)^k+1$.
\end{itemize}
Then
$$
n_1\ldots n_{i-1}\cdot n_i^{d}=n_1^d
$$
and hence, on the one hand,
$$
\begin{array}{rl}
n_1^d+\cdots +n_1\cdots n_{k-2}\cdot n_{k-1}^d & =(k-1)n_1^d=2^d\cdot 2^{(d-1)d^{k-1}}\\
& =\Big(2^{1+(d-1)d^{k-2}}\Big)^d=m_1^d,
\end{array}
$$
and, on the other hand,
$$
\begin{array}{rl}
n_1\cdots n_{k-1}\cdot n_{k}^d  & \displaystyle = n_1^{\frac{(1-(\frac{d-1}{d})^{k-1})}{{(1-\frac{d-1}{d})}}} \hspace*{-1ex}\cdot n_k^d=
n_1^{\frac{d^{k-1}-(d-1)^{k-1}}{d^{k-2}}}n_k^d\\[2ex]
& =2^{(d-1)(d^{k-1}-(d-1)^{k-1})+d}=m_1m_2^d.
\end{array}
$$
Therefore, $\delta_{n^d}(\FS_{n_1,\ldots,n_k})=\delta_{n^d}(\FS_{m_1,m_2})$. 

Of course, for any given $d$ there may exist ``smaller'' counterexamples: for instance, $\delta_{n^3}(\FS_{2,10,4})=\delta_{n^3}(\FS_{6,8})=3288$ and $\delta_{n^4}(\FS_{2,6,2,3 })=\delta_{n^4}(\FS_{8,3})=4744$.
\end{example}

\begin{example}
 If $f(n)=\log_a(n)$ (for some $a>1$) when $n>0$, and $f(0)=0$, then $\mathfrak{C}_{D,f}$ is not  sound: for instance,
$\delta_f(\FS_{2,2})=\delta_f(\FS_{8})=\log_a(8)$. In a similar way,  if $f(n)=\log_a(n+1)$ (for some $a>1$), then $\mathfrak{C}_{D,f}$ is not sound, either: for instance,
$\delta_f(\FS_{2,3,3})=\delta_f(\FS_{5,7})=\log_a(196608)$.
\end{example}

On the positive side, we shall show now two functions that define sound indices.
The following lemmas will be useful to prove it.

\begin{lemma}\label{nk-neq-ml}
For every $k,l\geq 1$ and $n_1,n_2,\ldots,n_k,m_1,m_2,\ldots,m_l\geq 2$,
if $\delta_f(\FS_{n_1,n_2,\ldots,n_k})=\delta_f(\FS_{m_1,m_2,\ldots,m_l})$,
$n_1\cdot n_2\cdots n_k=m_1\cdot m_2\cdots m_l$, and $n_k=m_l$, then
$\delta_f(\FS_{n_1,\ldots,n_{k-1}})=\delta_f(\FS_{m_1,\ldots,m_{l-1}})$.
\end{lemma}

\begin{proof}
If $n_1 \cdots n_k=m_1\cdots m_l$ and $n_k=m_l$, then
$n_1\cdots n_{k-1}=m_1\cdots m_{l-1}$. Thus, if, moreover, $\delta_f(\FS_{n_1,n_2,\ldots,n_k})=\delta_f(\FS_{m_1,m_2,\ldots,m_l})$, that is, 
$$
\begin{array}{l}
n_1\cdots n_kf(0)+n_1\cdots n_{k-1} f(n_k)+n_1\cdots n_{k-2} f(n_{k-1})+\cdots+f(n_1)\\
\quad=m_1\cdots m_lf(0)\! +m_1\cdots m_{l-1} f(m_l)\! +m_1\cdots m_{l-2} f(m_{l-1})\! +\cdots+\! f(m_1),
\end{array}
$$
then
$$
\begin{array}{l}
n_1\cdots n_{k-2} f(n_{k-1})+\cdots+n_1f(n_2)+f(n_1)\\
\qquad\hspace*{\fill}=m_1\cdots m_{l-2} f(m_{l-1})+\cdots+m_1f(m_2)+f(m_1)
\end{array}
$$
and hence 
$$
\begin{array}{l}
\delta_f(\FS_{n_1,n_2,\ldots,n_{k-1}})\\
\quad=n_1\cdots n_{k-1} f(0)+n_1\cdots n_{k-2} f(n_{k-1})+\cdots+n_1f(n_2)+f(n_1)\\
\quad=m_1\cdots m_{l-1} f(0)+m_1\cdots m_{l-2} f(m_{l-1})+\cdots+m_1f(m_2)+f(m_1)\\
\quad=\delta_f(\FS_{m_1,\ldots,m_{l-1}})\end{array}
$$
as we claimed.
\end{proof}

\begin{lemma}\label{lem:aux}
If $n_1,\ldots,n_k\geq 2$, then
$$
1+n_1+n_1n_2+\cdots+n_1\cdots n_{k-1}< n_1\cdots n_k.
$$
\end{lemma}

\begin{proof}
By induction on $k$. If $k=1$, the statement says that $1<n_1$, which is true by assumption.
 Assume now that the statement is true for any $n_1,\ldots,n_{k}\geq 2$, and let $n_{k+1}\geq 2$. Then,
$$
\begin{array}{r}
1+n_1+n_1n_2+\cdots+n_1\cdots n_{k-1}+n_1\cdots n_{k}  <
n_1\cdots n_{k}+ n_1\cdots n_k\qquad\\
=2n_1\cdots n_k\leq
n_1\cdots n_k\cdot n_{k+1}.
\end{array}
$$

\end{proof}

\begin{proposition}
If $f(n)=e^n$, then $\mathfrak{C}_{D,f}$ is sound.
\end{proposition}

\begin{proof}
Assume that there exist two  non-isomorphic fully symmetric trees $\FS_{n_1,\ldots,n_k}$ and $\FS_{m_1,\ldots,m_l}$ such that
$$\delta_{e^n}(\FS_{n_1,\ldots,n_k})=\delta_{e^n}(\FS_{m_1,\ldots,m_l}),$$
 that is, such that
\begin{equation}
\begin{array}{l}
n_1\cdots n_k +n_1\cdots n_{k-1}e^{n_k}+\cdots+n_1e^{n_2}+e^{n_1}\\
\qquad =
m_1\cdots m_l+m_1\cdots m_{l-1}e^{m_l}+\cdots+m_1e^{m_2}+e^{m_1}.
\end{array}
\label{eq:exp}
\end{equation}
Assume that $l$ is the smallest depth of a fully symmetric tree with $e^n$-size equal to the $e^n$-size of another fully symmetric tree  non-isomorphic to it.

Since $e$ is transcendental, equality (\ref{eq:exp}) implies the equality of polynomials in $\ZZ[x]$
$$
\begin{array}{l}
n_1\cdots n_k +n_1\cdots n_{k-1}x^{n_k}+\cdots+n_1x^{n_2}+x^{n_1}\\
\qquad\qquad =
m_1\cdots m_l+m_1\cdots m_{l-1}x^{m_l}+\cdots+m_1x^{m_2}+x^{m_1}.
\end{array}
$$
If $l=1$, the right-hand side polynomial is simply $m_1+x^{m_1}$ and then the equality of polynomials implies that $k=1$ and $n_1=m_1$, which contradicts the assumption that $\FS_{n_1,\ldots,n_k}\neq \FS_{m_1,\ldots,m_l}$. Now assume that $l\geq 2$. 
This equality of polynomials implies the equality of their independent terms:
$n_1\cdots n_k=m_1\cdots m_l$. On the other hand, the non-zeroth power of $x$ with the largest coefficient in the left-hand side polynomial  is $x^{n_k}$ (because all coefficients are non-negative, and,  by Lemma \ref{lem:aux}, 
$n_1\cdots n_{k-1}$ alone is larger than the sum
$n_1\cdots n_{k-2}+\cdots+n_1+1$
of all other coefficients of non-zeroth powers of $x$) and, by the same reason,
the non-zeroth power of $x$ with the largest coefficient in the right-hand side polynomial is $x^{m_l}$.
The equality of polynomials implies then that $n_k=m_l$ and hence, by Lemma \ref{nk-neq-ml}, that 
$\delta_{e^n}(\FS_{n_1,\ldots,n_{k-1}})=\allowbreak \delta_{e^n}(\FS_{m_1,\ldots,m_{l-1}})$, against the assumption on $l$. We reach thus a contradiction that implies that there does not exist any pair of  non-isomorphic fully symmetric trees with the same $e^n$-size. By Lemma \ref{lem:sound1}, this implies that $\mathfrak{C}_{D,e^n}$ is sound.
\end{proof}

The same argument shows that $\mathfrak{C}_{D,f}$ is sound for every exponential function $f(n)=r^n$ with base  $r$ a transcendental real number. However, if $r$ is not transcendental, then $\mathfrak{C}_{D,r^n}$ need not be sound. For instance, $\delta_{2^n}(FS_{2,3})=\delta_{2^n}(FS_{3,2})=26$ and $\delta_{\sqrt{2}^n}(FS_{8,10})=\delta_{\sqrt{2}^n}(FS_{12,8})=352$.

\begin{proposition}
If $f(n)=\ln(n+e)$, then $\mathfrak{C}_{D,f}$ is sound.
\end{proposition}

\begin{proof}
The argument is similar to that of the previous proof.
Let $f(n)=\ln(n+e)$ and assume that there exist two  non-isomorphic fully symmetric trees $\FS_{n_1,\ldots,n_k}$ and $\FS_{m_1,\ldots,m_l}$ such that
$\delta_f(\FS_{n_1,\ldots,n_k})=\delta_f(\FS_{m_1,\ldots,m_l})$, that is, such that
\begin{equation}
\begin{array}{l}
n_1\cdots n_k+n_1\cdots n_{k-1}\ln(n_k+e)+\cdots+\ln(n_1+e)\\ 
\qquad =
m_1\cdots m_l+m_1\cdots m_{l-1}\ln(m_l+e)+\cdots+\ln(m_1+e).
\end{array}
\label{eq1}
\end{equation}
Assume that $l$ is the smallest depth of a fully symmetric tree with $f$-size equal to the $f$-size of a fully symmetric tree  non-isomorphic to it.

Applying the exponential function to both sides of equality (\ref{eq1}), we obtain
$$
\begin{array}{l}
e^{n_1\cdots n_k}(n_k+e)^{n_1\cdots n_{k-1}}\cdots (n_2+e)^{n_1}(n_1+e)\\
\qquad\qquad=
e^{m_1\cdots m_l}(m_l+e)^{m_1\cdots m_{l-1}}\cdots (m_2+e)^{m_1}(m_1+e).
\end{array}
$$
Since $e$ is transcendental, this implies the equality of polynomials in $\ZZ[x]$
$$
\begin{array}{l}
x^{n_1\cdots n_k}(n_k+x)^{n_1\cdots n_{k-1}}\cdots (n_2+x)^{n_1}(n_1+x)\\
\qquad\qquad=
x^{m_1\cdots m_l}(m_l+x)^{m_1\cdots m_{l-1}}\cdots (m_2+x)^{m_1}(m_1+x),
\end{array}
$$
which, since $n_1,\ldots,n_k,m_1,\ldots,m_l\geq 2$, on its turn implies the equalities
$$
\begin{array}{l}
x^{n_1\cdots n_k}=x^{m_1\cdots m_l}, \mbox{ i.e., } n_1\cdots n_k=m_1\cdots m_l,\\[1ex]
(x+n_k)^{n_1\cdots n_{k-1}}\cdots (x+n_2)^{n_1}(x+n_1)\\\qquad\qquad =
(x+m_l)^{m_1\cdots m_{l-1}}\cdots (x+m_2)^{m_1}(x+m_1).
\end{array}
$$
If $l=1$, the right-hand side polynomial in the second equality is simply $x+m_1$ and then this equality of polynomials implies that $k=1$ and $n_1=m_1$, which contradicts the assumption that $\FS_{n_1,\ldots,n_k}\neq \FS_{m_1,\ldots,m_l}$.  Now assume that $l\geq 2$. 
From the first equality we know that $n_1\cdots n_k=m_1\cdots m_l$.
Now, the root of the left-hand side polynomial in the second equality with largest multiplicity is $-n_k$
(because, by Lemma \ref{lem:aux}, $n_1\cdots n_{k-1}$ alone is greater than the degree of
$(x+n_{k-1})^{n_1\cdots n_{k-2}}\cdots (x+n_2)^{n_1}(x+n_1)$) and,  similarly,
the root of the right-hand side polynomial in the second equality with largest multiplicity is $-m_l$.
Then, the equality of both polynomials implies that $n_k=m_l$ and hence, by Lemma \ref{nk-neq-ml},
$\delta_f(\FS_{n_1,\ldots,n_{k-1}})=\delta_f(\FS_{m_1,\ldots,m_{l-1}})$, against the assumption on $l$.
As in the previous proof, this contradiction implies that $\mathfrak{C}_{D,f}$ is sound
\end{proof}

The same argument proves that, for every transcendental number $r>1$, the function $f(n)=\log_{r}(n+r)$ defines sound  indices $\mathfrak{C}_{D,f}$. However, if $r$ is not transcendental, then such a $\mathfrak{C}_{D,f}$ need not be sound. For instance, $\delta_{\log_2(n+2)}(FS_{9,6})=\delta_{\log_2(n+2)}(FS_{20,2})=81+\log_2(11)$.\medskip

In summary, each one of the functions $f(n)=\ln(n+e)$ and $f(n)=e^n$ defines, for every dissimilarity $D$,  a Colless-like index $\mathfrak{C}_{D,f}$ that reaches its minimum value on each $\TT_n^*$, 0, at exactly the fully symmetric trees.

\section*{Results}

\subsection*{Maximally unbalanced trees}

\noindent The next results give the maximum values
of $\mathfrak{C}_{D,f}$ on $\TT_n^*$ when $D=\MDM$, $\var$ or $\sd$ and $f(n)=\ln(n+e)$ or $f(n)=e^n$.
These maxima define the range of each $\mathfrak{C}_{D,f}$ on $\TT_n^*$, and then, dividing by them, we can define normalized Colless-like indices  that can be used to compare the balance of trees with different numbers of leaves.

We begin with the function $f(n)=\ln(n+e)$, which is covered by the following theorem.

\begin{theorem}\label{thm:mdm}
Let $f$ be a function $\NN\to \RR_{\geq 0}$ such that $0<f(k)< f(k-1)+f(2)$, for every $k\geq 3$. Then, for every $n\geq 2$, the indices $\mathfrak{C}_{\MDM,f}$,
$\mathfrak{C}_{\sd,f}$ and $\mathfrak{C}_{\var,f}$ reach  their maximum values on  $\TT^*_n$ exactly at  the comb $K_n$.
These maximum values are, respectively,
$$
\begin{array}{l}
\mathfrak{C}_{\MDM,\delta_{f}}(K_n)=\dfrac{f(0)+f(2)}{4}(n-1)(n-2),\\[2ex]
\mathfrak{C}_{\sd,\delta_{f}}(K_n)=\dfrac{f(0)+f(2)}{2\sqrt{2}}(n-1)(n-2),\\[2ex]
\mathfrak{C}_{\var,\delta_{f}}(K_n)=
\dfrac{(f(0)+f(2))^2}{12}(n-1)(n-2)(2n-3).
\end{array}
$$
\hspace*{\fill}\qed
\end{theorem}

The proof of this theorem is very long, and we devote to it the first three sections of the  Supporting File S1, one section for each dissimilarity.

It is straightforward to check that the function  $f(n)=\ln(n+e)$ satisfies the hypothesis of Theorem \ref{thm:mdm} (as to the inequality $f(k)\leq f(k-1)+f(2)$, notice that
$\ln(k+e)\leq \ln(k+e-1)+\ln(2)$ if, and only if,  $k+e\leq 2(k+e-1)$, and this last inequality holds
 (strictly) for every $k\in \NN$).
Therefore,
$\mathfrak{C}_{\MDM,{\ln(n+e)}}$,
$\mathfrak{C}_{\var,{\ln(n+e)}}$, and
$\mathfrak{C}_{\sd,{\ln(n+e)}}$
take their maximum values  on $\TT^*_n$  at the comb    $K_n$. In other words, the combs are the most unbalanced trees according to these indices.
Table 2 in the  Supporting File S2 gives the values of $\mathfrak{C}_{\MDM,{\ln(n+e)}}$,
$\mathfrak{C}_{\var,{\ln(n+e)}}$, and
$\mathfrak{C}_{\sd,{\ln(n+e)}}$  on $\TT^*_n$, for $n=2,3,4,5$, and the positions of the different trees in each $\TT_n^*$ according to the increasing order of the corresponding index.
\smallskip

As far as $f(n)=e^n$ goes, we have the following result.
We have also moved its proof to the  Supporting File S1.

\begin{theorem}\label{thm:mdm-exp}
For every $n\geq 2$:
\begin{enumerate}[(a)]
\item If  $n\neq 4$, then both 
$\mathfrak{C}_{\MDM,{e^n}}$ and $\mathfrak{C}_{\sd,{e^n}}$ 
reach their maximum on $\TT^*_n$ exactly at  the tree $\FS_1\star \FS_{n-1}$ (see Fig. \ref{fig:maxexp1}), and these maximum values are
$$
\begin{array}{l}
\displaystyle \mathfrak{C}_{\MDM,{e^n}}(\FS_1\star \FS_{n-1})= \frac{1}{2}(e^{n-1}+n-2),\\[2ex]
\displaystyle \mathfrak{C}_{\sd,{e^n}}(\FS_1\star \FS_{n-1})= \frac{1}{\sqrt{2}}(e^{n-1}+n-2).
\end{array}
$$

\item Both
$\mathfrak{C}_{\MDM,{e^n}}$ and $\mathfrak{C}_{\sd,{e^n}}$ reach their maximum on $\TT^*_4$ exactly at the comb $K_4$, and these maximum values are
$$
\begin{array}{l}
\displaystyle\mathfrak{C}_{\MDM,{e^n}}(K_4)=\frac{3}{2}(e^2+1),\\[2ex]
\displaystyle\mathfrak{C}_{\sd,{e^n}}(K_4)=\frac{3}{\sqrt{2}}(e^2+1).
\end{array}
$$

\item $\mathfrak{C}_{\var,{e^n}}$ always reaches its maximum on $\TT_n^*$ exactly at the tree $\FS_1\star \FS_{n-1}$, and it is
$$
\mathfrak{C}_{\var,{e^n}}(\FS_1\star \FS_{n-1})=\frac{1}{2}(e^{n-1}+n-2)^2.
$$
\end{enumerate}
\ \hspace*{\fill}\qed
\end{theorem}


So, according to $\mathfrak{C}_{\MDM,{e^n}}$,
$\mathfrak{C}_{\var,{e^n}}$, and
$\mathfrak{C}_{\sd,{e^n}}$,  the trees of the form $\FS_1\star \FS_{n-1}$ are the most unbalanced (except for $n=4$ and $D=\MDM$ or $\sd$, in which case the most unbalanced tree is the comb).
Table 2 in the  Supporting File S2 also gives the values of these indices on $\TT^*_n$, for $n=2,3,4,5$, and the positions of the different trees in each $\TT_n^*$ according to the increasing order of the corresponding index.

\subsection*{The R package ``CollessLike''}

We have written an \texttt{R} package called \textsl{CollessLike}, which is available at the CRAN (\url{https://cran.r-project.org/web/packages/CollessLike/index.html}), that computes the Colless-like indices and their normalized version, as well as several other balance indices, and simulates the distribution of these indices of $\TT_n$ under the $\alpha$-$\gamma$-model \cite{Ford2}. 
Among others, this package contains functions that:
\begin{itemize}
\item Compute the following balance indices for multifurcating trees:  the Sackin index $S$ \cite{Sackin:72,Shao:90}, the total cophenetic index $\Phi$  \cite{MRR:12}, and the Colles-like index $\mathfrak{C}_{D,f}$ for several predefined dissimilarities $D$ and functions $f$ as well as for any user-defined ones.  

Our function also computes the normalized versions (obtained by subtracting their minimum value and dividing by their range, so that they take values in $[0,1]$) of $S$, $\Phi$ and the Colless-like indices $\mathfrak{C}_{D,f}$ for which we have computed the range in Theorems  \ref{thm:mdm} and \ref{thm:mdm-exp}.
Recall from the aforementioned references that, for every $n\geq 2$:
\begin{itemize}
\item the range of $S$ on $\TT_n^*$ goes from $S(\FS_n)=n$ to $S(K_n)=\frac{1}{2}(n+2)(n-1)$
\item the range of $\Phi$ on $\TT_n^*$ goes from $\Phi(\FS_n)=0$ to $\Phi(K_n)=\binom{n}{3}$
\end{itemize}
Therefore, for every $T\in \TT_n$, the normalized Sackin and total cophenetic index are, respectively,
$$
\displaystyle S_{norm}(T)=\frac{S(T)-n}{\frac{1}{2}(n+2)(n-1)-n},\quad 
\displaystyle \Phi_{norm}(T)=\frac{\Phi(T)}{\binom{n}{3}},
$$
while, for instance, the normalized version of $\mathfrak{C}_{\MDM,\ln(n+e)}$ is
$$
\mathfrak{C}_{\MDM,\ln(n+e),norm}(T)=\frac{\mathfrak{C}(T)}{\frac{1+\ln(e+2)}{4}(n-1)(n-2)}.
$$

\item Given an $n\geq 2$, produce a sample of $N$ values of a balance index $S$, $\Phi$, or $\mathfrak{C}_{D,f}$ on trees in  $\TT_n$ generated following an $\alpha$-$\gamma$-model:  the parameters $N$, $n$, $\alpha$, $\gamma$ (with $0\leq\gamma\leq\alpha\leq 1$)  can be set by the user.

Due to the computational cost of this function, we have stored the values of $S$, $\Phi$, and $\mathfrak{C}_{\MDM,\ln(n+e)}$ (denoted henceforth simply by $\mathfrak{C}$) on the samples of $N=5000$ trees in each $\TT_n$ (for every $n=3,\ldots,50$  and for every  $\alpha,\gamma\in \{0,0.1,0.2,\ldots,0.9,1\}$ with $\gamma\leq \alpha$) generated in the study reported in the next subsection. In this way, if the user is interested in this range of numbers of leaves and this range of parameters, he or she can study the distribution of the corresponding balance index efficiently and quickly.
	
\item Given a tree $T\in\TT_n$, estimate the percentile $q_{T,n,\alpha,\gamma}$ of its balance index $S$, $\Phi$, or $\mathfrak{C}_{D,f}$ with respect to the distribution of this index on $\TT_n$ under some $\alpha$-$\gamma$-model. If $n,\alpha,\gamma$ are among those mentioned in the previous item, for the sake of efficiency this function uses the database of computed indices to simulate the distribution of the balance index of $\TT_n$ under this $\alpha$-$\gamma$-model.
\end{itemize} 

For instance, the unlabeled tree $T\in \TT_8^*$  in Fig. \ref{tree1} is the shape of a phylogenetic tree  randomly generated under the $\alpha$-$\gamma$-model with $\alpha=0.7$ and $\gamma=0.4$ (using \texttt{set.seed(1000)} for reproducibility). The values of its balance indices are given in the figure's caption.


Fig. \ref{dis1} displays the estimation of the density function of the balance indices $\mathfrak{C}$, $S$, and $\Phi$ under the $\alpha$-$\gamma$-model with $\alpha=0.7$ and $\gamma=0.4$ on  $\TT_8$, obtained using the 5000 random trees gathered in our database. Moreover, the estimated percentiles of the balance indices of the tree of Fig.~\ref{tree1} are also shown in the figure.

Fig.~\ref{per1} shows a percentile plot of $\mathfrak{C}$, $S$, and $\Phi$ under the $\alpha$-$\gamma$-model for $\alpha=0.7$ and $\gamma=0.4$ on  $\TT_8$. The percentiles of the tree of Fig.~\ref{tree1} are given by the area to the left of the vertical lines.

%

Ford's $\alpha$-model for bifurcating phylogenetic trees \cite{Ford1}, which includes as special cases the {Yule}, or  {Equal-Rate Markov}, model \cite{Harding71,Yule}  and the {uniform}, or {Proportional to Distinguishable Arrangements},  model  \cite{CS,Pinelis}, is on its turn a special case of the $\alpha$-$\gamma$-model,  corresponding to the case $\alpha=\gamma$. So, this package allows also to study this model.
For example, the unlabeled tree in Fig.~\ref{treedis} has been generated (with \texttt{set.seed(1000)}) using $n=8$ and $\alpha=\gamma=0.5$, which corresponds to the uniform model. The figure also depicts the estimation of the density functions and of the percentile plots of $\mathfrak{C}$, $S$, and $\Phi$ on $\TT_8$ under this model, as well as the percentile values of the tree. %

%

\subsection*{Experimental results on TreeBASE}\label{sec:conc}

To assess the performance of  $\mathfrak{C}_{\MDM,\ln (n+\mathrm{e})}$, which we abbreviate by $\mathfrak{C}$, we downloaded (December 13-14, 2015) all phylogenetic trees  in  the TreeBASE database \cite{Treebase} using the function \verb?search_treebase()? of the R package \texttt{treebase} \cite{trebaR}. We obtained 13,008 trees, from which 80 had format problems that prevented R from reading them, so we restricted ourselves to the remaining 12,928 trees. To simplify the language, we shall still refer to this slightly smaller subset of phylogenetic trees as ``all trees in TreeBASE''. Only 4,814 among  these 12,928 trees in TreeBASE are bifurcating. 

Then, for every phylogenetic tree $T$ in this set, we have computed its Colless-like index $\mathfrak{C}(T)$, its Sackin index $S(T)$, and its total cophenetic index $\Phi(T)$. 
We have then compared the results obtained with $\mathfrak{C}$,  $S$, and $\Phi$ on TreeBASE in the following ways (all analysis have been performed with R \cite{R}).

\subsubsection*{Behavior as functions of the number of leaves.}\label{sec:5.1.1}
For every number of leaves $n$, we have computed the mean and the variance of  $\mathfrak{C}$, $S$ and   $\Phi$ on all trees with $n$ leaves in TreeBASE. Then, we have computed the regression of these values as a function of $n$. 

As far as the means go, the best fits have been:
\begin{itemize}
 \item\emph{Colless-like index}: $\overline{\mathfrak{C}} \approx 0.5351\cdot n^{1.5848}$, with a coefficient of determination of 
$R^2=0.9869$ and a p-value for the exponent $p<2\cdot 10^{-16}$.
 \item\emph{Sackin index}: $\overline{S}\approx 1.4512\cdot n^{1.4359}$, with a coefficient of determination of $R^2= 0.9953$  and a p-value for the exponent $p<2\cdot 10^{-16}$.
 \item\emph{Total cophenetic index}: $\overline{\Phi}\approx 0.1894\cdot n^{2.5478}$, with a coefficient of determination of $R^2= 0.9945$ and a p-value for the exponent $p<2\cdot 10^{-16}$.
\end{itemize}
Fig.~\ref{REL-INDEXS} depicts these mean values of $\mathfrak{C}$ (left), $S$ (center), and $\Phi$ (right) as functions of $n$.
%

Thus, $S$ and $\mathfrak{C}$ have similar growth rates, while $\Phi$ has a growth rate one order higher in magnitude. This difference vanishes if we normalize the indices by their range, which are $O(n^2)$ for $\mathfrak{C}$ and  $S$, and $O(n^3)$ for $\Phi$:
$$
\begin{array}{l}
\overline{\mathfrak{C}_{norm}} \approx 0.8389\cdot n^{-0.4152}\\[1ex]
\overline{S_{norm}}\approx 2.9024\cdot n^{-0.5641}\\[1ex]
\overline{\Phi_{norm}}\approx 1.1364\cdot n^{-0.4522}
\end{array}
$$

As far as the behavior of the variances goes, the best fits are the following:
\begin{itemize}
 \item\emph{Colless index}: $\mathrm{Var}(\mathfrak{C}) \approx 0.07599\cdot n^{3.12831}$, with a  
coefficient of determination of  $R^2=0.962$ and a p-value for the exponent $p<2\cdot 10^{-16}$.
 \item\emph{Sackin index}: $\mathrm{Var}(S)\approx 0.03182\cdot n^{3.22441}$, with a coefficient of determination of $R^2= 0.9575$  and a p-value for the exponent $p<2\cdot 10^{-16}$.

\item\emph{Total cophenetic index}: $\mathrm{Var}(\Phi)\approx 0.0041\cdot n^{5.2075}$, with a coefficient of determination of $R^2= 
0.9812$ and a p-value for the exponent $p<2\cdot 10^{-16}$.
\end{itemize}
The results are  in the same line as before, with the variances of $\mathfrak{C}$ and $S$ having similar growth rates, and the variance of  $\Phi$ having a growth rate two orders of magnitude higher. This difference  vanishes again when we normalize the indices:
$$
\begin{array}{l}
\mathrm{Var}(\mathfrak{C}_{norm}) \approx 0.18677\cdot n^{-0.87169}\\
\mathrm{Var}(S_{norm})\approx 0.12728\cdot n^{-0.77559}\\
\mathrm{Var}(\Phi_{norm})\approx 0.1476\cdot n^{-0.7925}
\end{array}
$$
So, in summary, $\mathfrak{C}$ has, on TreeBASE and relative to the range of values, a slightly larger mean growth rate and a slightly smaller variance growth rate
 than the other two indices.


\subsubsection*{Numbers of ties.}
The number of ties (that is, of pairs of different trees with the same index value) of a balance index is an interesting measure of quality, because the smaller its frequency of ties,  the bigger its ability to rank the balance of any pair of different trees.  Although, in our opinion, this ability 
need not always be an advantage: for instance, 
neither $\Phi$ nor $S$ take the same, minimum, value on all different fully symmetric trees with the same numbers of leaves (for example, $S(\FS_6)=6$ but $S(\FS_{2,3})=S(\FS_{3,2})=12$; and  $\Phi(\FS_6)=0$, but $\Phi(\FS_{3,2})=3$ and $\Phi(\FS_{2,3})=6$; cf.  Fig. \ref{fig:1}), while $\mathfrak{C}$ applied to any fully symmetric tree is always 0.  In this case, we believe that these ties are fair.


Anyway, for every number of leaves $n$ and for every one of all three indices under scrutiny, we have computed the numbers of pairs of trees with $n$ leaves  in TreeBASE having the same value of the corresponding index (in the case of $\mathfrak{C}$, up to 16 decimal digits). 
Fig.~\ref{TIES} plots  the frequencies of ties of $\mathfrak{C}$, $S$ and $\Phi$ as functions of $n$. 
As it can be seen in this graphic, $\mathfrak{C}$ and $\Phi$ have a similar number of ties, and consistently less ties than $S$. 


\subsubsection*{Spearman's rank correlation.}
In order to measure whether all three indices sort the trees according to their balance in the same way or not,
we have computed the Spearman's  rank correlation coefficient \cite{Spearman} of the indices on all trees in TreeBASE, as well as grouping them by their number of leaves $n$. 

The global Spearman's rank correlation coefficient 
of $\mathfrak{C}$ and $S$ is $0.9765$,  and that of $\mathfrak{C}$ and $\Phi$ is $0.9619$.
The graphics in Fig.~\ref{SPEARMAN} plot these coefficients as functions of $n$.
As it can be seen, Spearman's rank correlation coefficient  for $\mathfrak{C}$ and $S$ grows with $n$, approaching to 1,
while the coefficient  for $\mathfrak{C}$ and $\Phi$ shows a decreasing tendency with $n$.


\subsection*{Does TreeBASE fit the uniform model or the alpha-gamma model?}

In this subsection, we  test whether the distribution of the Colless-like index of the phylogenetic trees in TreeBASE agrees with its theoretical distribution under either the uniform model for multifurcating phylogenetic trees \cite{OdenShao} or the $\alpha$-$\gamma$-model \cite{Ford2} for some parameters $\alpha,\gamma$. To do it, we use the normalized version $\mathfrak{C}_{norm}$ of $\mathfrak{C}$, which can be used simultaneously on trees with different numbers of leaves.

To estimate the theoretical distribution of this index under the two aforementioned theoretical models, for every $n=3,\ldots,50$ we have generated, on the one hand, 10,000 random phylogenetic trees in $\TT_n$ under the uniform model using the algorithm described in \cite{OdenShao}, and, on the other hand,  5000 random phylogenetic trees in $\TT_n$ under the $\alpha$-$\gamma$-model for every pair of parameters $(\alpha,\gamma)\in \{0,0.1,0.2,\ldots,0.9,1\}^2$ with $\gamma\leq \alpha$. We have computed the  value of $\mathfrak{C}_{norm}$ on all these trees,
and we have used  the distribution of these values as an estimation of the corresponding theoretical distribution.
To test whether the distribution of the normalized Colless-like index on TreeBASE (or on some subset of it: see below) fits one of these theoretical distributions, we have performed two non-parametric statistical tests on the observed set of indices of TreeBASE and the corresponding simulated set of indices: Pearson's chi-squared test and the Kolmogorov-Smirnov test, using bootstrapping techniques in the latter to avoid problems with ties. 

As a first approach, we have performed these tests on the whole set of trees in TreeBASE. The p-values obtained in all tests, be it for the uniform model or for any considered pair $(\alpha,\gamma)$, have turned out to be negligible. Then, we conclude confidently that the distribution of the normalized Colless-like index on the TreeBASE does not fit either the uniform model or any $\alpha$-$\gamma$-model when we round $\alpha,\gamma$ to one decimal place. For instance,  Fig.~\ref{fig:distunif} displays the distribution of $\mathfrak{C}_{norm}$ on TreeBASE and its estimated theoretical distribution under the uniform model. As it can be seen, these distributions are quite different, which confirms the conclusion of the statistical test.

%

Fig.~\ref{fig:distag} displays the distribution of $\mathfrak{C}_{norm}$ for all trees in TreeBASE and its estimated theoretical distribution under the $\alpha$-$\gamma$-model for the pair of parameters $\alpha,\gamma$ that gave the largest p-values in the goodness of fit tests, which are $\alpha=0.7$ and $\gamma=0.4$. Although graphically both distributions are quite similar, the p-values of the Pearson chi-squared test and of the Kolmogorov-Smirnov test are virtually zero. One might think that the high ``peaks'' of the theoretical distribution near $0$ and $1$ could have influenced the outcome of these statistical tests. For this reason, we have repeated them without taking into account these ``extreme'' values, and the results have been the same.

Since TreeBASE gathers phylogenetic trees of different types and from different sources, we have also considered subsets of it  defined by means of attributes. More specifically, besides the whole TreeBASE as explained above, we have also considered the following subsets of it:
\begin{itemize}
\item All trees in TreeBASE up to repetitions: we have removed 513 repeated trees  (which represent about a 4\% of the total).
\item All trees with their \texttt{kind} attribute equal to ``Species''. This \texttt{kind} attribute can take three values: ``Barcode tree'', ``Gene Tree'' and ``Species Tree''. 
\item All trees with their \texttt{kind} attribute equal to ``Species'' and their \texttt{type} attribute equal to ``Consensus''.  This \texttt{type} attribute can take two values: ``Consensus'' and ``Single''. 
\item All trees with their \texttt{kind} attribute equal to ``Species'' and their \texttt{type} attribute equal to ``Single''. 
\end{itemize}

We have repeated the study explained above for these four subsets of TreeBASE, comparing the distribution of the normalized Colless-like indices of their trees with the estimated theoretical distributions by means of goodness-of-fit tests, and the results have been the same, that is, all p-values have also turned out to be  negligible. Our conclusion is, then, that neither the whole TreeBASE nor any of these four subsets of it seem to fit either the uniform model or some $\alpha$-$\gamma$-model.

%

\section*{Conclusions}\label{sec:trueconc}

In this paper we have introduced a family of \emph{Colless-like} balance indices $\mathfrak{C}_{D,f}$, which depend on a dissimilarity $D$ and a function $f:\NN\to \RR_{\geq 0}$, 
 that generalize  the Colless index  to multifurcating phylogenetic trees by defining the \emph{balance} of an internal node as the spread, measured through $D$,  of the sizes, defined through $f$, of the subtrees rooted at its children. We have proved that every combination of a dissimilarity $D$ and a function either $f(n)=\ln(n+e)$ or $f(n)=e^n$, defines a Colless-like index that is \emph{sound} in the sense that the maximally balanced trees according to it are exactly
 the fully symmetric ones. But, the growth of the function $f$ used to define the sizes that are compared in the definition of $\mathfrak{C}_{D,f}$ determines strongly what the most unbalanced trees are, and hence it has influence on the very notion of ``balance''  measured by the index.
 
In our opinion, choosing $\ln(n+e)$ instead of $e^n$ seems a more sensible decision, because, on the one hand, the most unbalanced trees according to the former are the expected ones ---the combs--- and, on the other hand, we have encountered several hard numeric problems when working with the extremely large figures that appear when using $e^n$-sizes on trees with internal nodes of high degree. As far a choosing the dissimilarity $D$ goes,  $\MDM$ and $\sd$  define indices that are  proportional to the Colless index when applied to bifurcating trees. Among these two options, we recommend to use $\MDM$ because it only involves linear operations, and hence it has less numerical precision problems than $sd$, that uses a square root of a sum of squares. This is the reason we have stuck to $\mathfrak{C}_{\MDM,\ln(n+e)}$ in the numerical experiments reported in the Results section.

We want to call the reader's attention on the problem posed in the ``Sound Colless-like indices" subsection: to find functions $f$ such that $\mathfrak{C}_{D,f}$ is sound. Our conjecture is that no function $f:\NN\to \NN$ satisfies this property.

\section*{Supporting information}

%

\paragraph*{S1 file: Proofs of Theorems \ref{thm:mdm} and \ref{thm:mdm-exp}.} The file provides the detailed proofs of Theorems \ref{thm:mdm} and \ref{thm:mdm-exp}.

\paragraph*{S2 file: Tables.} The file provides two tables, quoted in the main text, with the values of several Colless-like indices on $\TT^*_n$ for $n=2,3,4,5$.

\section*{Competing interests}
  The authors declare that they have no competing interests.


\section*{Acknowledgements}
 This  research has been partially supported by the \textsl{Obra Social la Caixa}  through the ``Programa Pont La Caixa per a grups de recerca de la UIB''  and by the Spanish Ministry of Economy and Competitiveness and European Regional Development Fund through project DPI2015-67082-P (MINECO/FEDER).
 We thank K. Bartoszek for several useful suggestions and A. Salda\~na Plomer for making available to us his Java script that  generates random phylogenetic trees with a fixed number of leaves with uniform distribution.
 



\section*{Figures}

\begin{figure}[h!]
\begin{center}
\includegraphics[width=0.95\linewidth]{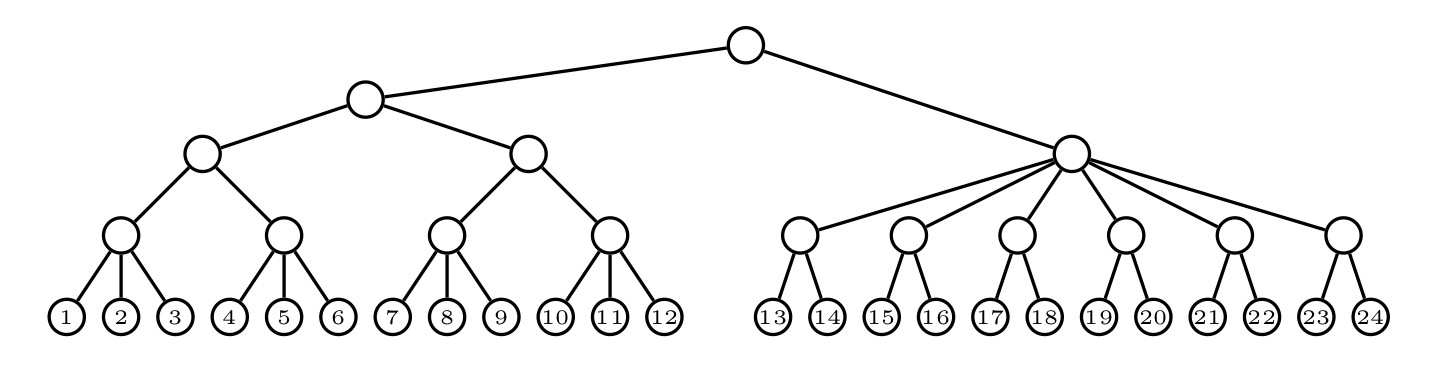}
\end{center}
\caption{\label{fig:nocoll1} Each node in this asymmetric tree has all its children with the same number of descendant leaves as well as with the same number of descendant nodes.}
\end{figure}

\begin{figure}[h!]
\begin{center}
\includegraphics[width=0.25\linewidth]{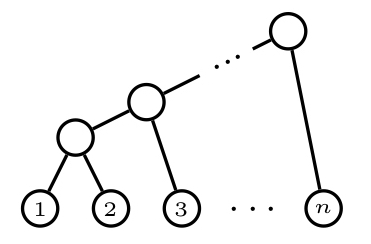}
\end{center}
\caption{\label{fig:exs1} 
A comb $K_n$ with $n$ leaves.}
\end{figure}

\begin{figure}[h!]
\begin{center}
\includegraphics[width=0.25\linewidth]{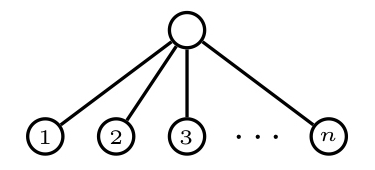}
\end{center}
\caption{\label{fig:exs2} 
A star  $\FS_n$ with $n$ leaves.}
\end{figure}

\begin{figure}[h!]
\begin{center}
\includegraphics[width=0.4\linewidth]{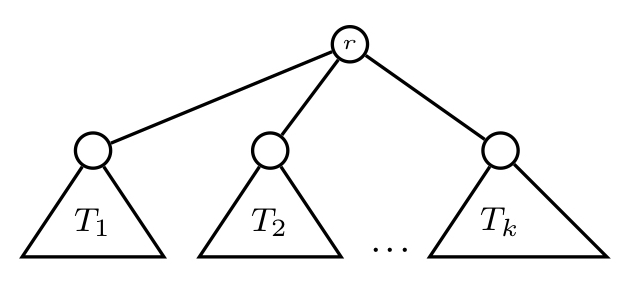}
\end{center}
\caption{\label{fig:star} The (phylogenetic) tree $T_1\star\cdots\star T_k$.}
\end{figure}

\begin{figure}[h!]
\begin{center}
\includegraphics[width=0.9\linewidth]{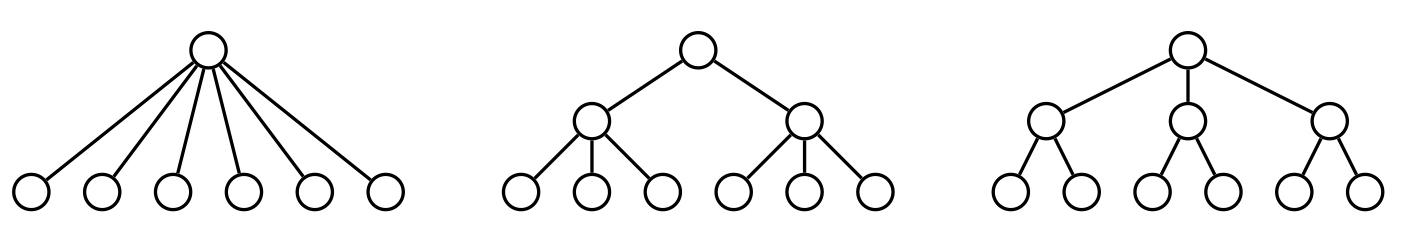}
\end{center}
\caption{\label{fig:1} 
Three fully symmetric  trees with 6 leaves: from left to right, $\FS_6$, $\FS_{2,3}$ and $\FS_{3,2}$.}
\end{figure}

\begin{figure}[h!]
\begin{center}
\includegraphics[width=0.8\linewidth]{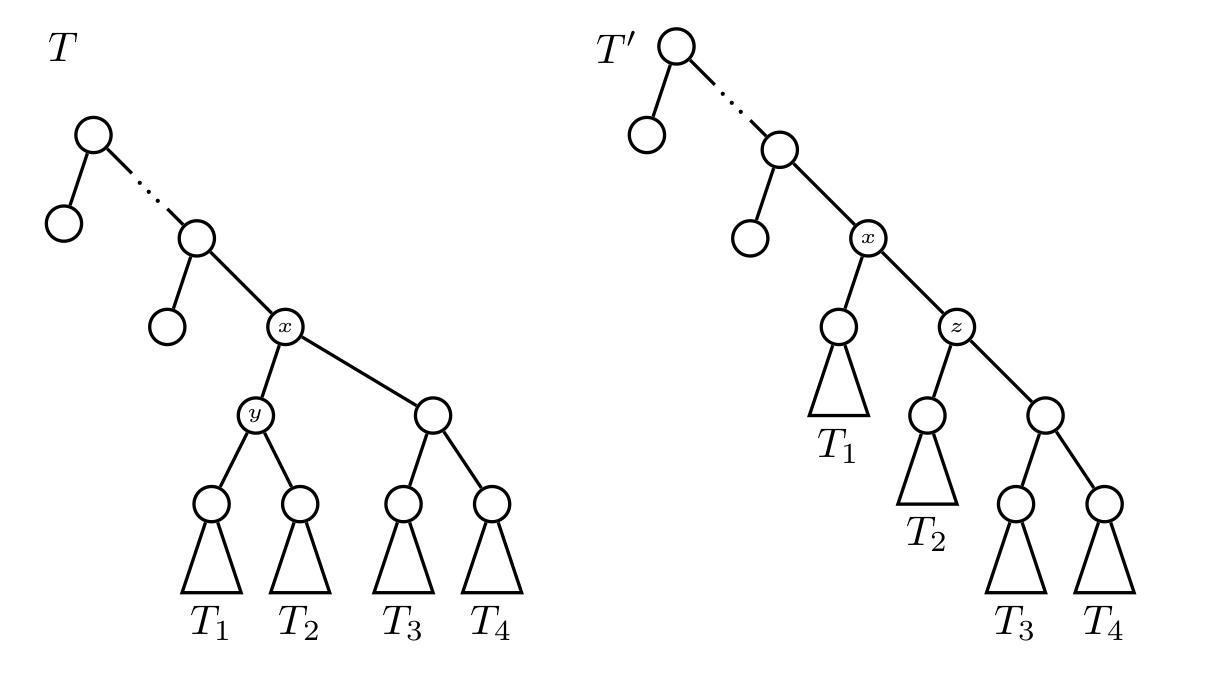}
\end{center}
\caption{\label{fig:cat1} The trees $T$ and $T'$ in the proof of Lemma \ref{lem:cat-C}.}
\end{figure}

\begin{figure}[h!]
\begin{center}
\includegraphics[width=0.25\linewidth]{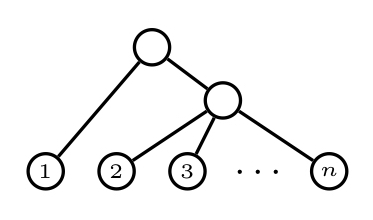}
\end{center}
\caption{\label{fig:maxexp1} 
The tree  $\FS_1\star \FS_{n-1}$.}
\end{figure}

\begin{figure}[h!]
\begin{center}   
\includegraphics[width=0.5\linewidth]{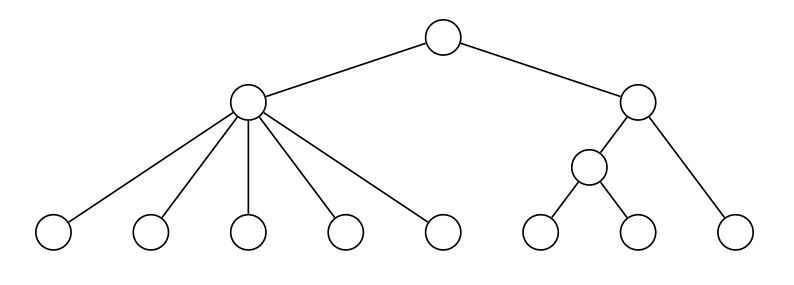}
\end{center}
\caption{A  tree with 8 leaves randomly generated under the $\alpha$-$\gamma$-model with $\alpha=0.7$ and $\gamma=0.4$. Its indices are $\mathfrak{C}(T)=1.746 $, $S(T)=18$, and $\Phi(T)=14$, and its normalized indices are $\mathfrak{C}_{norm}(T)=0.06518 $, $S_{norm}(T)=0.3704$, and $\Phi_{norm}(T)=0.25$.}
\label{tree1}
\end{figure} 

\begin{figure}[h!] 
\centering 
 \includegraphics[height=7cm]{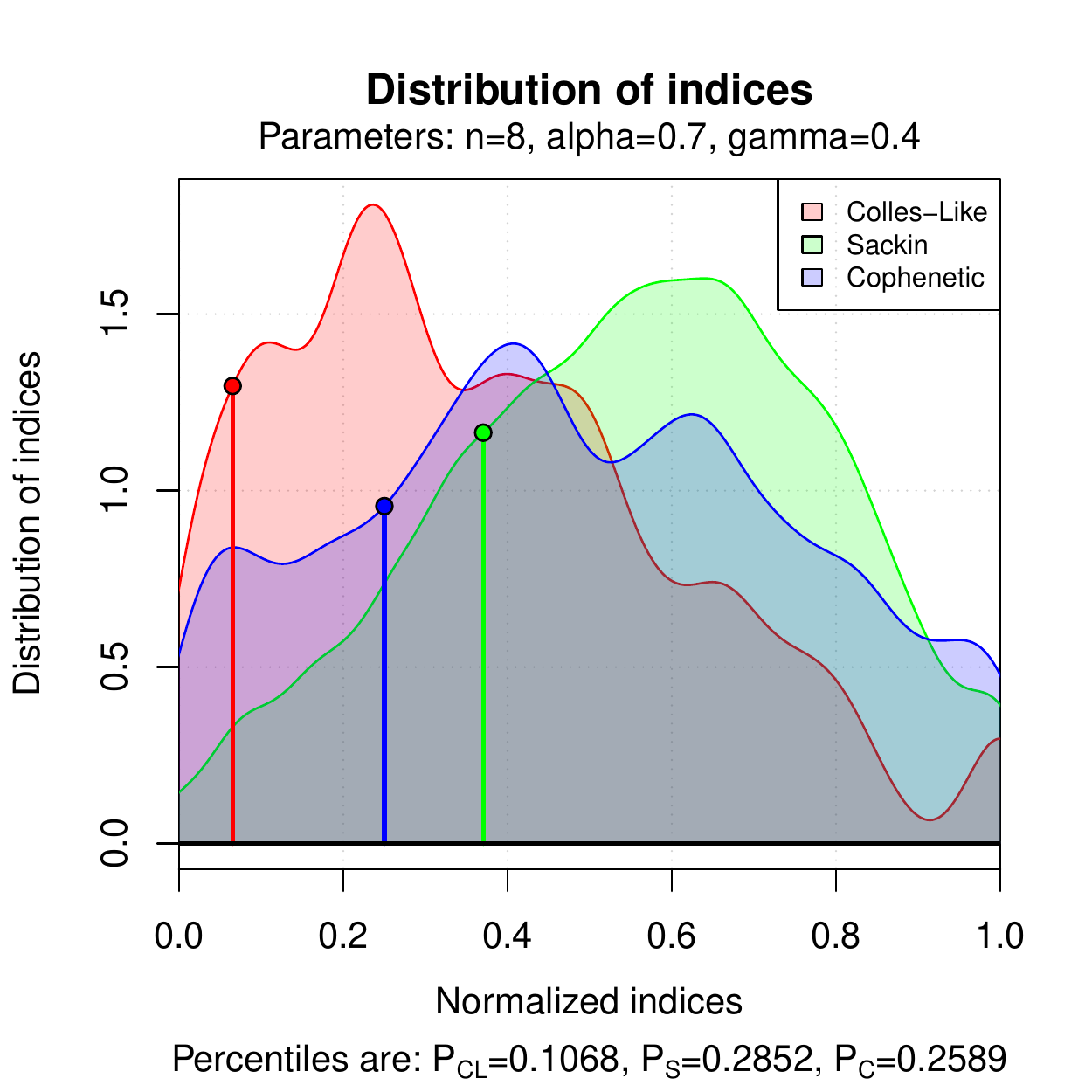}
 \caption{The estimated density function of the distribution of $\mathfrak{C}$, $S$ and $\Phi$ on $\TT_8$ under the $\alpha$-$\gamma$-model with $\alpha=0.7$ and $\gamma=0.4$. The percentiles of the tree in Fig.~\ref{tree1} are also represented. }
\label{dis1}
\end{figure} 

\begin{figure}[h!] 
\centering 
 \includegraphics[height=7cm]{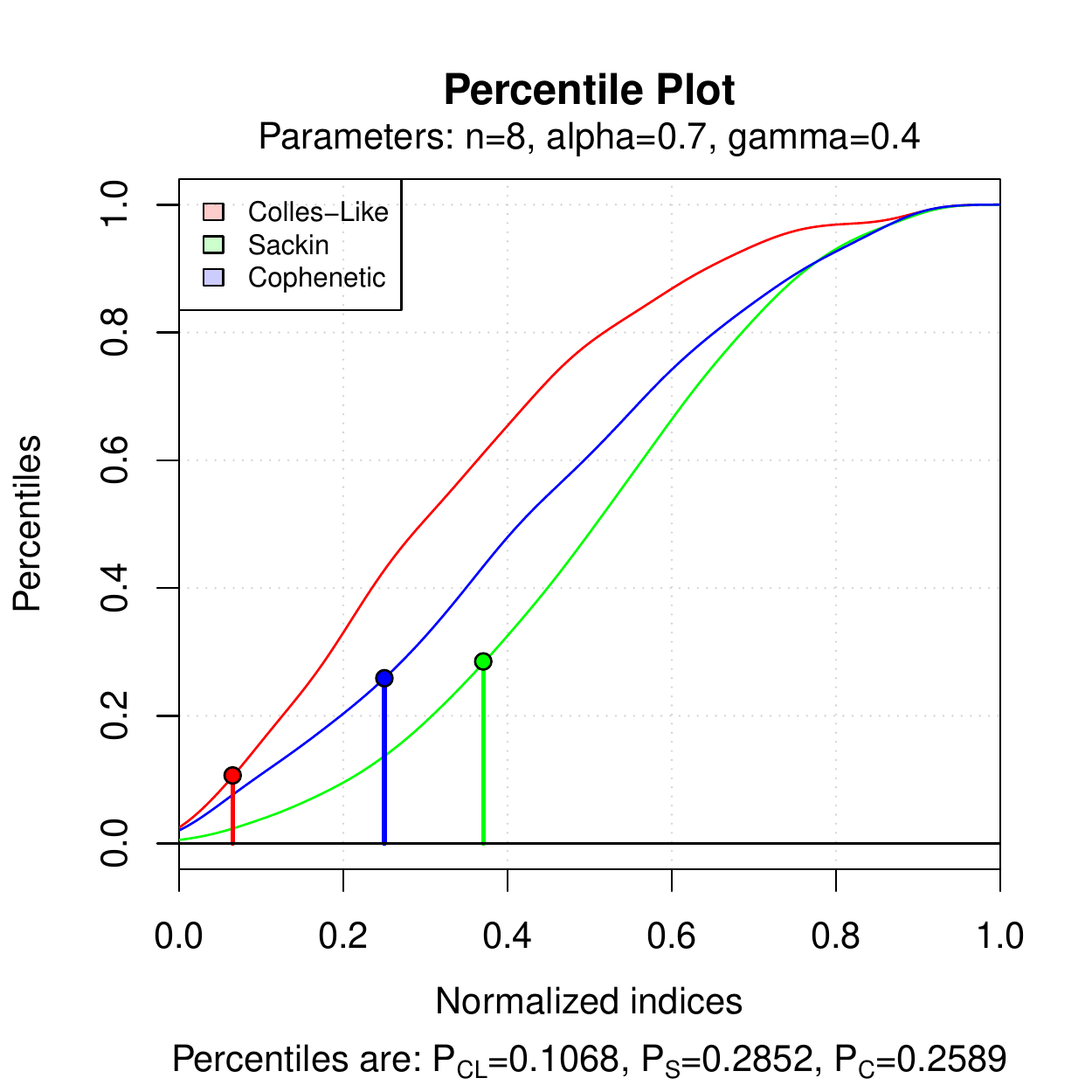}
 \caption{Percentile plot of the distribution of $\mathfrak{C}$, $S$ and $\Phi$ on $\TT_8$ under the $\alpha$-$\gamma$-model with $\alpha=0.7$ and $\gamma=0.4$. The percentiles of the tree of Fig.~\ref{tree1} are also highlighted.}
\label{per1}
\end{figure} 

\begin{figure}[h!]
\begin{center} 
\includegraphics[height=10cm]{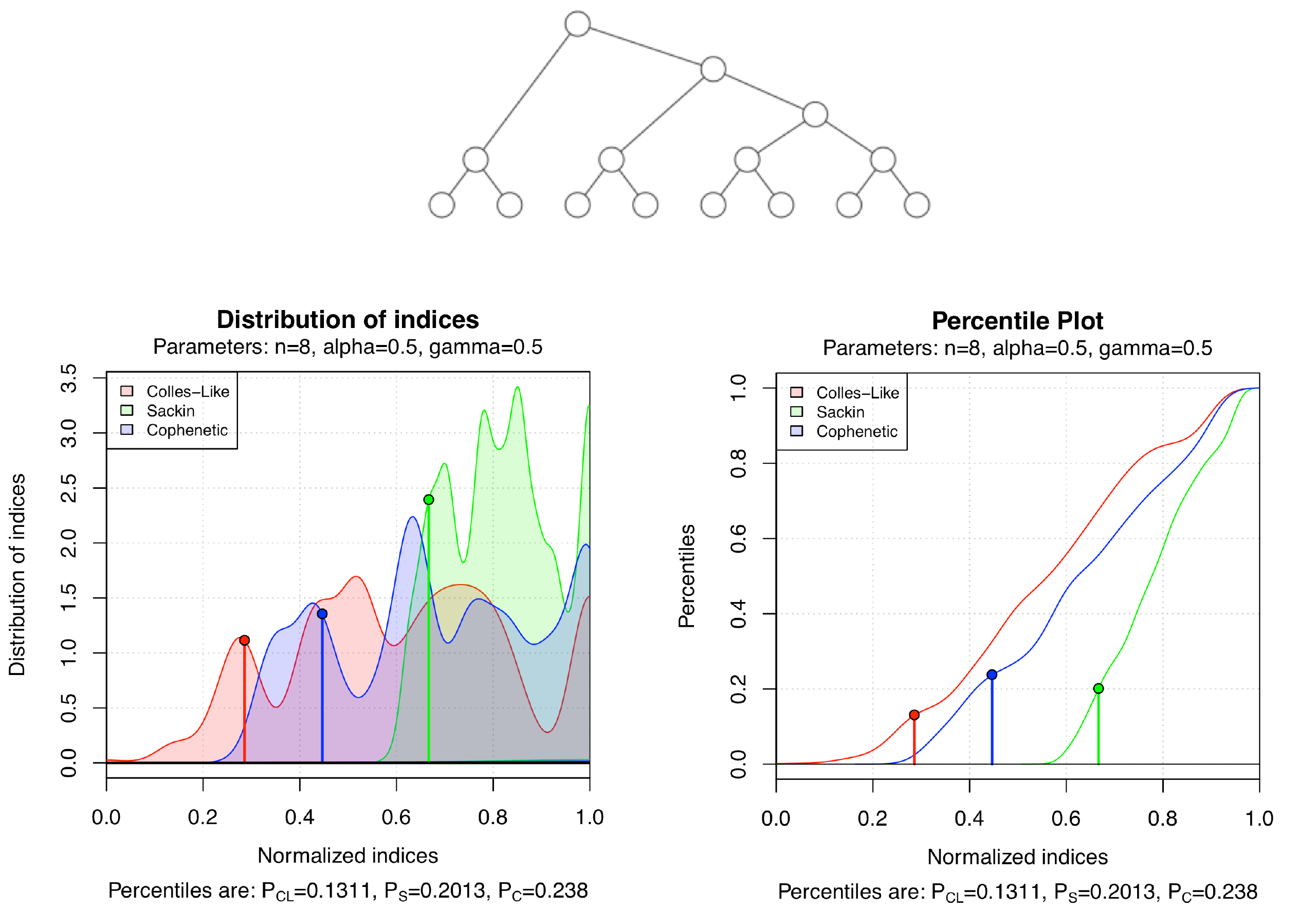}
 \caption{A bifurcating tree randomly generated under the uniform model, the estimated density function of the distribution of the three balance indices on $\TT_8$ under the uniform model, and their percentile plot. }
\label{treedis} 
\end{center}
\end{figure}

\begin{figure}[h!]
\begin{center}
 \begin{tabular}{ccc}
\includegraphics[height=4cm]{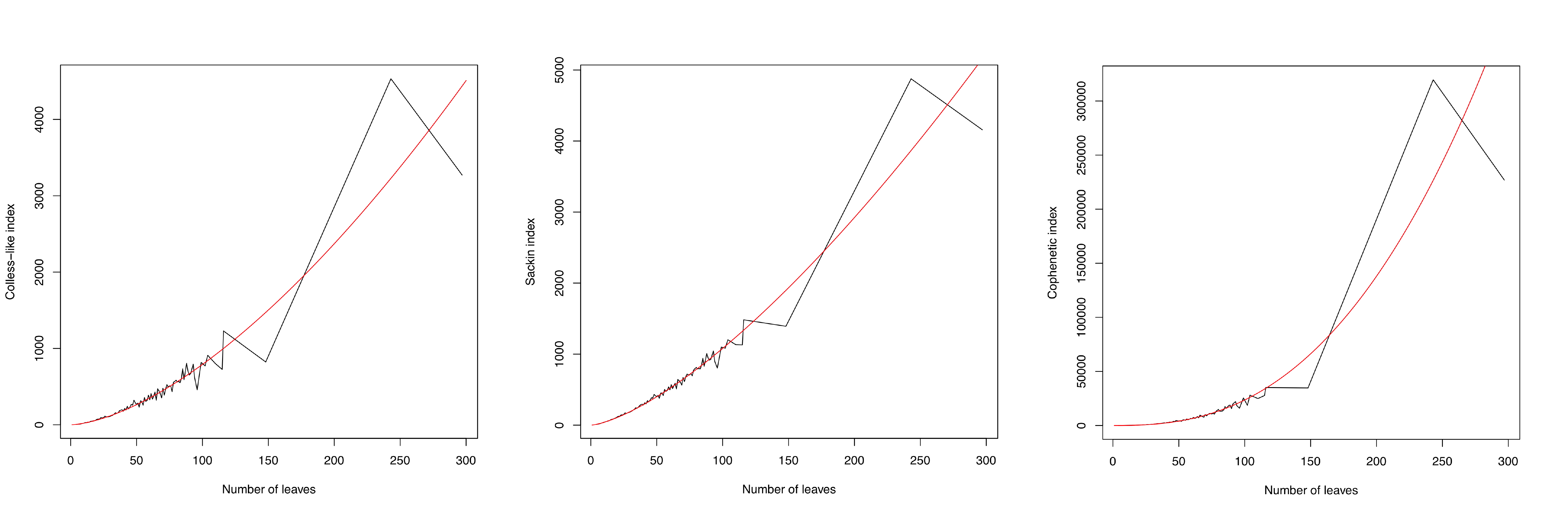}
 \end{tabular}
\end{center}
\caption{Growth of the mean value of $\mathfrak{C}$ (left),  $S$ (center), and $\Phi$ (right)   in TreeBASE, as functions of the trees' numbers of leaves $n$.}
\label{REL-INDEXS}
\end{figure}

\begin{figure}[h!]
 \begin{center}
\includegraphics[height=7cm]{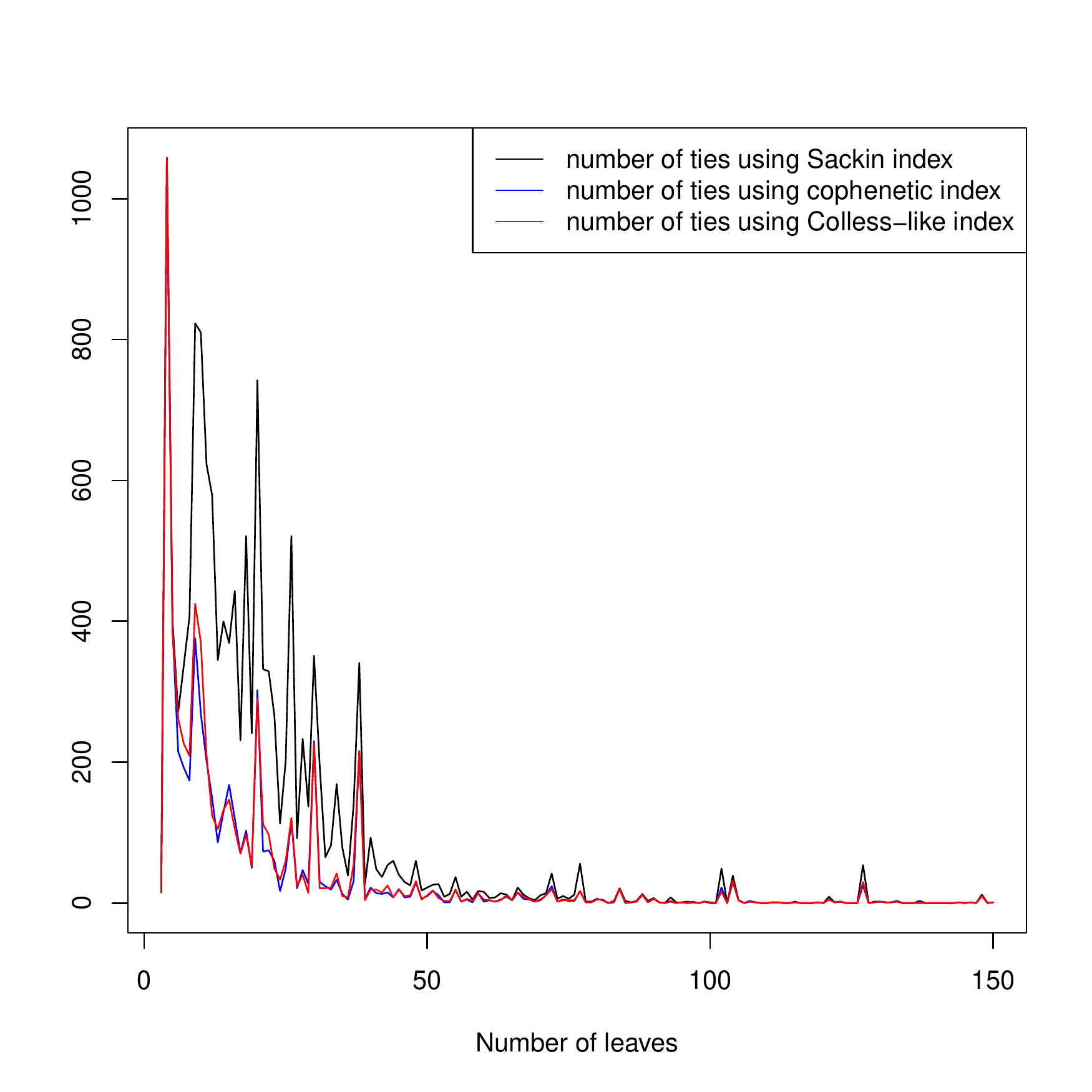}
 \end{center}
 \caption{Numbers of ties of $\mathfrak{C}$, $S$, and $\Phi$ in TreeBASE, as functions of the trees' numbers of leaves $n$.}
\label{TIES}
\end{figure}

\begin{figure}[h!]
 \begin{center}
 \begin{tabular}{cc}
\includegraphics[height=5cm]{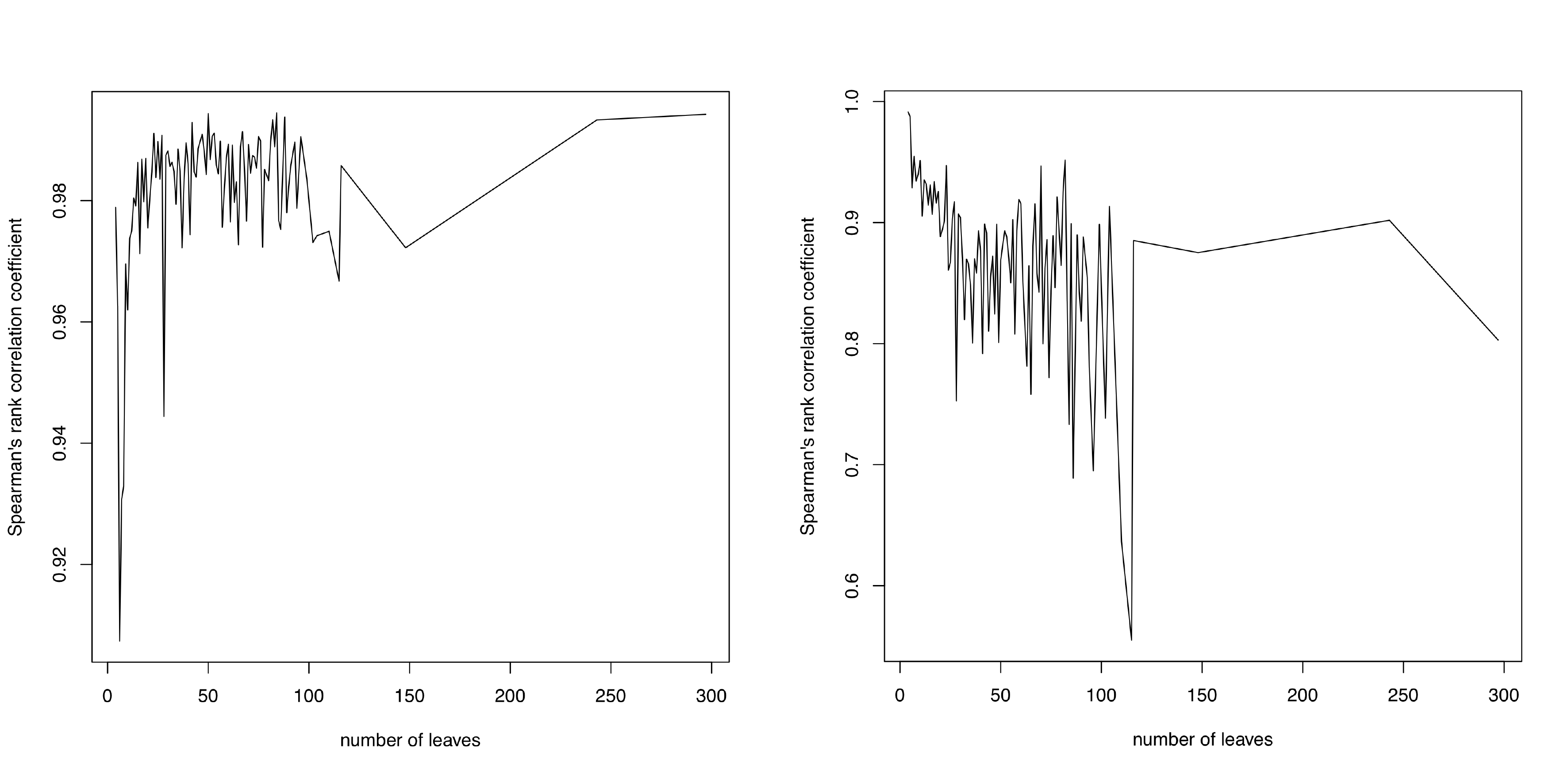} 
\end{tabular}
\end{center}
 \caption{Spearman's rank correlation coefficient of $\mathfrak{C}$ and $S$ (left)
and of $\mathfrak{C}$ and $\Phi$ (right) in TreeBASE, as functions of the trees' numbers of leaves $n$.}
\label{SPEARMAN}
\end{figure}

\begin{figure}[h!]
 \begin{center}
 \begin{tabular}{cc}
\includegraphics[height=7cm]{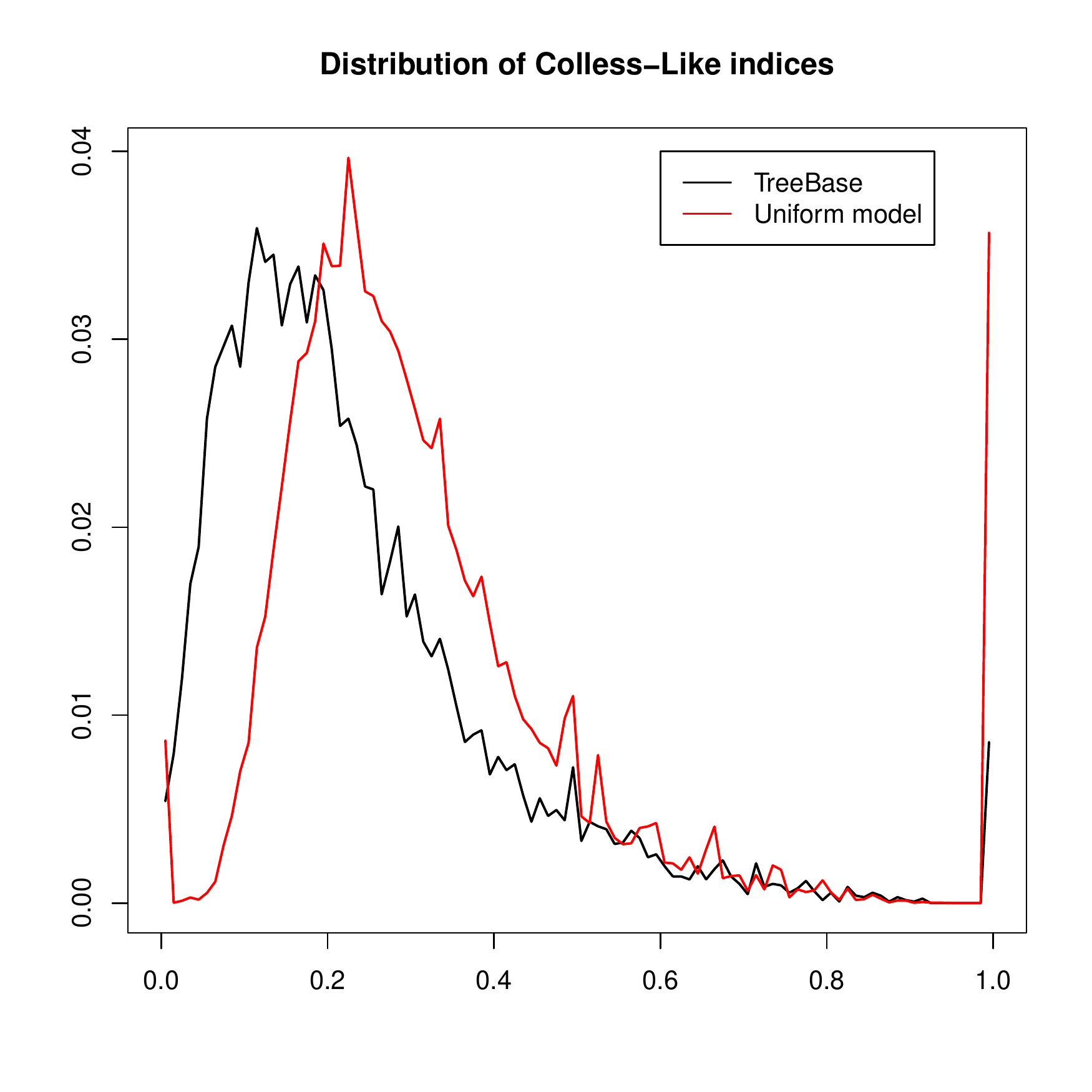}
\end{tabular}
\end{center}
 \caption{The distribution of $\mathfrak{C}_{norm}$ on all trees in TreeBASE (black line) and its estimated theoretical distribution under the uniform model (red line).}
\label{fig:distunif}
\end{figure}

\begin{figure}[h!]
 \begin{center}
 \begin{tabular}{cc}
\includegraphics[height=7cm]{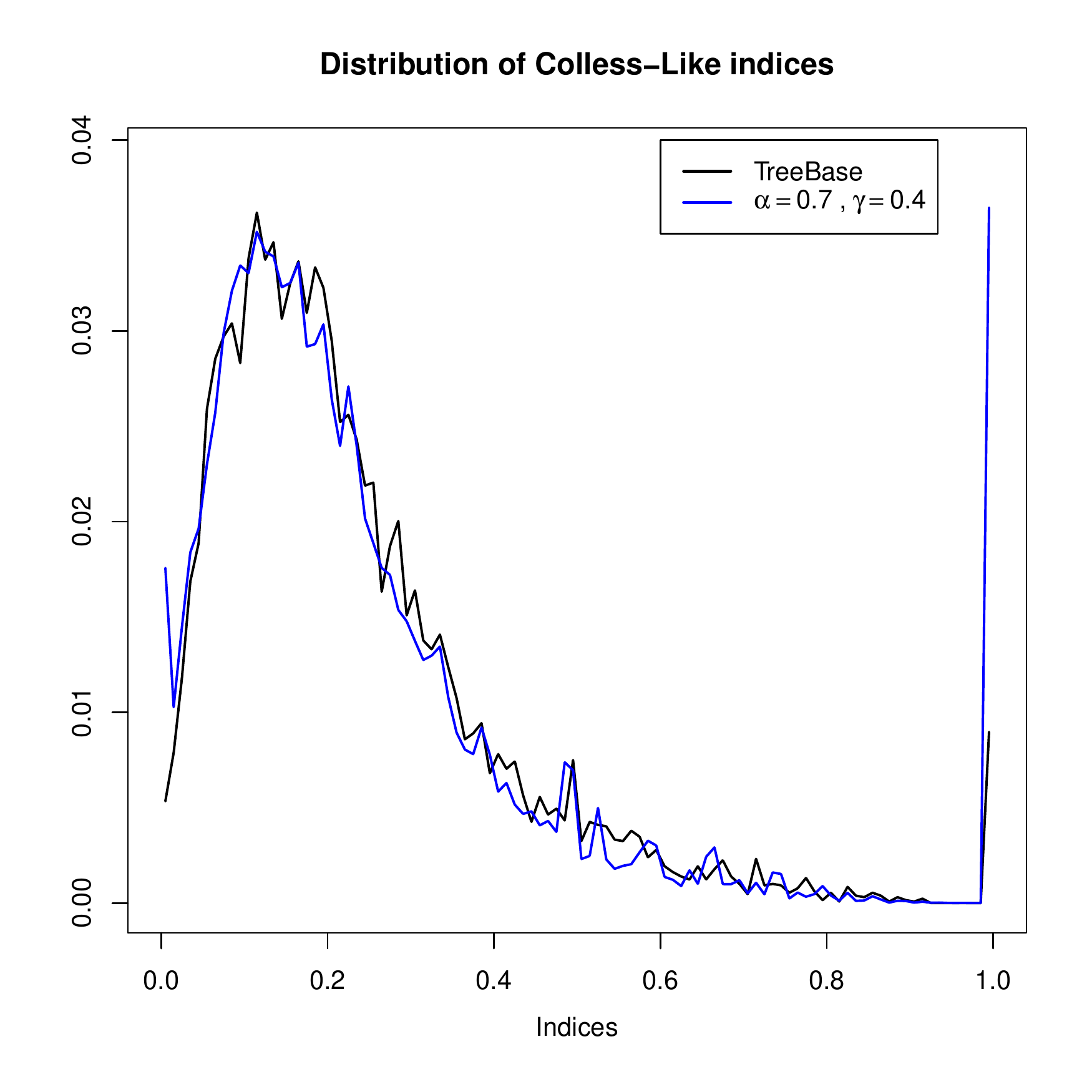}
\end{tabular}
\end{center}
 \caption{The distribution of $\mathfrak{C}_{norm}$ on all trees in TreeBASE (black line) and its estimated theoretical distribution  under the $\alpha$-$\gamma$-model with $\alpha=0.7$ and $\gamma=0.4$ (blue line).}
\label{fig:distag}
\end{figure}

\section*{Supplementary files}

\begin{center}
\Large Supplementary file S1: Proofs of Theorems 18 and 19
\end{center}

\noindent This supplementary document contains the proofs of Theorems 18 (Sections 1--3) and 19 (Sections 4--6) in the main text.

\section{Proof of the thesis of Theorem 18 for $\mathfrak{C}_{\MDM,f}$}\label{SM1}

Let in this section, and in the next two ones, $f:\NN\to \RR_{\geq 0}$ be any mapping such that $0<f(k)<f(k-1)+f(2)$ for every $k\geq 3$. Notice that if $f$ satisfies this condition then $f(k)>0$ not only  for every $k\geq 3$, but also for $k=2$, because
$0<f(3)< 2f(2)$. To simplify the notations, we shall denote in this section $\delta_f$ and $\mathfrak{C}_{\MDM,f}$ by $\delta$ and $\mathfrak{C}$, respectively, and we shall denote ${bal}_{\MDM,f}$ on a tree $T$ by ${bal}_T$ or simply by ${bal}$ when it is not necessary to specify the tree.

We split this proof  into several lemmas.

\begin{lemma}\label{lem:mdmodd}
For every $(x_1,\ldots,x_{2n+1})\in \RR^{2n+1}$ with $n\geq 1$, if $x_i$ is the median of $\{x_1,\ldots,x_n\}$,
then 
$$
\MDM (x_1,\ldots,x_{2n+1})\leq \MDM(x_1,\ldots,x_{i-1},x_{i+1},\ldots,x_{2n+1}).
$$
Moreover, this inequality is strict unless $x_1=\cdots=x_{2n+1}$.
\end{lemma}

\begin{proof}
After rearranging $x_1,\ldots, x_{2n+1}$ if necessary, we   assume that $x_1\leq \cdots\leq\allowbreak x_{2n+1}$, in which case their median (i.e., their middle value) is $x_{n+1}$, and  the median of $x_1,\ldots,x_n,x_{n+2},\ldots, x_{2n+1}$ is $M=(x_n+x_{n+2})/2$. We want to prove that
$$
\MDM (x_1,\ldots,x_{2n+1})\leq\MDM (x_1,\ldots,x_n,x_{n+2},\ldots, x_{2n+1}).
$$
This inequality is true, because
$$
\begin{array}{l}
\displaystyle \MDM (x_1,\ldots,x_{2n+1})  = \frac{1}{2n+1}\sum\limits_{i=1}^{2n+1}|x_i-x_{n+1}|\\[1ex]
\displaystyle\qquad =
\frac{1}{2n+1}\Big(\sum\limits_{i=1}^{n}(x_{n+1}-x_i)+\sum\limits_{i=n+2}^{2n+1}(x_i-x_{n+1})\Big) 
= \frac{1}{2n+1}\Big(\sum\limits_{i=n+2}^{2n+1} x_i-\sum\limits_{i=1}^{n} x_i\Big)
\\[3ex]
\displaystyle\MDM (x_1,\ldots,x_n,x_{n+2},\ldots, x_{2n+1}) = \frac{1}{2n}\Big(\sum\limits_{i=1}^{n}|x_i-M|
+\sum\limits_{i=n+2}^{2n+1}|x_i-M|\Big)\\[1ex]
\displaystyle\qquad =\frac{1}{2n}\Big(\sum\limits_{i=1}^{n}(M-x_i)
+\sum\limits_{i=n+2}^{2n+1}(x_i-M)\Big)  = \frac{1}{2n}\Big(\sum\limits_{i=n+2}^{2n+1} x_i-\sum\limits_{i=1}^{n} x_i\Big),
\end{array}
$$
and, clearly,
$$
\frac{1}{2n+1}\Big(\sum\limits_{i=n+2}^{2n+1} x_i-\sum\limits_{i=1}^{n} x_i\Big)\leq \frac{1}{2n}\Big(\sum\limits_{i=n+2}^{2n+1} x_i-\sum\limits_{i=1}^{n} x_i\Big).
$$
Moreover, the inequality is strict unless 
$\sum\limits_{i=n+2}^{2n+1} x_i=\sum\limits_{i=1}^{n} x_i$,  which, under the assumption that $x_1\leq \cdots\leq x_{2n+1}$, is equivalent to  $x_1=\cdots=x_{2n+1}$.
\end{proof}

Unfortunately, the thesis of this lemma is false for vectors of numbers of even length. For instance,
consider the vector $(1,1,2,2)$: if we remove any single element, its $\MDM$ decreases.  But, providentially,
we can always increase the MDM of an even quantity of real numbers by removing \emph{two} of them.

\begin{lemma}\label{lem:mdmeven}
For every $(x_1,\ldots,x_{2n})\in \RR^{2n}$ with $n\geq 2$, if $x_i,x_j$, with $i<j$, are the middle values of 
$\{x_1,\ldots,x_{2n}\}$, then
$$
\MDM (x_1,\ldots,x_{2n})\leq \MDM(x_1,\ldots,x_{i-1},x_{i+1},\ldots,x_{j-1},x_{j+1},\ldots,x_{2n}).
$$
Moreover, this inequality is strict unless $\{x_1,\ldots,x_{2n}\}$ consists either of $2n$ copies of a single element or of $n$ copies of two different  elements.
\end{lemma}

\begin{proof}
After rearranging $x_1,\ldots, x_{2n}$ if necessary, we  assume that $x_1\leq \cdots\leq x_{2n}$, so that their middle values are $x_{n},x_{n+1}$, and hence  their median is $M=(x_n+x_{n+1})/2$ and  the median of $x_1,\ldots,x_{n-1},x_{n+2},\ldots, x_{2n}$ is $M'=(x_{n-1}+x_{n+2})/2$.
We want to prove that 
$$
\MDM (x_1,\ldots,x_{2n})\leq\MDM (x_1,\ldots,x_{n-1},x_{n+2},\ldots, x_{2n}).
$$
And, indeed, 
$$
\begin{array}{l}
\MDM (x_1,\ldots,x_{2n})\leq\MDM (x_1,\ldots,x_{n-1},x_{n+2},\ldots, x_{2n})\\ \quad \Longleftrightarrow\displaystyle
\frac{1}{2n}\sum\limits_{i=1}^{2n}|x_i-M|\leq
\frac{1}{2n-2}\Big(\sum\limits_{i=1}^{n-1}|x_i-M'|+\sum\limits_{i=n+2}^{2n}|x_i-M'|\Big)\\ \quad \displaystyle\Longleftrightarrow
(2n-2)\Big(\sum\limits_{i=1}^{n}(M-x_i)+\hspace*{-1ex}\sum\limits_{i=n+1}^{2n}(x_i-M)\Big)\leq
2n\Big(\sum\limits_{i=1}^{n-1}(M'-x_i)+\hspace*{-1ex}\sum\limits_{i=n+2}^{2n}(x_i-M')\Big)\\ \quad \displaystyle\Longleftrightarrow
(n-1)\Big(\sum\limits_{i=n+1}^{2n}x_i-\sum\limits_{i=1}^{n}x_i\Big)\leq
n\Big(\sum\limits_{i=n+2}^{2n}x_i-\sum\limits_{i=1}^{n-1}x_i\Big)=n\Big(\sum\limits_{i=n+1}^{2n}x_i-\sum\limits_{i=1}^{n}x_i\Big)-n(x_{n+1}-x_n)\\ \quad \displaystyle\Longleftrightarrow
n(x_{n+1}-x_n)\leq \sum\limits_{i=n+1}^{2n}x_i-\sum\limits_{i=1}^{n}x_i
=\sum\limits_{i=1}^{n}(x_{n+i}-x_{n+1-i})
\end{array}
$$
and this last inequality is true because $x_{n+1}-x_n\leq x_{n+i}-x_{n+1-i}$ for every $i=1,\ldots,n$.
Moreover, the inequality is strict unless $x_{n+1}-x_n= x_{n+i}-x_{n+1-i}$ for every $i=1,\ldots,n$, that is, unless $x_1=\cdots=x_n$ and $x_{n+1}=\cdots=x_{2n}$.
\end{proof}

\begin{lemma}\label{lem:2mdm}
Let $f$ be a mapping $\NN\to \RR_{\geq 0}$ such that $f(k)>0$, for every $k\geq 2$, and
let $T$ be a tree of the form $T_1\star\cdots\star T_k$, with $k\geq 3$ (see Fig. 4 in the main text).
\begin{enumerate}[(a)]
\item If $k$ is an odd number and if $\delta(T_1)$ is the median of $\{\delta(T_1),\ldots,\delta(T_k)\}$, 
then the tree $T'=T_1\star(T_2\star \cdots \star T_k)$  (cf. Fig. \ref{fig:lem2mdm}) satisfies that $\mathfrak{C}(T')> \mathfrak{C}(T)$.

\item If $k$ is an even number and if $\delta(T_1), \delta(T_2)$ are the middle values of $\{\delta(T_1),\ldots,\delta(T_k)\}$,  with  $\delta(T_1)\leq \delta(T_2)$, then the tree  $T''=T_1\star(T_2\star (T_3\star \cdots \star T_k))$    (cf. Fig. \ref{fig:lem2mdm}) satisfies that $\mathfrak{C}(T'')> \mathfrak{C}(T)$.
\end{enumerate}
\end{lemma}

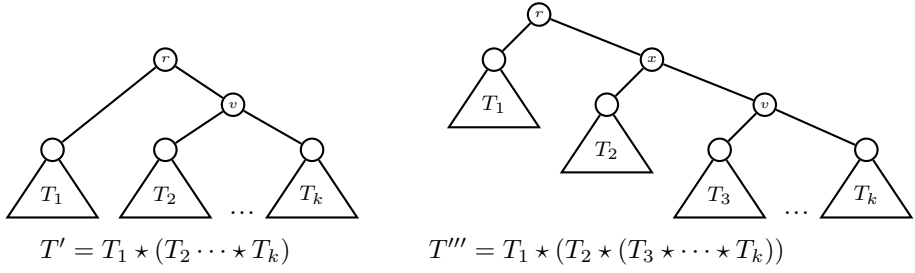
\begin{figure}[htb]
\begin{center}
\begin{tikzpicture}[thick,>=stealth,scale=0.3]
\draw(0,0) node[tre] (z1) {}; 
\draw (z1)--(-2,-3)--(2,-3)--(z1);
\draw(0,-2) node  {\footnotesize $T_1$};
\draw(5,0) node[tre] (z2) {}; 
\draw (z2)--(3,-3)--(7,-3)--(z2);
\draw(5,-2) node  {\footnotesize $T_2$};
\draw(8,-2.8) node {.}; 
\draw(8.4,-2.8) node {.}; 
\draw(8.8,-2.8) node {.}; 
\draw(11.5,0) node[tre] (z3) {}; 
\draw (z3)--(9.5,-3)--(13.5,-3)--(z3);
\draw(11.5,-2) node  {\footnotesize $T_k$};
\draw(8,2) node[tre] (v) {}; \etq v
\draw(5,4) node[tre] (r) {}; \etq r
\draw (5,-4.5) node {$T'=T_1\star(T_2\cdots\star T_k)$};
\draw (r)--(z1);
\draw (r)--(v);
\draw (v)--(z2);
\draw (v)--(z3);
\end{tikzpicture}
\qquad
\begin{tikzpicture}[thick,>=stealth,scale=0.3]
\draw(0,0) node[tre] (z1) {}; 
\draw (z1)--(-2,-3)--(2,-3)--(z1);
\draw(0,-2) node  {\footnotesize $T_1$};
\draw(5,-2) node[tre] (z2) {}; 
\draw (z2)--(3,-5)--(7,-5)--(z2);
\draw(5,-4) node  {\footnotesize $T_2$};
\draw(10,-4) node[tre] (z3) {}; 
\draw (z3)--(8,-7)--(12,-7)--(z3);
\draw(10,-6) node  {\footnotesize $T_3$};

\draw(13,-6.8) node {.}; 
\draw(13.4,-6.8) node {.}; 
\draw(13.8,-6.8) node {.}; 

\draw(16.5,-4) node[tre] (zk) {}; 
\draw (zk)--(14.5,-7)--(18.5,-7)--(zk);
\draw(16.5,-6) node  {\footnotesize $T_k$};

\draw(2,2) node[tre] (r) {}; \etq r
\draw(7,0) node[tre] (x) {}; \etq x
\draw(12,-2) node[tre] (v) {}; \etq v
\draw (r)--(z1);
\draw (r)--(x);
\draw (x)--(z2);
\draw (x)--(v);
\draw (v)--(z3);
\draw (v)--(zk);
\draw (5,-8.5) node {$T'''= T_1\star(T_2\star(T_3\star\cdots\star T_k))$};
\end{tikzpicture}
\end{center}
\caption{\label{fig:lem2mdm} 
The trees $T'$ and $T''$ in Lemma \ref{lem:2mdm}.}
\end{figure}

\begin{proof}
Let $t_i=\delta(T_i)$, for every $i=1,\ldots,k$.

As to (a), the only nodes in $T$ or $T'$ with different ${bal}$ value in both trees are the roots and the new node $v$ in $T'$. Therefore,
$$
\begin{array}{l}
\mathfrak{C}(T')-\mathfrak{C}(T)  \displaystyle ={bal}_{T'}(v)+{bal}_{T'}(r)-{bal}_{T}(r)\\
 \displaystyle\qquad\quad=\MDM(t_2,\ldots,t_k)+\dfrac{1}{2}\Big|\sum\limits_{i=2}^k t_i+f(k-1)-t_1\Big|-\MDM(t_1,\ldots,t_k) \geq \dfrac{1}{2}\Big|\sum\limits_{i=2}^kt_i+f(k-1)-t_1\Big|> 0,
\end{array}
$$
where the first inequality is a consequence of Lemma \ref{lem:mdmodd} and the second inequality is strict because, since $t_1$ is the median of $\{t_1,\ldots,t_k\}$ and $k\geq 3$, there is some $i\geq 2$ such that $t_i\geq t_1$ and hence
$$
\Big|\sum\limits_{i=2}^kt_i+f(k-1)-t_1\Big|=\sum\limits_{i=2}^kt_i+f(k-1)-t_1\geq f(k-1)>0
$$

As far as (b) goes,  the only nodes in $T$ or $T''$ with different ${bal}$ value in both trees are the roots and the new nodes $x,v$ in $T''$. Therefore,
$$
\begin{array}{l}
\displaystyle\mathfrak{C}(T'')-\mathfrak{C}(T)  ={bal}_{T''}(v)+{bal}_{T''}(x)+{bal}_{T''}(r)-{bal}_{T}(r)\\
\displaystyle\quad =\MDM(t_3,\ldots,t_k)+\dfrac{1}{2}\Big|\sum\limits_{i=3}^kt_i+f(k-2)-t_2\Big|+\dfrac{1}{2}\Big|\sum\limits_{i=2}^kt_i+f(k-2)+f(2)-t_1\Big|-\MDM(t_1,\ldots,t_k)\\
\displaystyle\quad \geq \dfrac{1}{2}\Big(\Big|\sum\limits_{i=3}^kt_i+f(k-2)-t_2\Big|+\Big|\sum\limits_{i=2}^kt_i+f(k-2)+f(2)-t_1\Big|\Big)> 0,
\end{array}
$$
where the first inequality is a consequence of Lemma \ref{lem:mdmeven} and the second inequality is strict because, by assumption, $t_1\leq t_2$ and $k\geq 3$, and therefore 
$$
\Big|\sum\limits_{i=2}^kt_i+f(k-2)+f(2)-t_1\Big| = \sum\limits_{i=2}^kt_i+f(k-2)+f(2)-t_1
\geq f(k-2)+f(2)>0
$$
\end{proof}

\begin{lemma}\label{lem:3mdm}
Let $f$ be a mapping $\NN\to \RR_{\geq 0}$ such that  $f(2)>0$.
Consider the trees $T$ and $T'$ depicted in Fig. \ref{fig:lem3mdm}, where, in  both trees, all nodes in the path from $r$ to $x$ are binary, and  $T'$ is obtained from $T$ by simply interchanging the subtrees $T_l$ and $T_{l-1}$. If $\delta(T_l)< \delta(T_{l-1})$ and $\delta(T_l)\leq \delta(T_0)$, then $\mathfrak{C}(T')> \mathfrak{C}(T)$.
\end{lemma}

\begin{figure}[htb]
\begin{center}
\begin{tikzpicture}[thick,>=stealth,scale=0.25]
\draw(8.5,10.5) node[tre] (r) {}; \etq r
\draw(7.5,7.5) node [tre] (v1) {}; 
\draw(11.5,7.5) node[tre] (x2) {}; 
\draw(10.5,4.5) node [tre] (v2) {}; 
\draw (12.8,6.2)--(x2);
\draw (13.7,5.3) node {.};
\draw (13.4,5.6) node {.};
\draw (13.1,5.9) node {.};
\draw(15,4) node[tre] (xl1) {}; 
\draw (xl1) node {\tiny $y$};  
\draw(14,1) node [tre] (vl1) {}; 
\draw (xl1)--(14,5);
\draw(18,1) node[tre] (xl) {}; 
\draw (xl) node {\tiny $x$};  
\draw(17,-2) node [tre] (vl) {}; 
\draw(20,-2) node[tre] (z) {}; 
\draw (z)--(19,-5)--(21,-5)--(z);
\draw(20,-6) node  {$T_0$};
\draw (vl)--(16,-5)--(18,-5)--(vl);
\draw(17,-6) node  {$T_l$};
\draw (vl1)--(13,-2)--(15,-2)--(vl1);
\draw(14,-3) node  {$T_{l-1}$};
\draw (v2)--(9.5,1.5)--(11.5,1.5)--(v2);
\draw(10.5,0.5) node  {$T_2$};
\draw (v1)--(6.5,4.5)--(8.5,4.5)--(v1);
\draw(7.5,3.5) node  {$T_1$};

\draw  (r)--(x2);
\draw  (r)--(v1);
\draw  (x2)--(v2);
\draw  (xl1)--(vl1);
\draw  (xl1)--(xl);
\draw  (xl)--(vl);
\draw  (xl)--(z);
\draw(10.5,10.5) node  {$T$};
\end{tikzpicture}
\qquad
\begin{tikzpicture}[thick,>=stealth,scale=0.25]
\draw(8.5,10.5) node[tre] (r) {}; \etq r
\draw(7.5,7.5) node [tre] (v1) {}; 
\draw(11.5,7.5) node[tre] (x2) {}; 
\draw(10.5,4.5) node [tre] (v2) {}; 
\draw (12.8,6.2)--(x2);
\draw (13.7,5.3) node {.};
\draw (13.4,5.6) node {.};
\draw (13.1,5.9) node {.};
\draw(15,4) node[tre] (xl1) {}; 
\draw (xl1) node {\tiny $y$};  
\draw(14,1) node [tre] (vl1) {}; 
\draw (xl1)--(14,5);
\draw(18,1) node[tre] (xl) {}; 
\draw (xl) node {\tiny $x$};  
\draw(17,-2) node [tre] (vl) {}; 
\draw(20,-2) node[tre] (z) {}; 
\draw (z)--(19,-5)--(21,-5)--(z);
\draw(20,-6) node  {$T_0$};
\draw (vl)--(16,-5)--(18,-5)--(vl);
\draw(17,-6) node  {$T_{l-1}$};
\draw (vl1)--(13,-2)--(15,-2)--(vl1);
\draw(14,-3) node  {$T_{l}$};
\draw (v2)--(9.5,1.5)--(11.5,1.5)--(v2);
\draw(10.5,0.5) node  {$T_2$};
\draw (v1)--(6.5,4.5)--(8.5,4.5)--(v1);
\draw(7.5,3.5) node  {$T_1$};

\draw  (r)--(x2);
\draw  (r)--(v1);
\draw  (x2)--(v2);
\draw  (xl1)--(vl1);
\draw  (xl1)--(xl);
\draw  (xl)--(vl);
\draw  (xl)--(z);
\draw(10.5,10.5) node  {$T'$};
\end{tikzpicture}

\end{center}
\caption{\label{fig:lem3mdm} The trees $T$ and $T'$ in Lemma \ref{lem:3mdm}.}
\end{figure}
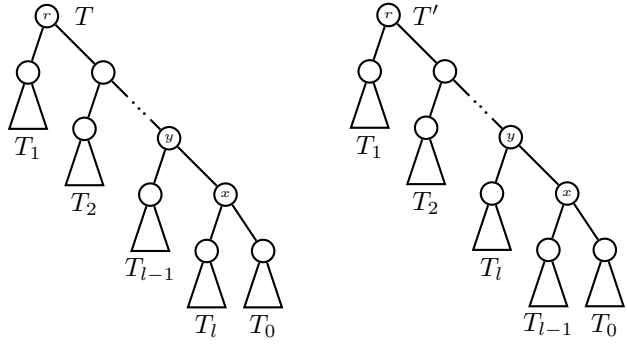

\begin{proof}
Let $t_i=\delta(T_i)$, for every $i=0,\ldots,l$, so that $t_l<t_{l-1}$ and $t_l\leq t_0$.
Since $\delta(T_y)=\delta(T'_y)$, the only nodes in $T$ or $T'$ with different ${bal}$ value in both trees are $x$ and $y$, and therefore,
$$
\begin{array}{l}
\hspace*{-1ex}\mathfrak{C}(T')-\mathfrak{C}(T) ={bal}_{T'}(x)+{bal}_{T'}(y)-{bal}_{T}(x)-{bal}_{T}(y)\\[1ex]
\ =\dfrac{1}{2}|t_{l-1}-t_0|+\dfrac{1}{2}|t_{l-1}+t_0+f(2)-t_l|-\dfrac{1}{2}|t_{l}-t_0|-\dfrac{1}{2}|t_{l}+t_0+f(2)-t_{l-1}|=(*)
\\[1ex]
\end{array}
$$
Now we must distinguish two cases:
\begin{itemize}
\item If $t_l < t_{l-1}\leq t_0$, then
$$
\hspace*{-1ex}\begin{array}{rl}
(*) \hspace*{-1.8ex} & =\dfrac{1}{2}\big((t_0-t_{l-1})+(t_{l-1}+t_0+f(2)-t_l)-(t_0-t_l)-(t_{l}+t_0+f(2)-t_{l-1})\big)\\[2ex]
 & =\dfrac{1}{2}(t_{l-1}-t_l)>0
 \end{array}
$$
\item If $t_l \leq t_0 \leq t_{l-1}$ and $t_l<t_{l-1}$, then
$$
\hspace*{-1ex}\begin{array}{rl}
(*)  \hspace*{-1.8ex}& =\dfrac{1}{2}\big((t_{l-1}-t_0)+(t_{l-1}+t_0+f(2)-t_l)-(t_0-t_l)-|t_{l}+t_0+f(2)-t_{l-1}|\big)\\[2ex]
 & =  \hspace*{-0.5ex}\left\{\begin{array}{l}
 \hspace*{-0.5ex}\dfrac{1}{2}\big(2t_{l-1}-t_0+f(2)-(t_{l}+t_0+f(2)-t_{l-1})\big) = \dfrac{1}{2}(3t_{l-1}-2t_0-t_{l}) \geq \dfrac{1}{2}(t_{l-1}-t_{l})  > 0\\[2ex]
 \qquad\mbox{(if $t_{l}+t_0+f(2)-t_{l-1}\geq 0$)}\\[2ex]
 \hspace*{-0.5ex} \dfrac{1}{2}\big(2t_{l-1}-t_0+f(2)-(t_{l-1}-t_{l}-t_0-f(2))\big) = \dfrac{1}{2}(t_{l-1}+t_{l}+2f(2)) > 0\\[2ex]
 \qquad\mbox{(if $t_{l}+t_0+f(2)-t_{l-1}\leq 0$)}
\end{array}\right.
\end{array}
$$
\end{itemize}
and therefore, in all cases,  $\mathfrak{C}(T')-\mathfrak{C}(T)>0$, as we claimed.
\end{proof}

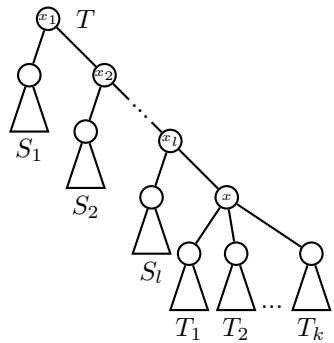
\begin{figure}[htb]
\begin{center}
\begin{tikzpicture}[thick,>=stealth,scale=0.25]
\draw(8.5,10.5) node[tre] (r) {}; 

\draw(7.5,7.5) node [tre] (v1) {}; 
\draw(11.5,7.5) node[tre] (x2) {}; 
\draw (r) node {\tiny $x_1$};  
\draw (x2) node {\tiny $x_2$};  
\draw(10.5,4.5) node [tre] (v2) {}; 
\draw (12.8,6.2)--(x2);
\draw (13.7,5.3) node {.};
\draw (13.4,5.6) node {.};
\draw (13.1,5.9) node {.};
\draw(15,4) node[tre] (xl1) {}; 
\draw (xl1) node {\tiny $x_l$};  
\draw(14,1) node [tre] (vl1) {}; 
\draw (xl1)--(14,5);

\draw(18,1) node[tre] (xl) {}; 
\draw (xl) node {\tiny $x$};  

\draw(16,-2) node [tre] (t1) {}; 
\draw (t1)--(15,-5)--(17,-5)--(t1);
\draw(16,-6) node  {$T_1$};
\draw(18.5,-2) node[tre] (t2) {}; 
\draw (t2)--(17.5,-5)--(19.5,-5)--(t2);
\draw(18.5,-6) node  {$T_2$};
\draw(22.5,-2) node[tre] (tk) {}; 
\draw (tk)--(21.5,-5)--(23.5,-5)--(tk);
\draw(22.5,-6) node  {$T_k$};
\draw(20,-4.8) node {.}; 
\draw(20.4,-4.8) node {.}; 
\draw(20.8,-4.8) node {.}; 

\draw (vl1)--(13,-2)--(15,-2)--(vl1);
\draw(14,-3) node  {$S_l$};
\draw (v2)--(9.5,1.5)--(11.5,1.5)--(v2);
\draw(10.5,0.5) node  {$S_2$};
\draw (v1)--(6.5,4.5)--(8.5,4.5)--(v1);
\draw(7.5,3.5) node  {$S_1$};

\draw  (r)--(x2);
\draw  (r)--(v1);
\draw  (x2)--(v2);
\draw  (xl1)--(vl1);
\draw  (xl1)--(xl);
\draw  (xl)--(t1);
\draw  (xl)--(t2);
\draw  (xl)--(tk);
\draw(10.5,10.5) node  {$T$};
\end{tikzpicture}
\end{center}
\caption{\label{fig:lem4mdm1} The tree $T$ in Lemma \ref{lem:4mdm}.}
\end{figure}

\begin{lemma}\label{lem:4mdm}
Let $f$ be a mapping $\NN\to \RR_{\geq 0}$ such that $0<f(k)< f(k-1)+f(2)$, for every $k\geq 3$, and let $T$ be the tree depicted in Fig. \ref{fig:lem4mdm1}, where $l\geq 1$, $x_1$ is the root, all nodes in the path from $x_1$ to $x_l$ are binary,  and $k\geq 3$.
Assume moreover that $\delta(S_1)\leq \delta(S_2)\leq \cdots \leq \delta(S_l)$.

\begin{enumerate}[(a)]
\item Assume that $k$ is odd and that  $\delta(T_1)$ is the median of 
$\{\delta(T_1),\ldots,\delta(T_k)\}$.

\begin{itemize}
\item[(a.1)] If $\delta(S_l)\leq \delta(T_1)$, then the tree $T'$ depicted in Fig. \ref{fig:lem4mdm}, obtained by pruning the subtree $T_1$ and regrafting it in the arc ending in $x$, satisfies that $\mathfrak{C}(T')> \mathfrak{C}(T)$.

\item[(a.2)] If $\delta(S_l)> \delta(T_1)$, then the tree $T''$ depicted in Fig. \ref{fig:lem4mdm}, obtained by pruning the subtree $T_1$ and regrafting it in the arc ending in $x_l$, satisfies that $\mathfrak{C}(T'')> \mathfrak{C}(T)$.
\end{itemize}

\item Assume that $k$ is even and that $\delta(T_1),\delta(T_2)$ are the middle values of 
$\{\delta(T_1),\ldots,\delta(T_k)\}$.

\begin{itemize}
\item[(b.1)] If $\delta(S_l)\leq \delta(T_1\star T_2)$, then the tree $T'$ depicted in Fig. \ref{fig:lem41mdm}, obtained by pruning the subtrees $T_1$ and $T_2$ and then inserting  $T_1\star T_2$ in the arc ending in $x$, satisfies that $\mathfrak{C}(T')> \mathfrak{C}(T)$.

\item[(b.2)] If $\delta(S_l)> \delta(T_1\star T_2)$, then the tree $T''$ depicted in Fig. \ref{fig:lem41mdm}, obtained by pruning the subtrees $T_1$ and $T_2$ and then inserting $T_1\star T_2$ in the arc ending in $x_l$, satisfies that $\mathfrak{C}(T'')> \mathfrak{C}(T)$.
\end{itemize}\end{enumerate}
\end{lemma}

\begin{figure}[htb]
\begin{center}
\begin{tikzpicture}[thick,>=stealth,scale=0.25]
\draw(8.5,10.5) node[tre] (r) {}; 
\draw(7.5,7.5) node [tre] (v1) {}; 
\draw(11.5,7.5) node[tre] (x2) {}; 
\draw (r) node {\tiny $x_1$};  
\draw (x2) node {\tiny $x_2$};  

\draw(10.5,4.5) node [tre] (v2) {}; 
\draw (12.8,6.2)--(x2);
\draw (13.7,5.3) node {.};
\draw (13.4,5.6) node {.};
\draw (13.1,5.9) node {.};
\draw(15,4) node[tre] (xl1) {}; 
\draw (xl1) node {\tiny $x_l$};  
\draw(18,1) node[tre] (y) {}; 
\draw (y) node {\tiny $y$};  
\draw(14,1) node [tre] (vl1) {}; 
\draw (xl1)--(14,5);
\draw(21,-2) node[tre] (xl) {}; 
\draw (xl) node {\tiny $x$};

\draw(17,-2) node [tre] (t1) {}; 
\draw (t1)--(16,-5)--(18,-5)--(t1);
\draw(17,-6) node  {$T_1$};

\draw(20,-4) node[tre] (t2) {}; 
\draw (t2)--(19,-7)--(21,-7)--(t2);
\draw(20,-8) node  {$T_2$};
\draw(24,-4) node[tre] (tk) {}; 
\draw (tk)--(23,-7)--(25,-7)--(tk);
\draw(24,-8) node  {$T_k$};
\draw(21.5,-6.8) node {.}; 
\draw(21.9,-6.8) node {.}; 
\draw(22.3,-6.8) node {.}; 

\draw (vl1)--(13,-2)--(15,-2)--(vl1);
\draw(14,-3) node  {$S_l$};
\draw (v2)--(9.5,1.5)--(11.5,1.5)--(v2);
\draw(10.5,0.5) node  {$S_2$};
\draw (v1)--(6.5,4.5)--(8.5,4.5)--(v1);
\draw(7.5,3.5) node  {$S_1$};

\draw  (r)--(x2);
\draw  (r)--(v1);
\draw  (x2)--(v2);
\draw  (xl1)--(vl1);
\draw  (xl1)--(y);
\draw  (xl)--(y);
\draw  (y)--(t1);
\draw  (xl)--(t2);
\draw  (xl)--(tk);
\draw(10.5,10.5) node  {$T'$};
\end{tikzpicture}
\
\begin{tikzpicture}[thick,>=stealth,scale=0.25]
\draw(8.5,10.5) node[tre] (r) {}; 

\draw(7.5,7.5) node [tre] (v1) {}; 
\draw(11.5,7.5) node[tre] (x2) {}; 
\draw (r) node {\tiny $x_1$};  
\draw (x2) node {\tiny $x_2$};  

\draw(10.5,4.5) node [tre] (v2) {}; 
\draw (12.8,6.2)--(x2);
\draw (13.7,5.3) node {.};
\draw (13.4,5.6) node {.};
\draw (13.1,5.9) node {.};
\draw(15,4) node[tre] (xl1) {}; 
\draw (xl1) node {\tiny $y$};  
\draw(18,1) node[tre] (y) {}; 
\draw (y) node {\tiny $x_l$};  
\draw(14,1) node [tre] (vl1) {}; 
\draw (xl1)--(14,5);
\draw(21,-2) node[tre] (xl) {}; 
\draw (xl) node {\tiny $x$};

\draw(17,-2) node [tre] (t1) {}; 
\draw (t1)--(16,-5)--(18,-5)--(t1);
\draw(17,-6) node  {$S_l$};

\draw(20,-4) node[tre] (t2) {}; 
\draw (t2)--(19,-7)--(21,-7)--(t2);
\draw(20,-8) node  {$T_2$};
\draw(24,-4) node[tre] (tk) {}; 
\draw (tk)--(23,-7)--(25,-7)--(tk);
\draw(24,-8) node  {$T_k$};
\draw(21.5,-6.8) node {.}; 
\draw(21.9,-6.8) node {.}; 
\draw(22.3,-6.8) node {.}; 

\draw (vl1)--(13,-2)--(15,-2)--(vl1);
\draw(14,-3) node  {$T_1$};
\draw (v2)--(9.5,1.5)--(11.5,1.5)--(v2);
\draw(10.5,0.5) node  {$S_2$};
\draw (v1)--(6.5,4.5)--(8.5,4.5)--(v1);
\draw(7.5,3.5) node  {$S_1$};

\draw  (r)--(x2);
\draw  (r)--(v1);
\draw  (x2)--(v2);
\draw  (xl1)--(vl1);
\draw  (xl1)--(y);
\draw  (xl)--(y);
\draw  (y)--(t1);
\draw  (xl)--(t2);
\draw  (xl)--(tk);
\draw(10.5,10.5) node  {$T''$};
\end{tikzpicture}
\end{center}
\caption{\label{fig:lem4mdm} The trees $T',T''$ in Lemma \ref{lem:4mdm}.(a).}
\end{figure}
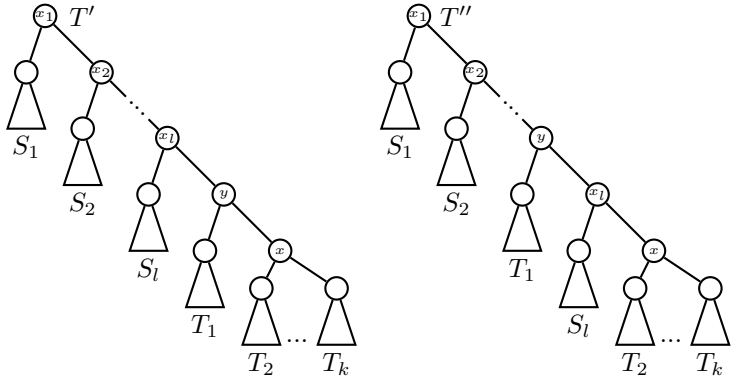

\begin{proof}
For every $i=1,\ldots,l$, let $x_i$ denote the parent of the root of $S_i$ in all trees in the statement.
Moreover, for every $i=1,\ldots,l$, let $s_i=\delta(S_i)$ and, for every $i=1,\ldots,k$, let $\delta(T_i)=t_i$, and let $t=\sum\limits_{i=1}^k t_i$. Recall that we are assuming throughout this proof that $s_1\leq\cdots\leq s_l$.
\medskip

\noindent\emph{(a)} Assume that $k\geq 3$ is odd and that
   $\MDM(t_1,\ldots,t_k)\leq$ $\MDM(t_2,\ldots,t_k)$.

As far as assertion (a.1) goes, let us assume that $s_l\leq t_1$ and therefore that $s_i\leq t$ for every $i=1,\ldots,l$. In this case, 
the only nodes in $T$ or $T'$ with different ${bal}$ value in both trees are $x_1,\ldots,x_l$, $x$ and the new node $y$ in $T'$. Then:
$$
\begin{array}{l}
\displaystyle \mathfrak{C}(T')-\mathfrak{C}(T) ={bal}_{T'}(x)-{bal}_{T}(x)+{bal}_{T'}(y)+\sum\limits _{i=1}^l({bal}_{T'}(x_i)-{bal}_{T}(x_i))
\\[2ex]
\displaystyle \qquad =\MDM(t_2,\ldots,t_k)-\MDM(t_1,\ldots,t_k)+\dfrac{1}{2}\Big|\sum\limits_{i=2}^k t_i+f(k-1)-t_1\Big| \\
\displaystyle \qquad\qquad +\dfrac{1}{2}\sum\limits _{i=1}^l\Big(\Big|t+\sum\limits_{j=i+1}^l s_j+f(k-1)+(l-i+1)f(2)-s_i\Big|-\Big|t+\sum\limits_{j=i+1}^l s_j+f(k)+(l-i)f(2)-s_i\Big|\Big)\\[2ex]
\displaystyle  \qquad
\geq \dfrac{1}{2}\sum\limits _{i=1}^l\Big(\Big|t+\sum\limits_{j=i+1}^l s_j+f(k-1)+(l-i+1)f(2)-s_i\Big|-\Big|t+\sum\limits_{j=i+1}^l s_j+f(k)+(l-i)f(2)-s_i\Big|\Big)\\[2ex]
\displaystyle  \qquad= \dfrac{1}{2}\sum\limits _{i=1}^l\left(\Big(t+\sum\limits_{j=i+1}^l s_j+f(k-1)+(l-i+1)f(2)-s_i\Big)-\Big(t+\sum\limits_{j=i+1}^l s_j+f(k)+(l-i)f(2)-s_i\Big)\right)\\[2ex]
\qquad\mbox{(because $s_i\leq t$ for every $i=1,\ldots,l$)}\\[1ex]
\displaystyle \qquad=\frac{l}{2}(f(k-1)+f(2)-f(k))> 0
\end{array}
$$

As far as assertion (a.2) goes, let us assume that $s_l> t_1$. Again, the only nodes in $T$ or $T''$ with different ${bal}$ value in both trees are $x_1,\ldots,x_l$, $x$ and the new node $y$. Therefore:

$$
\begin{array}{l}
\displaystyle \mathfrak{C}(T'')-\mathfrak{C}(T) ={bal}_{T''}(x)-{bal}_{T}(x)+{bal}_{T''}(y)+{bal}_{T''}(x_l)-{bal}_{T}(x_l) +\sum\limits _{i=1}^{l-1}({bal}_{T''}(x_i)-{bal}_{T}(x_i))
\\[2ex]
\displaystyle \quad =\MDM(t_2,\ldots,t_k)-\MDM(t_1,\ldots,t_k)+\dfrac{1}{2}\Big|\sum\limits _{i=2}^k t_i+s_l+f(k-1)+f(2)-t_1\Big|\\[1ex]
\displaystyle \quad\quad +\dfrac{1}{2}\Big|\sum\limits _{i=2}^k t_i+f(k-1)-s_l\Big|-\dfrac{1}{2}|t+f(k)-s_l|
\\[1ex]
\displaystyle  \quad \quad+\dfrac{1}{2}\sum\limits _{i=1}^{l-1}\Big(\Big|t+\sum\limits_{j=i+1}^l s_j+f(k-1)+(l-i+1)f(2)-s_i\Big|-\Big|t+\sum\limits_{j=i+1}^l s_j+f(k)+(l-i)f(2)-s_i\Big|\Big)\\[2ex]
\displaystyle  \quad\geq \dfrac{1}{2}\Big(\Big|\sum\limits_{i=2}^k t_i+s_l+f(k-1)+f(2)-t_1\Big|-|t+f(k)-s_l|\Big)\\[1ex]
\displaystyle  \quad\quad+ \dfrac{1}{2}\sum\limits _{i=1}^{l-1}\Big(\Big|t+\sum\limits_{j=i+1}^l s_j+f(k-1)+(l-i+1)f(2)-s_i\Big|-\Big|t+\sum\limits_{j=i+1}^l s_j+f(k)+(l-i)f(2)-s_i\Big|\Big)\\[2ex]
\displaystyle  \quad= \dfrac{1}{2}\Big(\sum\limits_{i=2}^k t_i+s_l+f(k-1)+f(2)-t_1-|t+f(k)-s_l|\Big)\\[1ex]
\displaystyle  \quad\quad+ \dfrac{1}{2}\sum\limits _{i=1}^{l-1}\Big(\Big(t+\sum\limits_{j=i+1}^l s_j+f(k-1)+(l-i+1)f(2)-s_i\Big)-\Big(t+\sum\limits_{j=i+1}^l s_j+f(k)+(l-i)f(2)-s_i\Big)\Big)\\[2ex]
\quad\mbox{(because $s_i\leq s_l$, for every $i=1,\ldots,l$, and $s_l>t_1$)%
}\\[2ex] 
\displaystyle \quad = \dfrac{1}{2}\Big(\sum\limits_{i=2}^k t_i+s_l+f(k-1)+f(2)-t_1-|t+f(k)-s_l|\Big)+\frac{l-1}{2}(f(k-1)+f(2)-f(k))\\[1ex]
\displaystyle  \quad \geq \dfrac{1}{2}\Big(\sum\limits_{i=2}^k t_i+s_l+f(k-1)+f(2)-t_1-|t+f(k)-s_l|\Big)\\[3ex]
\displaystyle  \quad =\left\{\begin{array}{ll}
\dfrac{1}{2}(2(s_l-t_1)+f(k-1)+f(2)-f(k))>0 & \mbox{\quad (if $s_l\leq t+f(k)$)}\\[2ex]
\displaystyle\frac{1}{2}\Big(2\sum\limits _{i=2}^k t_i+f(k-1)+f(2)+f(k)\Big)>0& \mbox{\quad (if $s_l\geq t+f(k)$)}
\end{array}
\right.
\end{array}
$$
\medskip

\begin{figure}[htb]
\begin{center}
\begin{tikzpicture}[thick,>=stealth,scale=0.25]
\draw(8.5,10.5) node[tre] (r) {}; 
\draw(7.5,7.5) node [tre] (v1) {}; 
\draw(11.5,7.5) node[tre] (x2) {}; 
\draw (r) node {\tiny $x_1$};  
\draw (x2) node {\tiny $x_2$};  

\draw(10.5,4.5) node [tre] (v2) {}; 
\draw (12.8,6.2)--(x2);
\draw (13.7,5.3) node {.};
\draw (13.4,5.6) node {.};
\draw (13.1,5.9) node {.};
\draw(15,4) node[tre] (xl1) {}; 
\draw (xl1) node {\tiny $x_l$};  
\draw (xl1)--(14,5);
\draw(14,1) node [tre] (vl1) {}; 
\draw(18,1) node[tre] (y) {}; 
\draw (y) node {\tiny $y$};  
\draw(17,-2) node[tre] (z) {}; 
\draw (z) node {\tiny $z$};

\draw(15.5,-5) node [tre] (t1) {}; 
\draw (t1)--(14.5,-8)--(16.5,-8)--(t1);
\draw(15.5,-9) node  {$T_1$};

\draw(18.5,-5) node[tre] (t2) {}; 
\draw (t2)--(17.5,-8)--(19.5,-8)--(t2);
\draw(18.5,-9) node  {$T_2$};

\draw(23,-2) node[tre] (xl) {}; 
\draw (xl) node {\tiny $x$};  

\draw(22,-5) node[tre] (t3) {}; 
\draw (t3)--(21,-8)--(23,-8)--(t3);
\draw(22,-9) node  {$T_3$};

\draw(26,-5) node[tre] (tk) {}; 
\draw (tk)--(25,-8)--(27,-8)--(tk);
\draw(26,-9) node  {$T_k$};
\draw(23.6,-7.8) node {.}; 
\draw(24,-7.8) node {.}; 
\draw(24.4,-7.8) node {.}; 

\draw (vl1)--(13,-2)--(15,-2)--(vl1);
\draw(14,-3) node  {$S_l$};
\draw (v2)--(9.5,1.5)--(11.5,1.5)--(v2);
\draw(10.5,0.5) node  {$S_2$};
\draw (v1)--(6.5,4.5)--(8.5,4.5)--(v1);
\draw(7.5,3.5) node  {$S_1$};

\draw  (r)--(x2);
\draw  (r)--(v1);
\draw  (x2)--(v2);
\draw  (xl1)--(vl1);
\draw  (xl1)--(y);
\draw  (xl)--(y);
\draw  (y)--(z);
\draw  (z)--(t1);
\draw  (z)--(t2);
\draw  (xl)--(t3);
\draw  (xl)--(tk);
\draw(10.5,10.5) node  {$T'$};
\end{tikzpicture}
\quad
\begin{tikzpicture}[thick,>=stealth,scale=0.25]
\draw(8.5,10.5) node[tre] (r) {}; 
\draw(7.5,7.5) node [tre] (v1) {}; 
\draw(11.5,7.5) node[tre] (x2) {}; 
\draw (r) node {\tiny $x_1$};  
\draw (x2) node {\tiny $x_2$};  

\draw(10.5,4.5) node [tre] (v2) {}; 
\draw (12.8,6.2)--(x2);
\draw (13.7,5.3) node {.};
\draw (13.4,5.6) node {.};
\draw (13.1,5.9) node {.};
\draw(15,4) node[tre] (xl1) {}; 
\draw (xl1) node {\tiny $y$};  
\draw (xl1)--(14,5);
\draw(18,1) node[tre] (y) {}; 
\draw (y) node {\tiny $x_l$};  

\draw(18,-2) node [tre] (vl1) {};

\draw(14,1) node[tre] (z) {}; 
\draw (z) node {\tiny $z$};

\draw(12,-2) node [tre] (t1) {}; 
\draw (t1)--(11,-5)--(13,-5)--(t1);
\draw(12,-6) node  {$T_1$};

\draw(15,-2) node[tre] (t2) {}; 
\draw (t2)--(14,-5)--(16,-5)--(t2);
\draw(15,-6) node  {$T_2$};

\draw (vl1)--(17,-5)--(19,-5)--(vl1);
\draw(18,-6) node  {$S_l$};

\draw(23,-2) node[tre] (xl) {}; 
\draw (xl) node {\tiny $x$};  

\draw(22,-5) node[tre] (t3) {}; 
\draw (t3)--(21,-8)--(23,-8)--(t3);
\draw(22,-9) node  {$T_3$};

\draw(26,-5) node[tre] (tk) {}; 
\draw (tk)--(25,-8)--(27,-8)--(tk);
\draw(26,-9) node  {$T_k$};
\draw(23.6,-7.8) node {.}; 
\draw(24,-7.8) node {.}; 
\draw(24.4,-7.8) node {.}; 

\draw (v2)--(9.5,1.5)--(11.5,1.5)--(v2);
\draw(10.5,0.5) node  {$S_2$};
\draw (v1)--(6.5,4.5)--(8.5,4.5)--(v1);
\draw(7.5,3.5) node  {$S_1$};

\draw  (r)--(x2);
\draw  (r)--(v1);
\draw  (x2)--(v2);
\draw  (xl1)--(z);
\draw  (xl1)--(y);
\draw  (xl)--(y);
\draw  (y)--(vl1);
\draw  (z)--(t1);
\draw  (z)--(t2);
\draw  (xl)--(t3);
\draw  (xl)--(tk);
\draw(10.5,10.5) node  {$T''$};
\end{tikzpicture}
\end{center}
\caption{\label{fig:lem41mdm} The trees  in Lemma \ref{lem:4mdm}.(b).}
\end{figure}
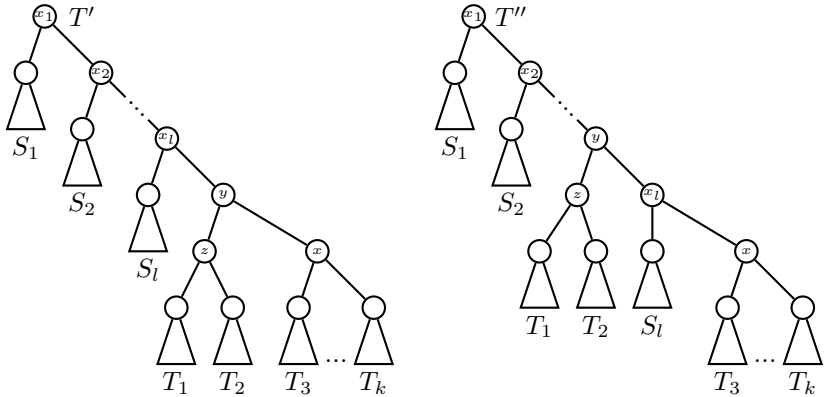

\noindent \emph{(b)} Assume now that $k\geq 3$ is even, and hence $k\geq 4$, and that  $\MDM(t_1,\ldots,t_k)\leq$ $\MDM(t_3,\ldots,t_k)$.

As far as assertion (b.1) goes, let us assume that $s_l\leq t_1+t_2+f(2)\leq t+f(2)$.
The only nodes in $T$ or $T'$ with different ${bal}$ value in both trees are $x_1,\ldots,x_l$, $x$ and the new nodes $y$ and $z$ in $T'$. Therefore:
$$
\begin{array}{l}
\displaystyle \mathfrak{C}(T')-\mathfrak{C}(T)  ={bal}_{T'}(x)-{bal}_{T}(x)+{bal}_{T'}(z)+{bal}_{T'}(y)  +{bal}_{T'}(x_l)-{bal}_{T}(x_l)+\sum\limits _{i=1}^{l-1}({bal}_{T'}(x_i)-{bal}_{T}(x_i))
\\[2ex]
\displaystyle\qquad =\MDM(t_3,\ldots,t_k)-\MDM(t_1,\ldots,t_k)+\dfrac{1}{2}|t_2-t_1|+\dfrac{1}{2}\Big|\sum\limits_{i=3}^k t_i+f(k-2)-(t_2+t_1+f(2))\Big|\\[1ex]
\displaystyle\qquad\qquad +\dfrac{1}{2}|t+f(k-2)+2f(2)-s_l| -\dfrac{1}{2}|t+f(k)-s_l|\\[1ex]
\displaystyle \qquad\qquad +\dfrac{1}{2}\sum\limits _{i=1}^{l-1}\Big(\Big|t+\sum\limits_{j=i+1}^l s_j+f(k-2)+(l-i+2)f(2)-s_i\Big|-\Big|t+\sum\limits_{j=i+1}^l s_j+f(k)+(l-i)f(2)-s_i\Big|\Big)\\[2ex]
\displaystyle \qquad\geq \dfrac{1}{2}\big(t+f(k-2)+2f(2)-s_l-|t+f(k)-s_l|\big)\\
\displaystyle \qquad\qquad +\dfrac{1}{2}\sum\limits _{i=1}^{l-1}\Big(t+\sum\limits_{j=i+1}^l s_j+f(k-2)+(l-i+2)f(2)-s_i-\Big(t+\sum\limits_{j=i+1}^l s_j+f(k)+(l-i)f(2)-s_i\Big)\Big)\\[2ex]
\qquad \mbox{(because $s_l\leq t+f(2)$ and $s_i\leq s_l$, for every $i=1,\ldots,l-1$)}\\[1ex]
\displaystyle\qquad =\dfrac{1}{2}\big(t+f(k-2)+2f(2)-s_l-|t+f(k)-s_l|\big)+\frac{l-1}{2}(f(k-2)+2f(2)-f(k))\\[2ex]
\displaystyle\qquad \geq \dfrac{1}{2}\big(t+f(k-2)+2f(2)-s_l-|t+f(k)-s_l|\big)+\frac{l-1}{2}(f(k-1)+f(2)-f(k))\\[2ex]
\displaystyle\qquad \geq \dfrac{1}{2}\big(t+f(k-2)+2f(2)-s_l-|t+f(k)-s_l|\big)\\[2ex]
\qquad=\left\{\begin{array}{ll}
\dfrac{1}{2}(f(k-2)+2f(2)-f(k))> 0&\mbox{\quad  if $s_l\leq t+f(k)$}\\[2ex]
\dfrac{1}{2}(2(t+f(2)-s_l)+f(k-2)+f(k))> 0&\mbox{\quad  if $s_l\geq t+f(k)$}
\end{array}\right.
\end{array}
$$

As far as assertion (b.2) goes, let us assume now that $s_l> t_1+t_2+f(2)$.
Again, the only nodes in $T$ or $T''$ with different ${bal}$ value in both trees are $x_1,\ldots,x_l$, $x$ and the new nodes $y,z$. Therefore:

$$
\begin{array}{l}
\displaystyle \mathfrak{C}(T'')-\mathfrak{C}(T)  ={bal}_{T''}(x)-{bal}_{T}(x)+{bal}_{T''}(z)+{bal}_{T''}(x_l)-{bal}_{T}(x_l)+{bal}_{T''}(y)+\sum\limits _{i=1}^{l-1}({bal}_{T''}(x_i)-{bal}_{T}(x_i))
\\[2ex]
\displaystyle  \qquad =\MDM(t_3,\ldots,t_k)-\MDM(t_1,\ldots,t_k)+\dfrac{1}{2}|t_2 -t_1|+\dfrac{1}{2}\Big|\sum\limits _{i=3}^k t_i+f(k-2)-s_l\Big| 
 \\
\displaystyle  \qquad\qquad -\dfrac{1}{2}|t+f(k)-s_l|+\dfrac{1}{2}\Big|\sum\limits _{i=3}^k t_i+s_l+f(k-2)+f(2)-(t_1+t_2+f(2))\Big|
\\[1ex]
\displaystyle  \qquad\qquad +\dfrac{1}{2}\sum\limits _{i=1}^{l-1}\Big(\Big|t+\sum\limits_{j=i+1}^l s_j+f(k-2)+(l-i+2)f(2)-s_i\Big|-\Big|t+\sum\limits_{j=i+1}^l s_j+f(k)+(l-i)f(2)-s_i\Big|\Big)\\[2ex]
\displaystyle  \qquad \geq \dfrac{1}{2}\Big(\sum\limits _{i=3}^k t_i+s_l+f(k-2)+f(2)-(t_1+t_2+f(2))-|t+f(k)-s_l|\Big)\\[1ex]
\displaystyle  \qquad\qquad +\dfrac{1}{2}\sum\limits _{i=1}^{l-1}\Big( t+\sum\limits_{j=i+1}^l s_j+f(k-2)+(l-i+2)f(2)-s_i -\Big(t+\sum\limits_{j=i+1}^l s_j+f(k)+(l-i)f(2)-s_i\Big)\Big)\\[1ex]
\end{array}
$$

$$\begin{array}{l}
\displaystyle \qquad= \dfrac{1}{2}\Big( \sum\limits _{i=3}^k t_i+s_l+f(k-2)-(t_1+t_2) -|t+f(k)-s_l|\Big) +\frac{l-1}{2}(f(k-2)+2f(2)-f(k))\\[1ex]
\displaystyle \qquad\geq \dfrac{1}{2}\Big( \sum\limits _{i=3}^k t_i+s_l+f(k-2)-(t_1+t_2) -|t+f(k)-s_l|\Big) +\frac{l-1}{2}(f(k-1)+f(2)-f(k))\\[1ex]
\displaystyle \qquad\geq \dfrac{1}{2}\Big( \sum\limits _{i=3}^k t_i+s_l+f(k-2)-(t_1+t_2) -|t+f(k)-s_l|\Big)\\[3ex]
\displaystyle\qquad=\left\{\begin{array}{ll}
\displaystyle \frac{1}{2}\big(2s_l-2(t_1+t_2)+f(k-2)-f(k)\big)> \frac{1}{2}(2f(2)+f(k-2)-f(k))>0 & \mbox{\ if $s_l\leq t+f(k)$}\\[1ex]
\displaystyle \frac{1}{2}\Big(2\sum\limits _{i=3}^k t_i+f(k-2)+f(k)\Big)>0 & \mbox{\ if $s_l\geq t+f(k)$}
 \end{array}\right.
\end{array}
$$
\end{proof}

\begin{corollary}\label{cor:basic1mdm}
For every non-binary  tree $T\in \TT^*_n$, there always exists a binary tree $T'\in \TT^*_n$ such that $\mathfrak{C}(T')> \mathfrak{C}(T)$.
\end{corollary}

\begin{proof}
We shall prove by complete induction on the sum $S$ of the degrees of the non-binary internal nodes in a  tree $T\in \TT_n^*$
that there always exists a binary  tree $T'\in \TT_n^*$ such that $\mathfrak{C}(T')\geq \mathfrak{C}(T)$.
Moreover, it will be clear from the proof that if $T$ isn't binary, then $T'$ can  be chosen so that this inequality is strict.

The assertion to be proved by induction is obviously true if $S=0$ (which means that $T$ is binary), so assume that $S>0$.
Let $x$ be an internal non-binary node of $T$
such that all nodes in the path from the root $r$ to $x$, except $x$ itself, are binary.

If $x$ is the root, then we apply Lemma \ref{lem:2mdm} and we obtain a tree $T_0$ with smaller $S$ and  larger $\mathfrak{C}$.
Then, by induction, there exists a binary tree $T'$ such that $\mathfrak{C}(T)<\mathfrak{C}(T_0)\leq \mathfrak{C}(T')$.

If $x$ is not the root $r$, let $r,x_2,\ldots,x_l,x$  be the path from the root to $x$: all these nodes except $x$ are binary. By Lemma \ref{lem:3mdm}, if we rearrange the subtrees rooted at the children of the nodes $r,x_2,\ldots,x_l$ in increasing order of their $\delta$-sizes, the $\mathfrak{C}$ value of the resulting tree increases without modifying the value of $S$; let $\widehat{T}$ be the tree obtained in this way. Next, by Lemma \ref{lem:4mdm}, in $\widehat{T}$ we can either prune a subtree $T_1$ rooted at one child of $x$ and regraft it to an arc in the path from $r$ to $x$ (adding a new binary node to the tree),
or we can prune two subtrees $T_1,T_2$ rooted at two children of $x$ and regraft their star $T_1\star T_2$ to an arc in the path from $r$ to $x$ (adding two new binary nodes to the tree), in both cases in such a way that the resulting tree $T_0$ has a larger $\mathfrak{C}$ and a smaller $S$. 
Then, by induction, there exists a binary tree $T'$ such that $\mathfrak{C}(T)\leq\mathfrak{C}(\widehat{T})<\mathfrak{C}(T_0)\leq \mathfrak{C}(T')$. 

This finishes the proof by induction.
\end{proof}

Therefore, the maximum $\mathfrak{C}$ value on $\TT^*_n$ is reached at a binary tree, where, by Proposition 6 in the main text,  it is equal to $(f(0)+f(2))/2$ times the Colless index, with $f(0)+f(2)\geq f(2)>0$. 
Then, since, by Lemma 1 in the main text, the maximum Colless index of a binary tree with $n$ leaves is reached exactly at the comb $K_n$, the same is true for $\mathfrak{C}$. So, the maximum value of $\mathfrak{C}$ on $\TT_n^*$ is reached exactly at $K_n$, and it is
$$
\mathfrak{C}(K_n)=\dfrac{f(0)+f(2)}{2}C(K_n)=\dfrac{f(0)+f(2)}{4} (n-1)(n-2).
$$

\section{Proof of the thesis of Theorem 18 for $\mathfrak{C}_{\sd,f}$}

\noindent The proof of Theorem 18 for $D=\sd$, the sample standard deviation, is very similar to the one provided for $D=\MDM$  in the previous section, but simpler, because Lemmas \ref{lem:mdmodd} and \ref{lem:mdmeven} are replaced by Lemma \ref{lem:var} below, which guarantees that it is always enough to remove a suitable element in a non-constant numeric vector of length at least 3, in order to increase its variance.  To simplify the notations, we shall denote in this section $\delta_f$ and $\mathfrak{C}_{\sd,f}$ by $\delta$ and $\mathfrak{C}$, respectively, and we shall denote ${bal}_{\sd,f}$ on a tree $T$ by ${bal}_T$ or simply by ${bal}$ when it is not necessary to specify the tree.

\begin{lemma}\label{lem:var}
For every $(x_1,\ldots,x_n)\in \RR^n$ with $n\geq 3$,  if $x_i$ is the value in the set $\{x_1,\ldots,x_n\}$ closest to its mean,
then 
$$
\var (x_1,\ldots,x_n)\leq \var(x_1,\ldots,x_{i-1},x_{i+1},\ldots,x_n),
$$
and thus, taking positive square roots,
$$
\sd (x_1,\ldots,x_n)\leq \sd(x_1,\ldots,x_{i-1},x_{i+1},\ldots,x_n).
$$
Moreover, these inequalities are strict unless either $x_1=\cdots=x_n$  or $n$ is even and $\{x_1,\ldots,x_n\}$ consists of $n/2$ copies of two different  elements.
\end{lemma}

\begin{proof}
Let $\overline{x}=(x_1+\cdots+x_n)/n$ and, after rearranging $x_1,\ldots, x_n$ if necessary, assume that $(x_n-\overline{x})^2\leq (x_i-\overline{x})^2$, for every $i=1,\ldots,n-1$. We shall prove that 
$$
\var (x_1,\ldots,x_n)\leq\var (x_1,\ldots,x_{n-1}).
$$
Indeed, let $\overline{x}'=(x_1+\cdots+x_{n-1})/(n-1)$. Then
$$
\var(x_1,\ldots,x_{n-1})\geq \var(x_1,\ldots,x_{n})\Longleftrightarrow
(n-1)\sum_{i=1}^{n-1} (x_i-\overline{x}')^2\geq (n-2) \sum_{i=1}^{n} \Big(x_i-\overline{x})^2
$$
Now
$$
\begin{array}{l}
\displaystyle (n-1)\sum_{i=1}^{n-1} (x_i-\overline{x}')^2 =
(n-1)\sum_{i=1}^{n-1} \Big(x_i-\overline{x}+\frac{1}{n-1}(x_n-\overline{x})\Big)^2\\
\qquad\qquad  \displaystyle =
(n-1)\sum_{i=1}^{n-1} \Big((x_i-\overline{x})^2+\frac{2}{n-1}(x_i-\overline{x}) (x_n-\overline{x})+\frac{1}{(n-1)^2}(x_n-\overline{x})^2\Big)\\
\qquad\qquad  \displaystyle =
(n-1)\sum_{i=1}^{n-1} (x_i-\overline{x})^2+2(x_n-\overline{x}) \sum_{i=1}^{n-1} (x_i-\overline{x})+(x_n-\overline{x})^2\\
\qquad\qquad  \displaystyle =
(n-1)\sum_{i=1}^{n-1} (x_i-\overline{x})^2+2(x_n-\overline{x})(\overline{x}-x_n)+(x_n-\overline{x})^2\\
\qquad\qquad  \displaystyle =
(n-1)\sum_{i=1}^{n-1} (x_i-\overline{x})^2-(x_n-\overline{x})^2\\
\qquad\qquad  \displaystyle =
(n-2)\sum_{i=1}^{n-1} (x_i-\overline{x})^2+\sum_{i=1}^{n-1} (x_i-\overline{x})^2-(x_n-\overline{x})^2\\
\qquad\qquad  \displaystyle \geq 
(n-2)\sum_{i=1}^{n-1} (x_i-\overline{x})^2+(n-1) (x_n-\overline{x})^2-(x_n-\overline{x})^2 = 
(n-2)\sum_{i=1}^{n} (x_i-\overline{x})^2
\end{array}
$$
as we wanted to prove. This inequality is an equality if, and only if, $(x_i-\overline{x})^2=(x_n-\overline{x})^2$ for every $i=1,\ldots,n-1$,
and it is easy to check that this condition holds exactly when either $x_1=\cdots=x_n$  or $n$ is even and $\{x_1,\ldots,x_n\}$ consists of $n/2$ copies of two different  elements.
\end{proof}

We prove now a series of lemmas that play in this proof the same role as Lemmas \ref{lem:2mdm} to \ref{lem:4mdm} in the last section. 

\begin{lemma}\label{lem:2sd}
Let $f$ be a mapping $\NN\to \RR_{\geq 0}$ such that $f(k)>0$, for every $k\geq 2$, and
let $T$ be a tree of the form $T_1\star\cdots\star T_k$, with $k\geq 3$.
If  $\delta(T_1)$ is the value in the set $\{\delta(T_1),\ldots,\delta(T_k)\}$ closest to its mean, 
then taking  the tree $T'=T_1\star(T_2\star \cdots \star T_k)$  depicted in the left-hand side of Fig. \ref{fig:lem2mdm}, we obtain that $\mathfrak{C}(T')> \mathfrak{C}(T)$.
\end{lemma}

\begin{proof}
Let $t_i=\delta(T_i)$, for every $i=1,\ldots,k$.
The only nodes in $T$ or $T'$ with different ${bal}$ value in both trees are the roots and the new node $v$ in $T'$. Therefore,
$$
\begin{array}{l}
\mathfrak{C}(T')-\mathfrak{C}(T)  \displaystyle={bal}_{T'}(r)+{bal}_{T'}(v)-{bal}_{T}(r)\\
\qquad \displaystyle=\dfrac{1}{\sqrt{2}}\Big|\sum\limits_{i=2}^kt_i+f(k-1)-t_1\Big|+\sd(t_2,\ldots,t_k)-\sd(t_1,\ldots,t_k)\displaystyle\geq \dfrac{1}{\sqrt{2}}\Big|\sum\limits_{i=2}^kt_i+f(k-1)-t_1\Big|\geq 0
\end{array}
$$
Now, the first inequality is strict unless either $t_1=t_2=\cdots=t_k$ or (up to reordering the trees $T_2,\ldots,T_k$) $k=2m\geq 4$, $t_1=\cdots=t_m$ and $t_{m+1}=\cdots=t_k$, and in both cases
$$
\Big|\sum\limits_{i=2}^kt_i+f(k-1)-t_1\Big|=\sum\limits_{i=3}^kt_i+f(k-1)>0
$$
Therefore, we always have that $\mathfrak{C}(T')-\mathfrak{C}(T)>0$.
\end{proof}

The proof of the following lemma is the same as that of Lemma \ref{lem:3mdm} (up to replacing the fractions $1/2$ by $1/\sqrt{2}$), and we shall not repeat it here.

\begin{lemma}\label{lem:3sd}
Let $f$ be a mapping $\NN\to \RR_{\geq 0}$ such that $f(2)>0$.
Consider the trees $T$ and $T'$ depicted in Fig. \ref{fig:lem3mdm}, where  $T'$ is obtained from $T$ by simply interchanging the subtrees $T_l$ and $T_{l-1}$. If $\delta(T_l)< \delta(T_{l-1})$ 
and $\delta(T_l)\leq \delta(T_{0})$,
then $\mathfrak{C}(T')> \mathfrak{C}(T)$. \qed
\end{lemma}

\begin{lemma}\label{lem:4sd}
Let $f$ be a mapping $\NN\to \RR_{\geq 0}$ such that $0<f(k)< f(k-1)+f(2)$, for every $k\geq 3$, and let $T$ be the tree depicted in Fig. \ref{fig:lem4mdm1}, where $l\geq 1$, $x_1$ is the root, all nodes in the path from $x_1$ to $x_l$ are binary,  and $k\geq 3$.
Assume moreover that $\delta(S_1)\leq \delta(S_2)\leq \cdots \leq \delta(S_l)$ and
that $\delta(T_1)$ is the value in the set $\{\delta(T_1),\ldots,\delta(T_k)\}$ closest to its mean.

\begin{enumerate}[(a)]
\item If $\delta(S_l)\leq \delta(T_1)$, then the tree $T'$ depicted in Fig. \ref{fig:lem4mdm}, obtained by pruning the subtree $T_1$ and regrafting it in the arc ending in $x$, is such that $\mathfrak{C}(T')> \mathfrak{C}(T)$.

\item If $\delta(S_l)> \delta(T_1)$, then the tree $T''$ depicted in Fig. \ref{fig:lem4mdm}, obtained by pruning the subtree $T_1$ and regrafting it in the arc ending in $x_l$, is such that $\mathfrak{C}(T'')> \mathfrak{C}(T)$.
\end{enumerate}
\end{lemma}

\begin{proof}
For every $i=1,\ldots,l$, let $s_i=\delta(S_i)$ and, for every $i=1,\ldots,k$, $\delta(T_i)=t_i$.
Let, moreover, $t=t_1+\cdots+t_k$.
So, we are assuming that $s_1\leq\cdots\leq s_l$ and that $\sd(t_1,\ldots,t_k)\leq \sd(t_2,\ldots,t_k)$.
Moreover, for every $i=1,\ldots,l$, we shall call $x_i$ the parent of the root of $S_i$ in all three trees $T,T',T''$.

As far as assertion (a) goes, the only nodes in $T$ or $T'$ with different ${bal}$ value in both trees are $x_1,\ldots,x_l$, $x$ and the new node $y$. Therefore:
$$
\begin{array}{l}
\displaystyle \mathfrak{C}(T')-\mathfrak{C}(T)  ={bal}_{T'}(x)-{bal}_{T}(x)+{bal}_{T'}(y)+\sum\limits _{i=1}^l\big({bal}_{T'}(x_i)-{bal}_{T}(x_i)\big)
\\[2ex]
\displaystyle\quad =\sd(t_2,\ldots,t_k)-\sd(t_1,\ldots,t_k)+\dfrac{1}{\sqrt{2}}\Big|\sum\limits_{i=2}^k t_i+f(k-1)-t_1\Big|\\[1ex]
\displaystyle\quad \quad +\sum_{i=1}^l\dfrac{1}{\sqrt{2}}\Big(\Big|t+\sum\limits_{j=i+1}^l s_j+f(k-1)+(l-i+1)f(2)-s_i\Big|-\Big|t+\sum\limits_{j=i+1}^l s_j+f(k)+(l-i)f(2)-s_i\Big|\Big)\\[2ex]
\displaystyle\quad \geq \dfrac{1}{\sqrt{2}}\sum\limits _{i=1}^l\Big(\Big|t+\sum\limits_{j=i+1}^l s_j+f(k-1)+(l-i+1)f(2)-s_i\Big|-\Big|t+\sum\limits_{j=i+1}^l s_j+f(k)+(l-i)f(2)-s_i\Big|\Big)\\[2ex]
\displaystyle\quad =\frac{l}{\sqrt{2}}(f(k-1)+f(2)-f(k))>0
\end{array}
$$
where, as in the proof of Lemma \ref{lem:4mdm}.(a1),
the  last equality is a consequence of the fact that, for every $i=1,\ldots,l$,
$s_i\leq s_l\leq t_1\leq t$.

Let us prove now assertion (b). Again, the only nodes in $T$ or $T''$ with different ${bal}$ value in both trees are $x_1,\ldots,x_l$, $x$ and the new node $y$. Therefore,
$$
\begin{array}{l}
\displaystyle \mathfrak{C}(T')-\mathfrak{C}(T)  ={bal}_{T''}(x)-{bal}_{T}(x)+{bal}_{T''}(x_l)+{bal}_{T''}(y)-{bal}_{T}(x_l) +\sum\limits _{i=1}^{l-1}({bal}_{T''}(x_i)-{bal}_{T}(x_i))
\\[2ex]
\displaystyle\quad =\sd(t_2,\ldots,t_k)-\sd(t_1,\ldots,t_k)+\dfrac{1}{\sqrt{2}}\Big|\sum\limits_{i=2}^k t_i+f(k-1)-s_l\Big|\\[1ex]
\displaystyle\quad\quad +\dfrac{1}{\sqrt{2}}\Big|\sum\limits_{i=2}^k t_i+s_l+f(k-1)+f(2)-t_1\Big| -\dfrac{1}{\sqrt{2}}|t+f(k)-s_l|\\[1ex]
\displaystyle\quad\quad+\dfrac{1}{\sqrt{2}}\sum\limits _{i=1}^{l-1}\Big( \Big|t+\sum\limits_{j=i+1}^{l} s_j+f(k-1)+(l-i+1)f(2)-s_i\Big|-\Big|t+\sum\limits_{j=i+1}^{l}+f(k)+(l-i)f(2)-s_i\Big|\Big)\\[2ex]
\displaystyle\quad \geq \dfrac{1}{\sqrt{2}}\Big(\Big|\sum\limits_{i=2}^k t_i+s_l+f(k-1)+f(2)-t_1\Big|-|t+f(k)-s_l|\Big)\\[1ex]
\displaystyle\quad \quad +\dfrac{1}{\sqrt{2}}\sum\limits _{i=1}^{l-1}\Big(\Big|t+ \sum_{j=i+1}^{l}  s_j+f(k-1)+(l-i+1)f(2)-s_i\Big|\!-\!\Big|t+ \sum_{j=i+1}^{l}  s_j+f(k)+(l-i)f(2)-s_i\Big|\Big)\!>\!0\end{array}
$$
where the last strict inequality is derived as in the proof of Lemma \ref{lem:4mdm}.(a2).
\end{proof}

Then, using Lemmas \ref{lem:2sd} to \ref{lem:4sd} and arguing as in the proof of Corollary \ref{cor:basic1mdm}, we deduce
that, for every non-binary  tree $T\in \TT^*_n$, there always exists a binary  tree $T'\in 
\TT^*_n$ such that $\mathfrak{C}(T')>\mathfrak{C}(T)$.  Starting from this fact, the same argument that completes the proof of Theorem 
18 for $\mathfrak{C}_{\MDM,f}$ also completes it for $\mathfrak{C}_{\sd,f}$.

\section{Proof of the thesis of Theorem 18 for $\mathfrak{C}_{\var,f}$}

\noindent The proof is similar to those described in the previous two sections, using Lemma \ref{lem:var}
and proving a series of lemmas that show how to increase the  Colless-like index of a non-binary  tree by making it ``more binary.''
To simplify the notations, in this section, we shall denote $\delta_f$ and $\mathfrak{C}_{\var,f}$ by simply $\delta$ and $\mathfrak{C}$, respectively, and we shall denote ${bal}_{\var,f}$ on a tree $T$ by ${bal}_T$ or simply by ${bal}$ when it is not necessary to specify the tree.

\begin{lemma}\label{lem:2}
Let $f$ be a mapping $\NN\to \RR_{\geq 0}$ such that $f(k)>0$, for every $k\geq 2$, and
let $T$ be a tree of the form $T_1\star\cdots\star T_k$, with $k\geq 3$.
If $\delta(T_1)$ is the value in the set $\{\delta(T_1),\ldots,\delta(T_k)\}$ closest to its mean, 
then taking  the tree $T'=T_1\star(T_2\star \cdots \star T_k)$  depicted in the left-hand side of Fig. \ref{fig:lem2mdm}, we obtain that $\mathfrak{C}(T')> \mathfrak{C}(T)$.
\end{lemma}

\begin{proof}
Let $t_i=\delta(T_i)$, for every $i=1,\ldots,k$.
The only nodes in $T$ or $T'$ with different ${bal}$ value in both trees are the roots and the new node $v$ in $T'$. Therefore,
$$
\begin{array}{l}
\displaystyle \mathfrak{C}(T')-\mathfrak{C}(T)  ={bal}_{T'}(v)+{bal}_{T'}(r)-{bal}_{T}(r)\\
\displaystyle \qquad =\var(t_2,\ldots,t_k)+\frac{1}{2}\Big(\sum\limits_{i=2}^kt_i+f(k-1)-t_1\Big)^2-\var(t_1,\ldots,t_k)
 \geq \frac{1}{2}\Big(\sum\limits_{i=2}^kt_i+f(k-1)-t_1\Big)^2\geq 0
\end{array}
$$
Now, the first inequality is strict unless either $t_1=t_2=\cdots=t_k$ or (up to reordering the trees $T_2,\ldots,T_k$) $k=2m\geq 4$, $t_1=\cdots=t_m$ and $t_{m+1}=\cdots=t_k$, and in both cases the last inequality is strict (cf. the proof of Lemma \ref{lem:2sd}).
\end{proof}

\begin{lemma}\label{lem:3}
Let $f$ be a mapping $\NN\to \RR_{\geq 0}$ such that $f(2)>0$.
Consider the trees $T$ and $T'$ depicted in Fig. \ref{fig:lem3mdm}, where, in  both trees, the nodes in the path connecting the root $r$ with $x$ are binary, and  $T'$ is obtained from $T$ by simply interchanging the subtrees $T_l$ and $T_{l-1}$. If $\delta(T_l)< \delta(T_{l-1})$, then $\mathfrak{C}(T')> \mathfrak{C}(T)$.
\end{lemma}

\begin{proof}
Let $t_i=\delta(T_i)$, for every $i=0,\ldots,l$, so that $t_l< t_{l-1}$.
The only nodes in $T$ or $T'$ with different ${bal}$ value in both trees are $x$ and $y$, and therefore,
$$
\begin{array}{l}
\displaystyle \mathfrak{C}(T')-\mathfrak{C}(T) ={bal}_{T'}(x)+{bal}_{T'}(y)-{bal}_{T}(x)-{bal}_{T}(y)\\[1ex]
\displaystyle\qquad =\frac{1}{2}(t_{l-1}-t_0)^2+\frac{1}{2}(t_{l-1}+t_0+f(2)-t_l)^2-\frac{1}{2}(t_{l}-t_0)^2-\frac{1}{2}(t_{l}+t_0+f(2)-t_{l-1})^2
\\[2ex]
\displaystyle\qquad =\frac{1}{2}(t_{l-1}-t_l)(t_{l-1}+t_l+2t_0+4f(2))> 0
\end{array}
$$
\end{proof}

\begin{lemma}\label{lem:4}
Let $f$ be a mapping $\NN\to \RR_{\geq 0}$ such that $0<f(k)< f(k-1)+f(2)$, for every $k\geq 3$, and let $T$ be the tree depicted in Fig. \ref{fig:lem4mdm1}, where $l\geq 1$, $x_1$ is the root, all nodes in the path from $x_1$ to $x_l$ are binary,  and $k\geq 3$.
Assume moreover that $\delta(S_1)\leq \delta(S_2)\leq \cdots \leq \delta(S_l)$ and
that $\delta(T_1)$ is the value in the set $\{\delta(T_1),\ldots,\delta(T_k)\}$ closest to its mean. Then:

\begin{enumerate}[(a)]
\item If $\delta(S_l)\leq \delta(T_1)$, then the tree $T'$ depicted in Fig. \ref{fig:lem4mdm}, obtained by pruning the subtree $T_1$ and regrafting it in the arc ending in $x$, is such that $\mathfrak{C}(T')> \mathfrak{C}(T)$.

\item If $\delta(S_l)> \delta(T_1)$, then the tree $T''$ depicted in Fig. \ref{fig:lem4mdm}, obtained by pruning the subtree $T_1$ and regrafting it in the arc ending in $x_l$, is such that $\mathfrak{C}(T'')> \mathfrak{C}(T)$.
\end{enumerate}
\end{lemma}

\begin{proof}
For every $i=1,\ldots,l$, let $s_i=\delta(S_i)$ and, for every $i=1,\ldots,k$, $\delta(T_i)=t_i$,
and let $t=t_1+\cdots+t_k$.
We are assuming that $s_1\leq\cdots\leq s_l$ and that $\var(t_1,\ldots,t_k)\leq \var(t_2,\ldots,t_k)$.
Moreover, for every $i=1,\ldots,l$, we shall call $x_i$ the parent of the root of $S_i$ in all three trees $T,T',T''$.

As far as assertion (a) goes, the only nodes in $T$ or $T'$ with different ${bal}$ value in both trees are $x_1,\ldots,x_l$, $x$ and the new node $y$. Therefore,
$$
\begin{array}{l}
\displaystyle\mathfrak{C}(T')-\mathfrak{C}(T)  ={bal}_{T'}(x)-{bal}_{T}(x)+{bal}_{T'}(y)+\sum\limits _{i=1}^l({bal}_{T'}(x_i)-{bal}_{T}(x_i))
\\[1ex]
\displaystyle\quad =\var(t_2,\ldots,t_k)-\var(t_1,\ldots,t_k)+\frac{1}{2}\Big(\sum\limits_{i=2}^k t_i+f(k-1)-t_1\Big)^2\\[1ex]
\displaystyle\quad \quad +\sum\limits _{i=1}^l\Big(\frac{1}{2}\Big(t+\sum\limits_{j=i+1}^l s_j+f(k-1)+(l-i+1)f(2)-s_i\Big)^2\hspace*{-1ex}-\frac{1}{2}\Big(t+\sum\limits_{j=i+1}^l s_j+f(k)+(l-i)f(2)-s_i\Big)^2\Big)\\[2ex]
\displaystyle\quad \geq \frac{1}{2}\sum\limits _{i=1}^l\Big(\Big(t+\sum\limits_{j=i+1}^l s_j+f(k-1)+(l-i+1)f(2)-s_i\Big)^2\hspace*{-1ex}-\Big(t+\sum\limits_{j=i+1}^l s_j+f(k)+(l-i)f(2)-s_i\Big)^2\Big)\\[2ex]
\displaystyle\quad = \frac{1}{2}\big(f(k-1)+f(2)-f(k)\big) \sum\limits _{i=1}^l \Big(2(t+\sum\limits_{j=i+1}^l s_j-s_i)+f(k-1)+f(k)+(2(l-i)+1)f(2)\Big)>0,
\end{array}
$$
where this last expression is $>0$ because $f(k-1)+f(2)-f(k)> 0$ and, for every $i=1,\ldots,l$,
$s_i\leq s_l\leq t_1\leq t$.

Let us prove now assertion (b). Again, the only nodes in $T$ or $T''$ with different ${bal}$ value in both trees are $x_1,\ldots,x_l$, $x$ and the new node $y$. Therefore, 
$$
\begin{array}{l}
\displaystyle\mathfrak{C}(T')-\mathfrak{C}(T)  ={bal}_{T''}(x)-{bal}_{T}(x)+{bal}_{T''}(y)+{bal}_{T''}(x_l)-{bal}_{T}(x_l) +\sum\limits _{i=1}^{l-1}({bal}_{T''}(x_i)-{bal}_{T}(x_i))
\\[1ex]
\displaystyle\quad =\var(t_2,\ldots,t_k)-\var(t_1,\ldots,t_k)+\frac{1}{2}\Big(\sum\limits_{i=2}^k t_i+s_l+f(k-1)+f(2)-t_1\Big)^2\\[1ex]
\displaystyle\qquad +\frac{1}{2}\Big(\sum\limits_{i=2}^k t_i+f(k-1)-s_l\Big)^2 -\frac{1}{2}(t+f(k)-s_l)^2\\[1ex]
\displaystyle\qquad+\frac{1}{2}\sum\limits _{i=1}^{l-1}\Big( \Big(t+\sum\limits_{j=i+1}^{l} s_j+f(k-1)+(l-i+1)f(2)-s_i\Big)^2-\Big(t+\sum\limits_{j=i+1}^{l}+f(k)+(l-i)f(2)-s_i\Big)^2\Big)\\[2ex]
\displaystyle\quad \geq \frac{1}{2}\Big(\Big(\sum\limits_{i=2}^k t_i+s_l+f(k-1)+f(2)-t_1\Big)^2-(t+f(k)-s_l)^2\Big)\\[1ex]
\displaystyle \qquad +\frac{1}{2}\sum\limits _{i=1}^{l-1}\Big(\Big(t+\sum\limits_{j=i+1}^{l} s_j+f(k-1)+(l-i+1)f(2)-s_i\Big)^2-\Big(t+\sum\limits_{j=i+1}^{l} s_j+f(k)+(l-i)f(2)-s_i\Big)^2\Big)\\[2ex]
\displaystyle\quad =\frac{1}{2}\Big(2\sum\limits_{i=2}^k t_i+f(k-1)+f(2)+f(k)\Big)(f(k-1)+f(2)-f(k)+2(s_l-t_1))\\[2ex]
\displaystyle\qquad +\frac{1}{2}(f(k-1)+f(2)-f(k)) \sum\limits _{i=1}^{l-1} \Big(2\Big(t+\sum\limits_{j=i+1}^{l} s_j-s_i\Big)+f(k-1)+f(k)+(2(l-i)+1)f(2)\Big)>0
\end{array}
$$
where this last expression is $> 0$ because $f(k-1)+f(2)-f(k)> 0$, $s_l> t_1$
 and, for every $i=1,\ldots,l-1$, $s_i\leq s_l$.
\end{proof}

Then, using Lemmas \ref{lem:2} to \ref{lem:4} and arguing as in the proof of Corollary \ref{cor:basic1mdm}, it can be proved
that, for every non-binary  tree $T\in \TT_n^*$, there always exists a binary  tree $T'\in \TT_n^*$ such that $\mathfrak{C}(T')>\mathfrak{C}(T)$. Therefore, the maximum $\mathfrak{C}$ value is reached at some binary tree. Since, for binary trees $T$, 
$\mathfrak{C}(T)=\frac{(f(0)+f(2))^2}{2}\cdot C^{(2)}(T)$ (see Proposition 7 in the main text), and $f(0)+f(2)>0$, it remains to prove that
 the binary  tree  in $\TT^*_n$ with maximum $C^{(2)}$ is exactly the comb. The proof of this fact follows closely that of
 Lemma 1 in the main text.

\begin{corollary}\label{cor:cat}
For every binary  tree   $T\in \TT^*_n$, if $T\neq K_n$, 
then $C^{(2)}(K_n)> C^{(2)}(T)$.
\end{corollary}

\begin{proof}
Using the argument of the proof of Lemma 1 in the main text, it is enough to prove that if $T$ and $T'$ are the  trees depicted in Fig.~6 in the main text, then, under the assumptions therein, 
$C^{(2)}(T')>C^{(2)}(T)$. And, indeed (using the notations therein),
$$
\begin{array}{l}
C^{(2)}(T')-C^{(2)}(T)  =(t_3+t_4-t_2)^2+(t_3+t_4+t_2-t_1)^2-(t_2-t_1)^2-(t_3+t_4-t_2-t_1)^2\\
\qquad =(t_3+t_4-t_1)(t_3+t_4+t_1+2t_2)>0
\end{array}
$$
where the last inequality holds because, by assumption, $t_1,t_2,t_3,t_4>0$ and $t_1+t_2\leq t_3+t_4$.
\end{proof}

Now, it is straightforward to check that
$$
C^{(2)}(K_n)=\sum_{k=1}^{n-2} k^2=\frac{1}{6}(n-1)(n-2)(2n-3),
$$
from where we obtain
$$
\mathfrak{C}(K_n)=\frac{(f(0)+f(2))^2}{2}\cdot C^{(2)}(K_n)=\frac{(f(0)+f(2))^2}{12}(n-1)(n-2)(2n-3),
$$
as we claimed in the statement.

\section{Proof of the thesis of Theorem 19   for $\mathfrak{C}_{\MDM,e^n}$}

\noindent To simplify the notations, we shall denote in this section $\delta_{e^n}$ and $\mathfrak{C}_{\MDM,e^n}$ by $\delta$ and $\mathfrak{C}$, respectively, and we shall denote ${bal}_{\MDM,f}$ on a tree $T$ by ${bal}_T$ or simply by ${bal}$ when it is not necessary to specify the tree.

\begin{lemma}\label{lem:prelexp1}
Let $n_1,\ldots,n_k\in \NN_{>0}$ and $n=n_1+\cdots+n_k$, and assume that $2\leq k\leq n-1$. Then
$$
e^{n_1}+\cdots+e^{n_k}+e^k\leq e^{n-k+1}+e^k+(k-1)e< e^n
$$
Moreover, the first inequality is strict unless there is at most one exponent $n_i\geq 2$.
\end{lemma}

\begin{proof}
To begin with, notice that if $1\leq x\leq \min(a,b)$, then
\begin{equation}
e^a+e^b\leq e^{a+b-x}+e^x,\label{eq:expab}
\end{equation}
 because
$$
e^b-e^x=(e^{b-x}-1)e^x\leq (e^{b-x}-1)e^a= e^{a+b-x}-e^a.
$$
Moreover, if  $x<\min(a,b)$ then the inequality is clearly strict.

Now, applying $k-1$ times inequality (\ref{eq:expab}) with $x=1$, we obtain
$$
e^{n_1}+\cdots+e^{n_k}  \leq e^{n_1+(n_2-1)+(n_3-1)+\cdots+(n_k-1)}+(k-1)e =
e^{n-k+1}+(k-1)e,
$$
which implies the first inequality. Moreover,  this inequality is strict unless all $n_i$ but, at most, one are 1, because
(\ref{eq:expab}) is strict if $x<a$ and $x<b$, which in this situation is translated to the existence of at least two $n_i,n_j$ greater than 1.

As far as the second inequality goes,
since $2\leq k\leq n-1$, and hence, in particular, $n\geq 3$ (and using, in the second inequality, that $x+1<e^x$ for every $x\in \RR_{>0}$), we have
$$
e^{n-k+1}+(k-1)e+e^k \leq 2e^{n-1}+(n-2)e<2e^{n-1}+e^{n-3}\cdot e=e^{n-2}(2e+1)<e^{n-2}\cdot e^2=e^n
$$
as we claimed. \end{proof}

\begin{lemma}\label{lem:prelexp2}
Let $n_1,\ldots,n_k,l\in \NN$ be such that $k\geq 1$, $k+l\geq 2$, each $n_i\geq 2$, and let $n=n_1+\cdots+n_k$. Then
$$
e^{n_1}+\cdots+e^{n_k}+e^{k+l}< e^{n+l}
$$
\end{lemma}

\begin{proof}
The case $l=0$ is a particular instance of the last lemma. So, we assume henceforth that $l\geq 1$.
If $k=1$, so that $n=n_1$, then applying inequality (\ref{eq:expab}) with $x=2$ we have
$$
e^{n}+e^{l+1}\leq e^{n+l-1}+e^2
$$
and the right hand side term is smaller than $e^{n+l}$ because $n+l\geq 3$.

Assume finally that $k\geq 2$. Applying $k-1$ times inequality (\ref{eq:expab}) with $x=2$, we obtain
$$
e^{n_1}+\cdots+e^{n_k}+e^{k+l}\leq
e^{n-2(k-1)}+(k-1)e^2+e^{k+l}
$$
Now, since $2\leq k\leq n/2$, and hence $n\geq  4$,
$$
\begin{array}{l}
e^{n-2(k-1)}+(k-1)e^2+e^{k+l}  \leq e^{n-2}+\big(\frac{n}{2}-1\big)e^2+e^{l+n/2}\\
\hspace*{2cm} < e^{n-2}+e^{n/2-2}\cdot e^2+e^{l+n/2}=e^{n-2}+e^{n/2}+e^{l+n/2} <3e^{n+l-2}<e^{n+l},
\end{array}
$$
as we claimed.
\end{proof}

\begin{lemma}\label{lem:exp1}
The largest $e^n$-size of a tree  in $\TT_n^*$ is $e^n+n$, and it is reached exactly at the star $\FS_n$.
\end{lemma}

\begin{proof}
The cases $n=1,2$ are obvious, because $\TT^*_n$ consists of a single tree.
Let now $n\geq 3$. We shall prove that for every $T\in \TT^*_n\setminus\{\FS_n\}$, there is a tree  $T'\in \TT^*_n$ with $\delta(T')>\delta(T)$. This shows that no tree other than $\FS_n$ can have the maximum $e^n$-size.

So, let $T=T_1\star\cdots \star T_m\in \TT_n^*\setminus\{\FS_n\}$, with $m\geq 2$. Let $l\geq 0$ be such that,  for every $i=1,\ldots,l$, 
the subtree $T_i$ consists of a single node, and, for every $i=l+1,\ldots,m$, $T_{i}=T_{i,1}\star\cdots\star T_{i,n_i}$ with $n_i\geq 2$; cf. Fig. \ref{fig:maxdelta}. Since $T\neq \FS_n$, $l<m$. Let now $T'$ be the tree 
$$
T'=T_{1}\star\cdots \star T_{l}\star T_{l+1,1}\star \cdots\star T_{l+1,n_{l+1}}\star \cdots \star T_{m,1}\star\cdots\star T_{m,n_m}.
$$
Then,
$$
\delta(T)  =e^{m}+l+\sum\limits_{i=l+1}^m \Big(e^{n_i}+\sum\limits_{j=1}^{n_i} \delta(T_{i,j})\Big)<
e^{l+n_{l+1}+\cdots+n_m}+l+\sum\limits_{i=l+1}^m \sum\limits_{j=1}^{n_i} \delta(T_{i,j})=\delta(T')
$$
by Lemma \ref{lem:prelexp2}.
\end{proof}

\begin{figure}[htb]
\begin{center}
\begin{tikzpicture}[thick,>=stealth,scale=0.3]

\draw(-1,-2) node[mintre] (z1) {}; 
\draw(-1,-3) node  {\footnotesize $T_1$};
\draw(-0.3,-2) node {.}; 
\draw(0,-2) node {.}; 
\draw(0.3,-2) node {.}; 
\draw(1,-2) node[mintre] (zl) {}; 
\draw(1,-3) node  {\footnotesize $T_l$};
\draw(4,0) node[mintre] (v11) {}; 
\draw (v11)--(3,-2)--(5,-2)--(v11);
\draw(4,-3) node  {\footnotesize $T_{l+1,1}$};
\draw(5.7,0) node {.}; 
\draw(6,0) node {.}; 
\draw(6.3,0) node {.}; 
\draw(8,0) node[mintre] (v1n) {}; 
\draw (v1n)--(7,-2)--(9,-2)--(v1n);
\draw(8,-3) node  {\footnotesize $T_{l+1,n_{l+1}}$};
\draw(10,2) node {.}; 
\draw(10.3,2) node {.}; 
\draw(10.6,2) node {.}; 
\draw(10.9,2) node {.}; 
\draw(11.2,2) node {.}; 
\draw(11.5,2) node {.}; 
\draw(15,0) node[mintre] (vm1) {}; 
\draw (vm1)--(14,-2)--(16,-2)--(vm1);
\draw(15,-3) node  {\footnotesize $T_{m,1}$};
\draw(16.7,0) node {.}; 
\draw(17,0) node {.}; 
\draw(17.3,0) node {.}; 
\draw(19,0) node[mintre] (vmn) {}; 
\draw (vmn)--(18,-2)--(20,-2)--(vmn);
\draw(19,-3) node  {\footnotesize $T_{m,n_{m}}$};
\draw(6,2) node[mintre] (w1) {}; 
\draw(17,2) node[mintre] (wm) {}; 
\draw(5,4) node[mintre] (r) {}; 
\draw(3,3.5) node  {\footnotesize $T$};
\draw (r)--(w1);
\draw (r)--(wm);
\draw (r)--(z1);
\draw (r)--(zl);
\draw (w1)--(v11);
\draw (w1)--(v1n);
\draw (wm)--(vm1);
\draw (wm)--(vmn);
\end{tikzpicture}
\qquad
\begin{tikzpicture}[thick,>=stealth,scale=0.3]
\draw(-1,-2) node[mintre] (z1) {}; 
\draw(-1,-3) node  {\footnotesize $T_1$};
\draw(-0.3,-2) node {.}; 
\draw(0,-2) node {.}; 
\draw(0.3,-2) node {.}; 
\draw(1,-2) node[mintre] (zl) {}; 
\draw(1,-3) node  {\footnotesize $T_l$};
\draw(4,0) node[mintre] (v11) {}; 
\draw (v11)--(3,-2)--(5,-2)--(v11);
\draw(4,-3) node  {\footnotesize $T_{l+1,1}$};
\draw(5.7,-2) node {.}; 
\draw(6,-2) node {.}; 
\draw(6.3,-2) node {.}; 
\draw(8,0) node[mintre] (v1n) {}; 
\draw (v1n)--(7,-2)--(9,-2)--(v1n);
\draw(8,-3) node  {\footnotesize $T_{l+1,n_{l+1}}$};
\draw(11,-2) node {.}; 
\draw(11.3,-2) node {.}; 
\draw(11.6,-2) node {.}; 
\draw(11.9,-2) node {.}; 
\draw(12.2,-2) node {.}; 
\draw(12.5,-2) node {.}; 
\draw(15,0) node[mintre] (vm1) {}; 
\draw (vm1)--(14,-2)--(16,-2)--(vm1);
\draw(15,-3) node  {\footnotesize $T_{m,1}$};
\draw(16.7,-2) node {.}; 
\draw(17,-2) node {.}; 
\draw(17.3,-2) node {.}; 
\draw(19,0) node[mintre] (vmn) {}; 
\draw (vmn)--(18,-2)--(20,-2)--(vmn);
\draw(19,-3) node  {\footnotesize $T_{m,n_{m}}$};
\draw(5,3) node[mintre] (r) {}; 
\draw(3,2.5) node  {\footnotesize $T'$};
\draw (r)--(z1);
\draw (r)--(zl);
\draw (r)--(v11);
\draw (r)--(v1n);
\draw (r)--(vm1);
\draw (r)--(vmn);
\end{tikzpicture}

\end{center}
\caption{\label{fig:maxdelta} The trees $T$ and $T'$ in the proof of Lemma \ref{lem:exp1}.}
\end{figure}
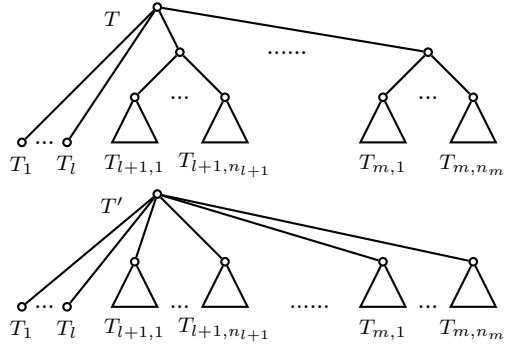

\begin{lemma}\label{lem:exp2}
For every $T\in\TT_n^*$ with $n\neq 1,3$, 
$$
2\mathfrak{C}(T)+\delta(T)\leq 2\mathfrak{C}(\FS_n)+\delta(\FS_n)= e^n+n,
$$
and the inequality is strict if $T\neq \FS_n$.

When $n=1$, the maximum of $2\mathfrak{C}+\delta$ is $2\mathfrak{C}(\FS_1)+\delta(\FS_1)=1$, and when
 $n=3$, this maximum is $2\mathfrak{C}(K_3)+\delta(K_3)=3e^2+4$ and it is reached exactly at $K_3$.
\end{lemma}

\begin{proof}
The cases $n=1,2$ are obvious, because then $\TT_n^*$ consists of a single tree, and the cases $n=3,4,5$ can be checked in the tables in supplementary file S2.
Notice that  $1<e^1+1$ and $3e^2+4< (e^3+3)+4$; we shall use these two inequalities below.
We prove now the general case $n\geq 6$  using the cases $n=1,\ldots,5$ and complete induction on $n$. 

Let $T=T_1\star \cdots \star T_k$, with $k\geq 2$ and $T_i\in \TT_{n_i}^*$ for every $i=1,\ldots,k$, so that $n=n_1+\cdots+n_k\geq 6$. If $k=n$, then $n_i=1$ for every $i$ and $T=\FS_n$, in which case $2\mathfrak{C}(T)+\delta(T)=e^n+n$. So, we shall assume that $k\leq n-1$, and we shall prove that, in this case
$2\mathfrak{C}(T)+\delta(T)< e^n+n$.

After renumbering the subtrees $T_i$ if necessary, assume that  there exists $l\geq 0$ such that $n_i=3$ for every $i\leq l$ 
and $n_i\neq 3$ for every $i=l+1,\ldots,k$. We shall prove first of all that
\begin{equation}
\MDM(\delta(T_1),\ldots,\delta(T_k))\leq \frac{1}{k}\sum\limits_{i=1}^k (e^{n_i}+n_i-2)
\label{ineq:mdm}
\end{equation}
Indeed, let $M=\mathit{Median}(\delta(T_1),\ldots,\delta(T_k))$. Then,
$$
\MDM(\delta(T_1),\ldots,\delta(T_k)) = \frac{1}{k}\sum\limits_{i=1}^k |\delta(T_i)-M|
\leq  \frac{1}{k} \sum\limits_{i=1}^k |\delta(T_i)-2|
\leq \frac{1}{k} \sum\limits_{i=1}^k (e^{n_i}+n_i-2)
$$
where the first inequality is due to the fact that $M$ is the real number that minimizes the function
$x\mapsto \sum\limits_{i=1}^k |\delta(T_i)-x|$, and the second inequality holds because $|\delta(T_i)-2|\leq e^{n_i}+n_i-2$ for every $i=1,\ldots,k$; on its turn, this inequality is a consequence, when $n_i\geq 2$,
 of Lemma \ref{lem:exp1}, and, when $n_i=1$, of the fact that if $T_1\in \TT_1^*$, then $\delta (T_i)=1$ and hence $|\delta(T_i)-2|=1<e^1+1-2$.
Now,
$$
\begin{array}{l}
\displaystyle 2\mathfrak{C}(T)+\delta(T)  =
2\Big(\sum\limits_{i=1}^k\mathfrak{C}(T_i)+\MDM(\delta(T_1),\ldots,\delta(T_k))\Big)+\sum\limits_{i=1}^k\delta(T_i)+e^k\\[1ex]
\displaystyle\qquad =
\sum\limits_{i=1}^k(2\mathfrak{C}(T_i)+\delta(T_i))+2\MDM(\delta(T_1),\ldots,\delta(T_k))+e^k\\[1ex]
\displaystyle\qquad =\sum\limits_{i=1}^l(2\mathfrak{C}(T_i)+\delta(T_i))
+\sum\limits_{i=l+1}^k(2\mathfrak{C}(T_i)+\delta(T_i)) +2\MDM(\delta(T_1),\ldots,\delta(T_k))+e^k\\[1ex]
\displaystyle\qquad \leq 
\sum\limits_{i=1}^l(3e^{2}+4)+\sum\limits_{i=l+1}^k(e^{n_i}+n_i)+\frac{2}{k}\sum\limits_{i=1}^k (e^{n_i}+n_i-2)+e^k\\[1ex]
\qquad\quad\mbox{(by the case $n=3$, the induction hypothesis, and inequality (\ref{ineq:mdm}))}\\[1ex]
\displaystyle\qquad \leq 
\sum\limits_{i=1}^l(e^3+3+4)+\sum\limits_{i=l+1}^k(e^{n_i}+n_i)+\sum\limits_{i=1}^k (e^{n_i}+n_i-2)+e^k\\[1ex]
\qquad\quad\mbox{(because $k\geq 2$ and $3e^2+4< e^3+3+4$)}\\
\displaystyle\qquad = 
\sum\limits_{i=1}^k(e^{n_i}+n_i)+4l+\sum\limits_{i=1}^k (e^{n_i}+n_i-2)+e^k
\\[1ex]
\displaystyle\qquad=2\sum\limits_{i=1}^k e^{n_i}+2n+4l-2k +e^k \leq 2\sum\limits_{i=1}^k e^{n_i}+e^k +2n+2k\\[1ex]
\displaystyle\qquad
 \leq 2 e^{n-(k-1)}+e^k+2(e+1)k+2n-2e\\
\qquad\quad\mbox{(by the first inequality in Lemma \ref{lem:prelexp1})}\\
\end{array}
$$
Thus, it remains to prove that, for every $n\geq 6$ and for every $2\leq k\leq n-1$,
\begin{equation}
2 e^{n-(k-1)}+e^k+2(e+1)k+2n-2e< e^n+n
\label{ineq:3}
\end{equation}
Since, for every $n\geq 6$, the function
$$
f_n(x)=e^n+n-\big(2 e^{n-(x-1)}+e^x+2(e+1)x+2n-2e\big)=e^n-2 e^{n-(x-1)}-e^x-2(e+1)x-n+2e
$$
is concave, because
$f_n''(x)<0$, its minimum value on the closed interval $[2,n-1]$ is reached at one of its ends. So,  in order to prove inequality (\ref{ineq:3}) for every $n\geq 6$ and for every $k=2,\ldots,n-1$, it is enough to prove that  $f_n(2)>0$ and $f_n(n-1)>0$ for every $n\geq 6$. And, indeed
\begin{itemize}
\item $f_n(2)=e^n-2 e^{n-1}-n-e^2-2e-4>0$ because the function
$g(x)=e^x-2 e^{x-1}-x-e^2-2e-4$ is increasing on $\RR_{\geq 2}$ and $g(5)>0$.

\item $f_n(n-1)=e^n-e^{n-1}-(2e+3)n-2e^2+4e+2>0$ by a similar reason.
\end{itemize}
This finishes the proof of the statement.
\end{proof}

Now we can proceed with the proof of Theorem  19.(a) for $D=\MDM$. The cases $n=2,3,4,5$ can be checked in the tables in the supplementary file S2. Notice in particular that, when $n=4$, the maximum is 
$$
\mathfrak{C}(K_4)=\frac{3}{2}(e^2+1)<\frac{1}{2}(e^{3}+2)+2;
$$
 we shall use this inequality below. We prove now, using the cases $n=1,\ldots,5$ and complete induction on $n$, that, for every $n\geq 6$,
\begin{quote}
\it The tree in $\TT_n^*$ with maximum $\mathfrak{C}$ is $\FS_1\star \FS_{n-1}$, with
$\mathfrak{C}(\FS_1\star \FS_{n-1})= \frac{1}{2}(e^{n-1}+n-2)$
\end{quote}
Recall  that, as in the previous sections,  Lemma \ref{lem:2mdm} implies that the maximum $\mathfrak{C}$ value on $\TT_n^*$ is reached at a tree with binary root. So, let $T=T_1\star T_2\in \TT_n^*$, with $T_1\in \TT_{n_1}^*$ and $T_2\in \TT_{n_2}^*$. We must distinguish two cases:
\medskip

\noindent\textit{a)} Assume that $n_1=1$, and therefore $n_2=n-1\geq 5$. In this case,
$$
\mathfrak{C}(T)=\mathfrak{C}(T_2)+\dfrac{1}{2}(\delta(T_2)-1)=
\dfrac{1}{2}(2\mathfrak{C}(T_2)+\delta(T_2)-1)\leq \dfrac{1}{2}(e^{n_2}+n_2-1)= \dfrac{1}{2}(e^{n-1}+n-2)
$$
by Lemma \ref{lem:exp2}. Moreover, the equality holds only when $T_2=\FS_{n-1}$.
\medskip

\noindent\textit{b)} Assume that $n_1,n_2\geq 2$ and, without any loss of generality, that $\delta(T_2)\leq \delta(T_1)$. Then,
$$
\begin{array}{rl}
\mathfrak{C}(T)&=\mathfrak{C}(T_1)+\mathfrak{C}(T_2)+\dfrac{1}{2}\big(\delta(T_1)-\delta(T_2)\big)\\[2ex]
&< \dfrac{1}{2}(e^{n_1-1}+n_1-2)+2+\dfrac{1}{2}(e^{n_2-1}+n_2-2)+2+\dfrac{1}{2}(e^{n_1}+n_1-n_2)=(*)\\
\end{array}
$$
This inequality is due to the following facts. On the one hand, $n_2\leq \delta(T_2)$ and $ \delta(T_1)\leq e^{n_1}+n_1$, by Lemma \ref{lem:exp1}, and hence $\delta(T_1)-\delta(T_2)\leq e^{n_1}+n_1-n_2$.
 On the other hand, by the induction hypothesis,
$\mathfrak{C}(T_i)\leq \frac{1}{2}(e^{n_i-1}+n_i-2)<\frac{1}{2}(e^{n_i-1}+n_i-2)+2$, unless $n_i=4$, in which case we still have
$\mathfrak{C}(T_i)\leq\mathfrak{C}(K_4) < \frac{1}{2}(e^{n_i-1}+n_i-2)+2$.

Let us continue
$$
(*) = \dfrac{1}{2}((1+e)e^{n_1-1}+e^{n_2-1}+2n_1+4)
\leq  \dfrac{1}{2}((2+e)e^{n-3}+2n)
$$
because $n_1,n_2\leq n-2$.
So, it remains to prove that, for every $n\geq 6$,
$$
(2+e)e^{n-3}+2n<e^{n-1}+n-2
$$
This is equivalent to
$$
(e^2-e-2)e^{n-3}-n-2>0,
$$
which is easy to prove, for instance noticing that 
$f(x)=(e^2-e-2)e^{x-3}-x-2$ is increasing on $\RR_{\geq 3}$ and that $f(5)>0$.
This finishes the proof of Theorem  19 for $D=\MDM$.

\section{Proof of the thesis of Theorem 19   for $\mathfrak{C}_{\sd,e^n}$}

\noindent  The proof of this case follows closely that of the case when $D=\MDM$ given in the last section. To begin with, it turns out that a key lemma similar to Lemma \ref{lem:exp2} also holds when $D=\sd$.
To simplify the notations, we shall  denote in this section $\delta_{e^n}$ and $\mathfrak{C}_{\sd,e^n}$ by $\delta$ and $\mathfrak{C}$, respectively, and we shall denote ${bal}_{\sd,f}$ on a tree $T$ by ${bal}_T$ or simply by ${bal}$ when it is not necessary to specify the tree.

\begin{lemma}\label{lem:exp2sd}
For every $T\in\TT_n^*$ with $n\neq 1,3$, 
$$
\sqrt{2}\cdot \mathfrak{C}(T)+\delta(T)\leq \sqrt{2}\cdot \mathfrak{C}(\FS_n)+\delta(\FS_n)= e^n+n.
$$
and the inequality is strict if $T\neq \FS_n$.

When $n=1$, the maximum of $\sqrt{2}\cdot \mathfrak{C}+\delta$ is $\sqrt{2}\cdot \mathfrak{C}(\FS_1)+\delta(\FS_1)=1+e$, and when
 $n=3$, this maximum is $\sqrt{2}\cdot \mathfrak{C}(K_3)+\delta(K_3)=3e^2+4$.
\end{lemma}

\begin{proof}
The cases $n=1,2$ are obvious, and the cases $n=3,4,5$ can be checked in Table 2 in the supplementary file S2.
We shall use that $1<e^1+1$ and the following inequalities:
\begin{subequations}
   \begin{align}
   &\sqrt{2}\cdot \mathfrak{C}(K_3)+\delta(K_3)=3e^2+4<(e^3+3)+4 \label{19.1}\\
    &\delta(K_3)=2e^2+3<(e^3+3)-5 \label{19.2}
\end{align}
\end{subequations}

%
%

We prove the general case $n\geq 6$ by induction on $n$ using an argument very similar to the one given in
the proof of Lemma \ref{lem:exp2}. Let $T=T_1\star \cdots \star T_k$, with $k\geq 2$ and $T_i\in \TT_{n_i}^*$ for every $i=1,\ldots,k$, so that $n=n_1+\cdots+n_k$. If $k=n$, then $n_i=1$ for every $i$ and $T=\FS_n$, in which case $\sqrt{2}\cdot \mathfrak{C}(T)+\delta(T)=e^n+n$. So, we shall assume that $k\leq n-1$.
Without any loss of generality, we assume that there exists $l\geq 0$ such that $T_i=K_3$ if, and only if, $i\leq l$.

Now, it turns out that
\begin{equation}
\displaystyle \sd(\delta(T_1),\ldots,\delta(T_k))\leq \frac{1}{\sqrt{k-1}}\Big(\sum\limits_{i=1}^k (e^{n_i}+n_i-2)-5l\Big)
\label{19.3}
\end{equation}
Indeed, let $m=(\delta(T_1)+\ldots+\delta(T_k))/k$. Then,
$$
\var(\delta(T_1),\ldots,\delta(T_k)) = \frac{1}{k-1}\sum\limits_{i=1}^k (\delta(T_i)-m)^2
\leq  \frac{1}{k-1} \sum\limits_{i=1}^k (\delta(T_i)-2)^2
$$
because $m$ is the real number that minimizes the function
$x\mapsto \sum\limits_{i=1}^k (\delta(T_i)-x)^2$. Taking square roots,
$$
\begin{array}{rl}
\displaystyle \sd(\delta(T_1),\ldots,\delta(T_k)) &\displaystyle \leq \sqrt{ \frac{1}{k-1} \sum\limits_{i=1}^k (\delta(T_i)-2)^2}
\leq \frac{1}{\sqrt{k-1}}\sum\limits_{i=1}^k |\delta(T_i)-2|\\[2ex]
&\displaystyle \leq \frac{1}{\sqrt{k-1}}\Big(\sum\limits_{i=1}^k (e^{n_i}+n_i-2)-5l\Big)
\end{array}
$$
where the second last inequality is a  consequence of Lemma \ref{lem:exp1}, inequality (\ref{19.2}), and the fact, already used in the previous section,  that if $T_1\in \TT_1^*$, then $|\delta(T_1)-2|=e-1= e+1-2$.
Then
$$
\begin{array}{l}
\displaystyle \sqrt{2}\mathfrak{C}(T)+\delta(T)  =
\sqrt{2}\Big(\sum\limits_{i=1}^k\mathfrak{C}(T_i)+\sd(\delta(T_1),\ldots,\delta(T_k))\Big)+\sum\limits_{i=1}^k\delta(T_i)+e^k\\[2ex]
\displaystyle\qquad =
\sum\limits_{i=1}^k(\sqrt{2}\mathfrak{C}(T_i)+\delta(T_i))+\sqrt{2}\sd(\delta(T_1),\ldots,\delta(T_k))+e^k\\[2ex]
\displaystyle\qquad \leq \sum\limits_{i=1}^k(e^{n_i}+n_i)+4 l+\frac{\sqrt{2}}{\sqrt{k-1}}\Big(\sum\limits_{i=1}^k (e^{n_i}+n_i-2)-5l\Big)+e^k\\[1ex]
\qquad\quad\mbox{(by the induction hypothesis and inequalities (\ref{19.1}) and (\ref{19.3}))}\\[2ex]
\displaystyle\qquad \leq \sum\limits_{i=1}^ke^{n_i}+n+4 l+\sqrt{2}\Big(\sum\limits_{i=1}^k e^{n_i}+n-2k-5l\Big)+e^k\\[1ex]
\displaystyle\qquad \leq (1+\sqrt{2})\sum\limits_{i=1}^k e^{n_i} +(1+\sqrt{2})n-2\sqrt{2}k+e^k\\[2ex]
\displaystyle\qquad \leq (1+\sqrt{2}) e^{n-(k-1)}+(1+\sqrt{2})(k-1)e+e^k+(1+\sqrt{2})n-2\sqrt{2}k\\[1ex]
\qquad\quad\mbox{(because of the first inequality in Lemma \ref{lem:prelexp1})}\\[1ex]
\displaystyle\qquad = (1+\sqrt{2}) e^{n-(k-1)}+(1+\sqrt{2})n+e^k+(e+\sqrt{2}e-2\sqrt{2})k -(1+\sqrt{2})e
\end{array}
$$
Thus, it remains to prove that, for every $n\geq 6$ and for every $2\leq k\leq n-1$,
\begin{equation}
(1+\sqrt{2}) e^{n-(k-1)}+(1+\sqrt{2})n+e^k+(e+\sqrt{2}e-2\sqrt{2})k-(1+\sqrt{2})e< e^n+n
\label{19.4}
\end{equation}

Now, for every $n\geq 1$,  the function
$$
\begin{array}{rl}
f_n(x) & =e^n+n-\big((1+\sqrt{2}) e^{n-(x-1)}+e^x+(1+\sqrt{2})n+(e+\sqrt{2}e-2\sqrt{2})x-(1+\sqrt{2})e\big)\\[1ex]
 & =e^n-(1+\sqrt{2}) e^{n-(x-1)}-e^x- \sqrt{2}n-(e+\sqrt{2}e-2\sqrt{2})x+(1+\sqrt{2})e
\end{array}
$$
 is concave, and therefore the minimum value of $f_n(x)$ on the closed interval $[2,n-1]$ will be reached at one of its ends. So,
 in order to prove inequality (\ref{19.4}) for every $n\geq 6$ and every $k=2,\ldots,n-1$, 
  it is enough to prove that $f_n(2)>0$ and $f_n(n-1)>0$ for every $n\geq 6$. And, indeed
\begin{itemize}
\item $f_n(2)=e^n-(1+\sqrt{2})e^{n-1}-\sqrt{2}n-(e^2+\sqrt{2}e+e-4\sqrt{2})>0$ because the function
$g(x)=e^x-(1+\sqrt{2})e^{x-1}-\sqrt{2}x-(e^2+\sqrt{2}e+e-4\sqrt{2})$ is increasing on $\RR_{\geq 3}$ and $g(5)>0$.

\item $f_n(n-1)=e^n-(1+\sqrt{2})e^2-e^{n-1}-\sqrt{2}n-(e+\sqrt{2}e-2\sqrt{2})(n-1)+(1+\sqrt{2})e$ by a similar reason.
\end{itemize}
This finishes the proof of the lemma.
\end{proof}

From here on, the proof  of Theorem  19 for $D=\sd$ proceeds as the one for $D=\MDM$ given in the previous section, using Lemma \ref{lem:exp2sd} instead of Lemma \ref{lem:exp2}; to ease the task of the reader we provide it.  The cases $n=2,3,4,5$ can be checked in Table 2 in the supplementary file S2. Notice in particular that, when $n=4$, the maximum $\mathfrak{C}$ value is 
\begin{equation}
\mathfrak{C}(K_4)=\frac{3}{\sqrt{2}}(e^2+1)<\frac{1}{\sqrt{2}}(e^{3}+2)+2.5 \label{19.5}
\end{equation}
 we shall use it below.

We prove now, using the cases $n=1,\ldots,5$ and complete induction on $n$, that, for every $n\geq 6$,
\begin{quote}
\it The tree in $\TT_n^*$ with maximum $\mathfrak{C}$ is $\FS_1\star \FS_{n-1}$, with
$\mathfrak{C}(\FS_1\star \FS_{n-1})= \frac{1}{\sqrt{2}}(e^{n-1}+n-2)$
\end{quote}
To begin with, notice that Lemma \ref{lem:2sd} implies that the maximum $\mathfrak{C}$ value on $\TT_n^*$ is reached at a tree with binary root.
So, let $T=T_1\star T_2\in \TT_n^*$, with $T_1\in \TT_{n_1}^*$ and $T_2\in \TT_{n_2}^*$. We must distinguish two cases:
\medskip

\noindent\textit{a)} Assume that $n_1=1$, and therefore $n_2=n-1\geq 5$. In this case,
$$
\mathfrak{C}(T)=\mathfrak{C}(T_2)+\dfrac{1}{\sqrt{2}}(\delta(T_2)-1)=
\dfrac{1}{\sqrt{2}}(\sqrt{2}\mathfrak{C}(T_2)+\delta(T_2)-1)\leq \dfrac{1}{\sqrt{2}}(e^{n-1}+n-1-1)= \dfrac{1}{\sqrt{2}}(e^{n-1}+n-2)
$$
by Lemma \ref{lem:exp2}. Moreover, the equality holds only when $T_2=\FS_{n-1}$.
\medskip

\noindent\textit{b)} Assume that $n_1,n_2\geq 2$ and, without any loss of generality, that $\delta(T_2)\leq \delta(T_1)$. Then,
$$
\begin{array}{rl}
\mathfrak{C}(T)&=\mathfrak{C}(T_1)+\mathfrak{C}(T_2)+\dfrac{1}{\sqrt{2}}\big(\delta(T_1)-\delta(T_2)\big)\\[2ex]
&< \dfrac{1}{\sqrt{2}}(e^{n_1-1}+n_1-2)+2.5+\dfrac{1}{\sqrt{2}}(e^{n_2-1}+n_2-2)+2.5+\dfrac{1}{\sqrt{2}}(e^{n_1}+n_1-n_2)=(*)\\
\end{array}
$$
This inequality is due to the following facts. On the one hand, $n_2\leq \delta(T_2)$ and $ \delta(T_1)\leq e^{n_1}+n_1$, by Lemma \ref{lem:exp1}, and hence $\delta(T_1)-\delta(T_2)\leq e^{n_1}+n_1-n_2$.
 On the other hand, by the induction hypothesis,
$\mathfrak{C}(T_i)\leq \frac{1}{\sqrt{2}}(e^{n_i-1}+n_i-2)<\frac{1}{\sqrt{2}}(e^{n_i-1}+n_i-2)+2.5$, unless $n_i=4$, in which case we still have
$\mathfrak{C}(T_i)\leq\mathfrak{C}(K_4) < \frac{1}{\sqrt{2}}(e^{n_i-1}+n_i-2)+2.5$ by inequality (\ref{19.5}).

Let us continue
$$
\begin{array}{rl}
(*) \hspace*{-2ex} &\displaystyle = \dfrac{1}{\sqrt{2}}((1+e)e^{n_1-1}+e^{n_2-1}+2n_1-4)+5
\leq  \dfrac{1}{\sqrt{2}}((2+e)e^{n-3}+2n+5\sqrt{2}-8)\\[1ex]
& \mbox{(because $n_1,n_2\leq n-2$)}\\[1ex]
&\displaystyle < \dfrac{1}{\sqrt{2}}((2+e)e^{n-3}+2n)<\frac{1}{\sqrt{2}}(e^{n-1}+n-2)
\end{array}
$$
because $(2+e)e^{n-3}+2n<e^{n-1}+n-2$, as it was proven in the last step of the proof of Theorem 19 for $D=\MDM$ in the last section.
This finishes the proof of Theorem  19  for $D=\sd$.

\section{Proof of the thesis of Theorem 19    for $\mathfrak{C}_{\var,e^n}$}

\noindent The stated maximum value of $\mathfrak{C}_{\var,e^n}$ on $\TT_n^*$, for $n=2,\ldots,5$, can be checked 
in Table 2 in the supplementary file S2. As far as the case when $n\geq 6$, it is a direct consequence of the corresponding result for $D=\sd$, established in the previous section.

Indeed, to begin with, notice that, since, for every node $v$ in a  tree $T$, ${bal}_{\var,f}(v)={bal}_{\sd,f}(v)^2$, we have that, for every  tree $T$,
$$
\mathfrak{C}_{\var,f}(T)=\sum_{v\in V_{int}(T)}{bal}_{\sd,f}(v)^2\leq \Big(\sum_{v\in V_{int}(T)}{bal}_{\sd,f}(v)\Big)^2=
\mathfrak{C}_{\sd,f}(T)^2
$$
So, for every $T\in \TT_n^*$ with $n\geq 6$,
$$
\begin{array}{rl}
\mathfrak{C}_{\var,e^n}(T) & \leq \mathfrak{C}_{\sd,e^n}(T)^2\leq \mathfrak{C}_{\sd,e^n}(\FS_1\star\FS_{n-1})^2
=\frac{1}{2}(e^{n-1}+n-2)^2\\[1ex]
& =\mathfrak{C}_{\var,e^n}(\FS_1\star\FS_{n-1})
\end{array}
$$
where the second inequality is strict if $T\neq  \FS_1\star\FS_{n-1}$. \newpage

\begin{center}
\Large Supplementary file S2: Some tables
\end{center}

\begin{table}
\begin{center}
\setlength{\tabcolsep}{3pt}
\begin{tabular}{c|c|c|c|c|}
Tree &  $\delta_{f}$ & $\mathfrak{C}_{\MDM,f}$ & $\mathfrak{C}_{\var,f}$& $\mathfrak{C}_{\sd,f}$ \\\hline
\begin{tikzpicture}[>=stealth,yscale=0.2,xscale=0.18]
\draw(0,0) node [mintre] (1) {};  
\draw(2,0) node [mintre] (2) {};   
\draw(1,2) node [mintre]  (r) {};
\draw(2,2.2) node (x) {};
\draw  (r)--(1);
\draw  (r)--(2);
\end{tikzpicture}
& 
\begin{tikzpicture}[>=stealth,yscale=0.2,xscale=0.15]
\draw(0,0) node{};  
\draw(0,1) node  {\footnotesize $x_2+2x_0$};
\draw(0,2.2) node (x) {};
\end{tikzpicture}
 & 
 \begin{tikzpicture}[>=stealth,yscale=0.2,xscale=0.15]
\draw(0,0) node{};  
\draw(0,1) node  {\footnotesize 0};
\draw(0,2.2) node (x) {};
\end{tikzpicture}
  & 
 \begin{tikzpicture}[>=stealth,yscale=0.2,xscale=0.15]
\draw(0,0) node{};  
\draw(0,1) node  {\footnotesize 0};
\draw(0,2.2) node (x) {};
\end{tikzpicture}
 & 
 \begin{tikzpicture}[>=stealth,yscale=0.2,xscale=0.15]
\draw(0,0) node{};  
\draw(0,1) node  {\footnotesize 0};
\draw(0,2.2) node (x) {};
\end{tikzpicture}

 
\\
\hline

\end{tabular}
\end{center}
\caption{Abstract values of $\delta_{f}$, $\mathfrak{C}_{\MDM,f}$, $\mathfrak{C}_{\var,f}$, and $\mathfrak{C}_{\sd,f}$
on $\TT^*_n$ for $n=2,3,4,5$. We denote $f(i)$ by $x_i$. 
\label{table:abs}}
\end{table}

\begin{table}
\begin{center}
\scriptsize \setlength{\tabcolsep}{1.5pt}
\begin{tabular}{c|c|c|c|c|c|c|c||c|c|c|c|c|c|c|}
Tree & $\delta_{\ln}$ &  $\mathfrak{C}_{\mathrm{M},{\ln}}$\vphantom{$\mathfrak{C}_{\delta_{{\ln}_{A}}}$}& Pos.&  $\mathfrak{C}_{\var,{\ln}}$& Pos. &  $\mathfrak{C}_{\sd,{\ln}}$& Pos. & $\delta_{e^n}$ &
$\mathfrak{C}_{\mathrm{M},{e^n}}$& Pos.&  $\mathfrak{C}_{\var,{e^n}}$& Pos. &  $\mathfrak{C}_{\sd,{e^n}}$& Pos.  \\
\hline

\end{tabular}
\end{center}
\caption{Numerical values (rounded to 4 decimal places) of $\delta_f$ (for $f(n)=\ln(n+e)$ and $f(n)=e^n$) 
and  of $\mathfrak{C}_{D,f}$ (for every combination of $D=\MDM$, $\var$ or $\sd$ and $f(n)=\ln(n+e)$ or $f(n)=e^n$) on $\TT^*_n$, for every $n=2,3,4,5$. To win some horizontal space, 
in the subscripts we have replaced $\MDM$ by simply M, and we have denoted
${\ln(n+e)}$ by ${\ln}$.
The columns labelled ``Pos.'' give the position of the tree  in its $\TT^*_n$ in increasing order of the 
Colless-like balance index corresponding to the column on its left. The rows are sorted, for each $n$, in increasing 
order of $\mathfrak{C}_{\MDM,{\ln}}$.
\label{table:concr}}
\end{table}

\end{document}